\documentclass[a4paper,11pt,reqno]{article}
\usepackage[margin=1.2in]{geometry}
\usepackage{amsmath}
\usepackage{amsfonts}
\usepackage{amsthm}
\usepackage{amssymb}
\usepackage{tikz}
\usepackage[all]{xy}
\usepackage{tikz-cd}
\usepackage[colorlinks]{hyperref}
\usepackage{enumerate}
\usepackage{mathtools}
\usepackage{setspace}
\usepackage{tocloft}

\definecolor{darkgreen}{rgb}{0,0.5,0}
\definecolor{darkblue}{rgb}{0,0.1,0.5}

\hypersetup{
    colorlinks=true, % make the links colored
    linkcolor=darkblue, % color TOC links in blue
    urlcolor=blue, % color URLs in red
    linktoc=all, % 'all' will create links for everything in the TOC
    citecolor=darkgreen
    }
\newcommand{\doi[2]}{\href{http://dx.doi.org/#1}{#2}}
\newtheorem{theorem}{Theorem}[section]
\newtheorem*{theorem*}{Theorem}
\newtheorem*{acknowledgement*}{Acknowledgement}

\newtheorem{conjecture}[theorem]{Conjecture}
\newtheorem{corollary}[theorem]{Corollary}
\newtheorem*{corollary*}{Corollary}

\newtheorem{example}[theorem]{Example}

\newtheorem{lemma}[theorem]{Lemma}
\newtheorem*{lemma*}{Lemma}

\newtheorem*{maintheorem*}{Main Theorem}

\newtheorem{question}[theorem]{Question}
\newtheorem{proposition}[theorem]{Proposition}

\newtheorem{remark}[theorem]{Remark}
\newtheorem*{remark*}{Remark}

\newcommand{\ad}{{\rm ad}}
\newcommand{\Rep}{{\rm Rep}}
\renewcommand{\dim}{{\rm dim}}

\newcommand{\cC}{\mathcal{C}}

\newcommand{\cM}{\mathcal{M}}

\newcommand{\V}{\mathcal{V}}
\newcommand{\cW}{\mathcal{W}}

\newcommand{\C}{\mathbb{C}}

\newcommand{\Z}{\mathbb{Z}}

\newcommand{\g}{\mathfrak{g}}
\renewcommand{\sl}{\mathfrak{sl}}

\newcommand{\SL}{\mathrm{SL}}
\newcommand{\ord}{\mathrm{ord}}

\newcommand{\Vir}{\mathrm{Vir}}
\newcommand{\sVir}{\mathrm{sVir}}
\newcommand{\osp}{\mathfrak{osp}}

\newcommand{\Y}{\mathcal{Y}}

\newcommand{\bpsi}{\bar{\psi}}

\renewcommand{\d}{\mathrm{d}}
\renewcommand{\i}{\mathrm{i}}
\newcommand{\tr}{\mathrm{tr}}

\newcommand{\rV}{\mathrm{V}} %universal affine VOA
\newcommand{\rW}{\mathrm{W}} %W-algebra
\newcommand{\bM}{\mathbb{M}} %affine Verma module
\newcommand{\bW}{\mathbb{W}} %affine Wakimoto module
\newcommand{\bV}{\mathbb{V}} %affine Irrep
\newcommand{\bL}{\mathbb{L}} %affine quotient...?
\newcommand{\LieL}{\mathcal{L}} %finite-dim Lie algbera modules

\renewcommand{\Re}{\mathfrak{Re}}

\newcommand{\CommentsForMe}[1]{$\quad$}

\usepackage{comment}
\usepackage{caption}
\usepackage{cancel}
\usepackage{mathtools}

\usepackage{smartdiagram}
\usepackage{relsize} %mathlarger{}
\makeatletter
\NewDocumentCommand{\smartdiagramx}{r[] m}{%
 \StrCut{#1}{:}\diagramtype\option
\IfStrEq{\diagramtype}{connected constellation diagram}{% true-conn const diagram
   \begin{tikzpicture}[every node/.style={align=center,let hypenation}]
   \foreach \smitem [count=\xi] in {#2}{\global\let\maxsmitem\xi}
   \pgfmathtruncatemacro\actualnumitem{\maxsmitem-1}
   \foreach \smitem [count=\xi] in {#2}{% 
	   \ifnumequal{\xi}{1}{ %true
	   \node[planet](planet){\smitem};  
	   }{%false
	   \pgfmathtruncatemacro{\xj}{\xi-1}
	   \pgfmathtruncatemacro{\angle}{360/\actualnumitem*\xj+45}	
	   \edef\col{\@nameuse{color@\xj}}
	   \node[satellite] (satellite\xj)
        at (\angle:\sm@core@distanceplanetsatellite) {\smitem };
	   }%  
   }%
   \foreach \smitem [count=\xi] in {#2}{%
      \ifnumgreater{\xi}{1}{ %true
	      \pgfmathtruncatemacro{\xj}{\xi-1}
	      \edef\col{\@nameuse{color@\xj}} 
	      \pgfmathtruncatemacro{\xk}{mod(\xj,\actualnumitem) +1}
	      \path[connection planet satellite,-] 
         (satellite\xj) edge[bend right] (satellite\xk);  
	   }{}
   }%
   \end{tikzpicture}  
   }{}%end-connected constellation diagram
}
\makeatother

\begin{document}
\title{Vertex algebras with big center \\
and a Kazhdan-Lusztig Correspondence}
\author{}
\date{}
\maketitle

\newcommand{\A}{\mathcal{A}}
\newcommand{\HA}{\mathrm{H}^{\kappa^*\!\!}\A}
\newcommand{\HHA}{\mathrm{H}^{\kappa}\mathrm{H}^{\kappa^*\!\!}\A}

\vspace*{-.5cm}

\hspace{1.5cm}
\begin{minipage}{.4\textwidth}
Boris L. Feigin \newline

Higher School of Economy, \newline 
Moscow, Russian Federation\newline

The Hebrew University \newline Jerusalem, Israel
\end{minipage}
\begin{minipage}{0.4\textwidth}
Simon D. Lentner \newline

Universität Hamburg,  \newline
Hamburg, Germany\newline
\vspace{1.1cm}
\end{minipage}
\footnote{
Corresponding author: 
Bundesstra{\ss}e 55,
20146 Hamburg,
simon.lentner@uni-hamburg.de
}

\hspace*{4cm} 
\vspace{1cm}

\begin{abstract}
We study the semiclassical limit $\kappa\to \infty$ of the generalized quantum Langlands kernel associated to a Lie algebra $\g$ and an integer level $p$. This vertex algebra acquires a big center, containing the ring of functions over the space of $\g$-connections. We conjecture that the fiber over the zero connection is the Feigin-Tipunin vertex algebra, whose category of representations should be equivalent to the small quantum group, and that the other fibres are precisely its twisted modules, and that the entire category of representations is related to the quantum group with a big center. In this sense we present a generalized Kazhdan-Lusztig conjecture, involving deformations by any $\g$-connection. We prove our conjectures in small cases $(\g,1)$ and $(\sl_2,2)$ by explicitly computing all vertex algebras and categories involved.

\end{abstract}

\newpage
\setcounter{tocdepth}{2}
{\small \tableofcontents}

\section{Introduction}
\subsection{Background}

A vertex algebra~$\V$ is, roughly speaking, a complex-analytic version of a commutative algebra, where the multiplication map depends on a formal complex variable
$$\Y:\;\V\otimes_\C \V\to \V[[z,z^{-1}]],$$
Here, $\V[[z,z^{-1}]]$ denotes the Laurent series in a formal variable $z$  with coefficients in $\V$. The axioms of a vertex operator algebra include a version of commutativity or locality, which relates $\Y(a,z_1)\Y(b,z_2)$ and $\Y(b,z_2)\Y(a,z_1)$ up to formal delta-functions on the singularities $z_1,\;z_2,\;z_1-z_2=0$. Another requirement is an action of the Virasoro algebra on~$\V$ compatible with conformal transformations of the variable~$z$. Standard mathematical textbooks on vertex algebras include \cite{Kac98,FBZ04}. Vertex algebras have a strong relation to analysis, algebraic geometry and mathematical physics, in particular conformal field theory. A vertex algebra has a notion of representation, and under suitable conditions the category of representations has the structure of a  braided tensor category \cite{HLZ}.

 A fundamental example is the vertex algebra $\rV^\kappa(\g)$ associated to an affine Lie algebra $\hat{\g}_\kappa$ at some level $\kappa$, see Section \ref{sec_twistAffine}. The conformal structure implies that certain analytic functions attached to the vertex algebra
 %\footnote{$n$-point correlation functions of primary fields} 
 solve the Knizhnik-Zamolochikov differential equation. The solutions are multivalued around the singularities at $z_i=z_j$, and the monodromy gives an action of the braid group on the space of solutions. The Drinfeld-Kohno theorem \cite{Koh88,Drin90}, see e.g. \cite{Kas95} Chp. 19, states that this braid group action can be described by the braiding of the quantum group $U_q(\g)$, a deformation of the universal enveloping algebra $U(\g)$ of a finite-dimensional semisimple Lie algebra $\g$, which has been defined by Drinfeld for this very purpose. The Kazhdan-Lusztig correspondence states that for generic $\kappa$ the braided tensor category of suitable\footnote{The Kazhdan-Lusztig category consists of representations of $\hat{\g}$ that are smooth (for any $v$ exists $n$ with $\hat{\g}_{\geq n}v=0$) and whose $L_0$-eigenspaces are integrable with respect to the action of the horizontal algebra $\g$} representations of $\rV^\kappa(\g)$  is equivalent to the category of weight modules of the quantum group $U_q(\g)$ at $q=e^{\frac{\pi\i}{m(\kappa+h^\vee)}}$ with $m$ the lacing number of $\g$ \cite{KL1,KL2,KL3,KL4} and \cite{Zh08, Hu17}. Note that as abelian categories, it is not difficult so see that both categories are equivalent to the category of integrable $\g$-modules, independently of $\kappa$, as long as  it is generic. 
 
 A closely related vertex algebra is the principal $W$-algebra $\rW^\kappa(\g)$ obtained from $\rV^\kappa(\g)$ by a quantum Hamiltonian reduction, here with respect to the principal nilpotent orbit, see \cite{FBZ04} Chapter 15. In the smallest case $\g=\sl_2$ the $\rW$-algebra is the Virasoro algebra itself, see Section \ref{sec_Vir}. For results on the representation theory of $W$-algebras, see \cite{Ara07}. A fundamental property of $W$-algebras is the Feigin-Frenkel duality \cite{FF92}, which states an isomorphism 
$\rW^\kappa(\g)\cong \rW^{\kappa^*}(\g^\vee)$, where 
 $\g^\vee$ is the Langlands dual and $\kappa,\kappa^ *$ are are pair of dual levels fulfilling $(\kappa+h^\vee)({\kappa^*}+(h^\vee)^\vee)m=1$, with $(h^\vee)^\vee$ the dual Coxeter number of the dual root system. This result and its limit $\kappa\to -h^\vee,\,\kappa^*\to \infty$ is of particular importance for the geometric Langlands conjecture, which again draws from conformal field theory, as discussed in Frenkel's lecture \cite{Fr05}. 
 
 A more recent example that has gathered considerable interest, is the Feigin-Tipunin algebra $\cW_p(\g)$. It is defined as a kernel of screenings and has  a $\g^\vee$-action with $\g^\vee$-invariant subalgebra $\rW^\kappa(\g)$;  this is proven in \cite{Sug21a} for simply-laced $\g$, otherwise it is conjectured. The conjectural logarithmic Kazhdan-Lusztig correspondence \cite{FGST06,FT10,AM14, Len21, Sug21a, Sug21b} states that $\cW_p(\g)$ has a finite non-semisimple category of representation, which is equivalent as a braided tensor category to the representations of the small quantum group $u_q(\g),\;q^{2p}=1$, a non-semisimple finite-dimensional Hopf algebra quotient of $U_q(\g)$ for $q$ a root of unity defined by Lusztig \cite{Lusz90}, more precisely a quasi-Hopf algebra variant $\tilde{u}_q(\sl_2)$ \cite{CGR20,GLO18,Ne21}. The smallest case  $\cW_p(\sl_2)$ has been studied previously under the name triplet algebra, and in this case the conjecture has been finally proven by the work of many authors \cite{FGST06, AM08, NT09, TW13, CGR20, MY20, CLR21, GN21}. For the respective statement for the singlet algebra and a more systematic treatment of this and other cases see \cite{CLR23}. The non-semisimplicity brings new challenges, see e.g. the surveys on logarithmic conformal field theory \cite{CR} and the role of the coend \cite{FS16}.  
 
 Back to our example $\rV^\kappa(\g)$, it is intriguing to study limits where the parameter $\kappa$ tends to degenerate values. This depends of course on choices, more precisely a choice of an integral form, which in most of our cases have been studied by  \cite{CL19} under the name deformable family. For a certain integral form of $\rV^\kappa(\g)$  the limit  $\kappa\to \infty$, which physically corresponds to the semiclassical limit, leads to a commutative algebra, which can be identified with the ring of functions on the space of regular $\g$-valued connections, as we review in Section \ref{sec_LimitAff}. The limit of critical level $\kappa\to -h^\vee$ acquires a big center, which is the ring of functions on the space of $\g$-opers, see \cite{FBZ04} Proposition 16.8.4. By \cite{Ara11} this center is equal to the center of the reduction $\rW^\kappa(\g)$ at critical level. By the Feigin-Frenkel duality mentioned above, this is equal to the center of the semiclassical limit of $\rW^\kappa(\g^\vee)$. With $\sl_2^\vee=\sl_2$ this means that the semiclassical limit of the Virasoro algebra  is the ring of functions on the space of quadratic differentials resp. Sturm-Liouville operators, and many properties of the associated differential equation are reflected in the representation theory. We discuss this in more detail in Section \ref{sec_LimitVir} and \ref{sec_SturmLiouville}. \\

 \subsection{Goals}
 
 The idea of the present article follows previous ideas in joint work of the first author in \cite{FF92} and in \cite{BBFLT13, BFL16, ACF22} and the line of work \cite{CG17, CGL20,CDGG21,CN22}: Suppose we have a deformable family of vertex algebras $\V^\kappa$, which in the semiclassical limit $\kappa\to \infty$ acquires a big central subalgebra $\mathcal{Z}(\V^\infty)$, which coincides with the ring of functions $\mathcal{O}(X)$ of some space $X$. The reader should have in mind for $X$ a Lie group. Strictly speaking, in our case $X$ will be the space of $\g$-connections (resp. $\g$-opers) on a formal punctured disc, up to regular coordinate transformations $G((t))$, with composition law(s) given by $\g$-connections on a $3$-punctured sphere with prescribed expansions around the punctures. In the case of regular singularities such connections can be classified in terms of their monodromy, which is an element in the Lie group $G$, and up to conjugation in the Borel  subgroup $B$. The existence and essentially uniqueness of a $\g$-connection on a $3$-punctured sphere with prescribed monodromies reproduces the group law in $G$ and $B$.
 
 For such a vertex algebra $\V^\infty$ with big central subalgebra $\mathcal{O}(X)$, the vertex algebra itself and its modules can be considered to be fiber bundles over $X$, where the fibre $\V^\infty|_x$ at a point $x\in X$ is given by a quotient $\V^\infty/(Z-Z(x))_{Z\in \mathcal{Z}(\V^\infty)}$. The fibre at zero $\V^\infty|_0$ (resp. the identity in the Lie group) is again a vertex algebra, and all other fibers are modules over it.  However, it is an important observation that for a vertex algebra with a large central subalgebra, in contrast to an algebra with a large central subalgebra,  the other fibres $\V^\infty|_x$ are no proper modules over $\V^\infty|_0$, but rather twisted modules, meaning that their vertex operator contains multivalued functions in the complex variable $z$, due to the nontrivial braiding present. Spoken more categorically, the fibres should be modules in a $X$-graded tensor category with crossed braiding, see \cite{ENOM09}.
 
 A different, more familiar occurrence of $G$-crossed tensor categories are orbi\-folds: Let $\cW$ be a vertex algebra with a finite group or Lie group $G$ acting as automorphisms. Then, again in general conjecturally, there is a $G$-crossed tensor category $\bigoplus_{g\in G}\cC_g$ with $\cC_0$ the usual braided tensor category of representations of $\cW$ and with $\cC_g$ the $\cC_0$-bimodule category of $g$-twisted representations, which are representations of the vertex algebra $\cW$ involving multivalued analytic functions with monodromy controlled by $g$, see the initial paragraphs of Section \ref{sec_twisted}. The typical main application is that all of these representations restrict to proper representations over the invariant subalgebra (or orbifold) $\cW^G$, and in good cases the representation category of $\cW^G$ consists precisely of twisted modules decomposed over this subalgebra. Categorically the decomposition is described by $G$-equivariantization and again returns a proper braided tensor category.

 Thirdly, a natural Hopf algebra with big center is the infinite-dimensional quantum group $U_q^{KPdC}(\g)$ of Kac-Procesi-DeConcini \cite{DKP92a}. This Hopf algebra has a central subalgebra isomorphic to the ring of functions on the big cell or Poisson dual $G^*$ resp. $(G^\vee)^*$, and the quotient is the small quantum group $u_q(\g)$. This is close, but not quite right for our purposes: We now restrict ourselves to the Hopf algebra quotient with big center $\mathcal{O}(B^\vee)$, whose category of representations fibres over $X=B^\vee$, which is contained in the Poisson dual as well. This is the abelian category for which we prove explicit correspondences. In fact, it should be a $B^\vee$-crossed braided tensor category and the equivariantization should be given by the mixed quantum group \cite{Gait21}, see Question \ref{quest_quantumgroup}.  If we want to see a full $G^\vee$-graded category, we can note that any element in $G^\vee$ is conjugate to an element in $B^\vee$, and accordingly there should be a larger version of this quantum group $U_q^G(\g)$ with central subalgebra $\mathcal{O}(G^\vee)$. Note that this matches the fact that $\g$-connections with regular singularities have  associated monodromy in $G$, but up to conjugation in~$B$.
 
\enlargethispage{1cm}
 
  The goal of this article is to present two related classes of examples $\V^\kappa$  for each $(\g,p)$, where we can study the correspondence between fibres in the limit, twisted representations of the zero-fibre, and corresponding representations of the quantum group with a big center. In these examples, the general conjecture is  that the limit $\V^\infty$ acquires as center the ring of functions over $\g$-connections resp. $\g$-opers (for $\sl_2$ Sturm-Liouville operators), and such that the zero fibre $\cW=\V^\infty|_0$ is equal to  the Feigin-Tipunin algebra $\cW_p(\g)$ \cite{CN22}. On the other hand, the logarithmic Kazhdan-Lusztig conjecture relates this representation theory to the small quantum group. Our ultimate goal would thus be four descriptions of the same tensor category, extending the logarithmic Kazhdan-Lusztig correspondence:

 %https://tex.stackexchange.com/questions/344586/how-do-i-make-these-two-diagrams-in-beamer
 \begin{center}
  \scalebox{0.67}{
\smartdiagramset{
planet text width=4.6cm,
satellite text width=4.3cm,
distance planet-satellite=5.5cm
}
%\smartdiagram[connected constellation diagram]
%\smartdiagramx[bubble diagram:45]
\smartdiagramx[connected constellation diagram]
{{\bf A tensor category} \\ 
~\\
\hspace{-0.7cm}with $G^\vee$-grading\mbox{,} $G^\vee$-action $\hspace{-0.7cm}$ \\ \hspace{-0.5cm} and crossed braiding $\hspace{-.5cm}$\\ 
 ~\\
 $\displaystyle \mathlarger{{\mathlarger{\hphantom{\cC_g}
\bigoplus_{g\in G^\vee}\cC_g}}}$,
modules over the \\ vertex algebra with big center \\
~\\ $\displaystyle\mathlarger{\mathlarger{\HA^{(p)}[\g\mbox{,}\infty]}}$,
modules over a \\ quantum group with big center \\ ~ \\ $\displaystyle\mathlarger{\mathlarger{U_q^G(\g)}}$,
$G^\vee$-crossed extension\\ of modules over the  \\  small quantum group \\ ~ \\ $\displaystyle \mathlarger{\mathlarger{\tilde{u}_q(\g)}}$,
$g$-twisted modules \\over the Feigin-Tipunin vertex algebra  \\ ~ \\$\displaystyle\mathlarger{\mathlarger{\cW_p(\g)}}$}
}
\end{center}
 In the specific cases $(\g,1)$, where the Feigin-Tipunin algebra is merely an affine Lie algebra, and $(\sl_2,2)$, where the Feigin-Tipunin algebra coincides with symplectic fermions, all of these questions are explicitly approachable: First, most work on the first and second deformable family is done in \cite{ACL19, ACF22} and \cite{CGL20} respectively, and we have explicit models for $\V^\kappa$. In particular, we can explicitly determine that the zero-fibre is indeed  $\cW_1(\g)$ and $\cW_2(\sl_2)$ as conjectured and determine the fibres $\cC_g$ as abelian categories. Second, using the explicit models for the Feigin-Tipunin algebra, we can completely determine the categories of $g$-twisted module. Third, as discussed the Kazhdan-Lusztig conjecture is proven for $(\sl_2,p)$, and by construction it is true for $(\g,1)$. In the case $(\sl_2,p)$ and restricted to $H\subset G$ the Cartan subgroup, the twisted modules are familiar and the above correspondences follow from the correspondence between the singlet algebra and the unrolled quantum group \cite{CLR23}, although the structure as a $H$-crossed extension is not worked out in this case either. 
 
 \enlargethispage{1cm}
\begin{center}
\hspace{-0.8cm} \begin{tikzcd}[scale=0.4, color=blue]
    && 
    1^\pm
    \ar[dl, swap, "\psi_0", shift left=-0.5ex, bend right=20] \\
    & 
    1^\mp
    \ar[ur, swap, "\bpsi_0", shift right=-0.5ex, bend right=20] 
\end{tikzcd} 
\hspace{0.6cm}
\begin{tikzcd}[scale=0.4]
    & 
    v^\pm 
    \ar[dl,"\psi_0", shift right=-0.5ex]
    \ar[dr,swap,"\bpsi_0", shift left=-0.5ex] & \\
    v^\mp
    \ar[dr, "\bpsi_0", shift right=-0.5ex] 
    %\ar[ur, "\bpsi_1", shift right=-0.5ex] 
    && 
    1^\pm 
    \ar[dl, swap, "\bpsi_0", shift left=-0.5ex] \\
    %ar[ul, swap, "\psi_1", shift left=-0.5ex] \\
    &
    1^\pm
    %\ar[ul, "\psi_1", shift right=-0.5ex]
    %\ar[ur, swap, "\bpsi_1", shift left=-0.5ex]
\end{tikzcd}  
 \begin{tikzcd}[scale=0.4, color=darkgreen]
    & 
    v^\pm 
    \ar[dl,"\psi_0", shift right=-0.5ex, bend right=20]
    \ar[dr,swap,"\bpsi_0", shift left=-0.5ex] & \\
    v^\mp
    \ar[dr, "\bpsi_0", shift right=-0.5ex] 
    \ar[ur, shift right=+1.5ex, bend right=20] 
    %\ar[ur, "\psi_0", shift right=-0.5ex] 
    && 
    1^\pm
    \ar[dl, "\psi_0", shift left=-0.5ex, bend right=20] \\
    %\ar[dl, swap, "\psi_0", shift left=-0.5ex] \\
    %\ar[ul, swap, "\psi_1", shift left=-0.5ex] \\
    & 
    1^\mp
    %\ar[ul, "\psi_1", shift right=-0.5ex]
    \ar[ur, shift right=+1.5ex, bend right=20] 
    %\ar[ur, swap, "\psi_0", shift left=-0.5ex]
\end{tikzcd}  
\includegraphics[scale=1.15]{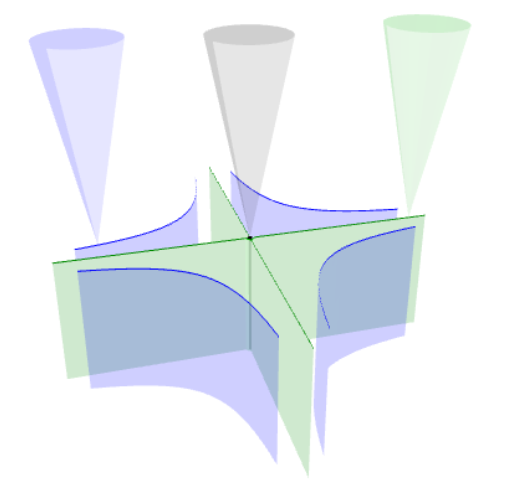}
    %$\hspace{-1cm}\substack{\mathrm{SL}_2(\C)\\\vspace{0cm}}$
\end{center}

\captionof{figure}{Case $(\sl_2,2)$: We show the projective objects in three category fibres $\cC_g$ over elements $\textcolor{blue}{\begin{psmallmatrix} a & 0 \\0 & a^{-1}\end{psmallmatrix}},\;\;\begin{psmallmatrix} 1 & 0 \\0 & 1\end{psmallmatrix},\;\;\textcolor{darkgreen}{\begin{psmallmatrix} 1 & t \\0 & 1\end{psmallmatrix}}$ representing the three types of conjugacy classes in~$\mathrm{SL}_2(\C)$.}\label{fig_fibre}

~\\
 
Let  us finally mention some further context: The role of the deformable family $\V^\kappa$ in this article is played by  $\A^{(n)}[\g,\kappa]$, which is an extension of $\hat{\g}_\kappa\times \hat{\g}_{\kappa^*}$ for $\frac{1}{\kappa+h^\vee}+\frac{1}{{\kappa^*}+h}=n,\;m|n$,
 and by their quantum Hamiltonian reductions in either or both factors, see Section~\ref{sec_couple}. These vertex algebras generalize  the algebra of chiral differential operators on a Lie group $G$ \cite{AG02}, they  play an important role for the geometric quantum Langlands correspondence and are called quantum geometric Langlands kernel. The geometric Langlands duality is recovered in the limit $\kappa\to \infty$ and for $n=0$. There is a 4D picture emerging from works by Kapustin and Witten \cite{KW07}, Creutzig and Gaiotto \cite{CG17}, Frenkel and Gaiotto \cite{FG18},  Costello, Dimofte, Linshaw and many others, for example \cite{CL19, CDG20, CL22, CGL20, CDGG21}. Here, the vertex algebras $\A^{(n)}[\g,\kappa]$ appear as 2D junctions between 3D boundary conditions in a 4D topological twist of a gauge theory attached to $G,\kappa$. They label the 2D junction between 3D boundary conditions $D_{p',q'}$ and $D_{p,q}$ related by a modular transformation $\begin{psmallmatrix} p'\\ q'\end{psmallmatrix}=ST^nS\begin{psmallmatrix} p\vphantom{'}\\ q\vphantom{'}\end{psmallmatrix}$, and their quantum Hamiltonian reductions in either factor label junctions between $N_{p,q}$ or $N_{p',q'}$. Here, the 3D boundary condition $D_{0,1}$ is labeled by the spin-ribbon category $D_{0,1}=\mathrm{Rep}(\mathcal{D}_\kappa)(\mathrm{Gr}_G)$ of twisted $\mathcal{D}$-modules over the affine Grassmannian and $N_{0,1}$ by the Kazhdan-Lusztig category associated to $G,\kappa$. For $\g=\sl_2,n=2$ we review in  Section \ref{sec_explicit_p2} from \cite{CGL20} the explicit realizations of $\A^{(2)}[\sl_2,\kappa]$ as $\sVir_{N=4}$ and of its one- and two-sided reductions as $\hat{\osp}(2|1)$ and $\sVir_{N=1}\otimes \mathcal{F}$.  The article \cite{CLNS22} shows how these vertex algebras appear as ''kernel VOA'' in certain dualities generalizing Feigin-Frenkel duality.\\\\

 \noindent
 The article has two parts, according to the two sides of the correspondence:

\subsection{Content Part 1}

The first part, which might be of independent interest, addresses the following question:\\

\begin{question}
An affine Lie algebra $\hat{\g}_\kappa$ and the associated vertex algebra  $\rV^\kappa(\g)$ have an action by the corresponding Lie group $G$. Similarly, the Feigin-Tipunin algebra $\cW_{p}(\g)$ has presumably an action by the corresponding dual Lie group $G^\vee$. What is the corresponding $G^\vee$-graded category of twisted modules $\bigoplus_{g\in G}\cC_g$, and how do the twisted modules and intertwining operators look explicitly? How do they decompose over the vertex subalgebra $\cW^G$? What are the relevant quantum group counterparts?
\end{question} 

For the definition of $g$-twisted modules with respect to non-semisimple actions (and accordingly logarithmic singularities) we refer to the works of Huang \cite{Hu10} and Bakalov \cite{Bak16, BS16}. For most statements it seems more conceptual to talk about $x$-twisted modules for an infinitesimal automorphism (resp. derivation) $x$. Although we do not need it in our explicit treatment, we want to point out that we seem to miss a sufficient treatment of the categorical notions in the case of $G$ a Lie group. A notion of sheaves of categories has been proposed in \cite{Gait13}, a notion of an action of an algebraic group on a tensor category is found in \cite{DEN17}, and for vertex algebra orbifolds for compact Lie groups see \cite{Mc20},

\begin{question}\label{quest_ENO}
    For $G$ an \emph{algebraic} group acting on the Drinfeld center of a tensor category $\cC_0$ resp. a modular tensor category $\cC_0$, is there a version of the results in \cite{ENOM09} that describe all $G$-graded resp. $G$-crossed extensions in terms of cohomological data? Such results would be invaluable for proving notoriously hard tensor category equivalences. \\
        
    We would also be very interested in a corresponding ''tangential'' categorical notion in the algebraic geometry setting, that is, a notion of a $\g$-crossed category, which should describe first-order deformation of a $G$-crossed category around the identity (for example by $\g\to \mathrm{Ext}^1(1,L)$ for $L$ the coend, together with structural data data along the lines of \cite{FGS22}). Accordingly, there should be an ENOM-correspondence between an action of $G$ and thus $\g$ on a Drinfeld center of $\cC_0$ resp. on $\cC_0$ and corresponding $\g$-crossed extensions. We would expect that the hard problem of determining explicitly the crossed extension from the action is much easier in the tangential setting. An interesting explicit example we have worked on is the $\mathrm{SO}_2(\mathbb{R})$-crossed extensions of $\mathrm{Vect}_{\mathbb{R}^2}$.
\end{question}
\CommentsForMe{The regular singularity should be higher categorical data}
\CommentsForMe{Maybe I see cases where there is a irreg singularity if $\cC_0$ has a singularity, as for the SO3-example if the category is wrongfully fixed to be vect}
\CommentsForMe{Maybe the irreg singularty is described by defects}

\noindent
We now discuss our concrete results: \\

In Section \ref{sec_twistAffine} we discuss the case of an affine Lie algebra, which also coincides with the Feigin-Tipunin algebra $\cW_1(\g)$ in the trivial case $p=1$. The explicit $g$-twisted mode algebra was in this case defined in \cite{BS16, Yang17}, so one can explicitly determine the vertex algebra modules. It seems a common occurrence that, because $g$ is an inner automorphisms, there exists an isomorphism between the $g$-twisted mode algebra and the untwisted mode algebra, so as abelian categories all $\cC_g$ are equivalent. However, the action of the Virasoro algebra is deformed by this isomorphism and the structure of the twisted module as a Virasoro module changes drastically. While a semisimple $g$ shifts the conformal weight by some non-integer, which is the  familiar case, we find for a unipotent  $g$ infinite towers of extensions and $L_0$ non-diagonalizable.
\begin{lemma*}[\ref{lm_twistedrepaffine}]
Consider the inner automorphism $g=e^{-2\pi\i y}$ of $\g$ with $y=s+x\in\g$ for $s$ semisimple and $x$ nilpotent. The image of  the twisted Virasoro action is as follows 
$$f(L_n^{\hat{\g}_{\kappa,g}})=L_n^{\hat{\g}_\kappa}
-\frac{2\kappa}{2(\kappa+h^\vee)}y_n+\delta_{n,0}\frac{\kappa^2}{2(\kappa+h^\vee)}\langle y,y\rangle$$
\end{lemma*}

In Section \ref{sec_twistedTriplet} we turn to the Feigin-Tipunin algebra $\cW_p(\g)$ and explicitly treat the case $(\g,p)=(\sl_2,2)$, where it is equal to the even part of the symplectic fermions. The twisted mode algebra for $g=e^{-2\pi\i y}$ and $y=\begin{psmallmatrix} r & 0 \\2t & -r\end{psmallmatrix}$ can be calculated explicitly, following  \cite{BS16}:
\begin{corollary*}[\ref{cor_twistedSF}]
The algebra of twisted mode operators is generated by $\psi_{r+n},\bpsi_{-r+n}$ for $n\in \Z$  with the following relations
\begin{align*}
\{ \psi_{r+m}^g, \psi_{r+n}^g\}&= 2t\delta_{m,-n} \\
\{ \psi_{r+m}^g, \bpsi_{-r+n}^g\}&= (r+m)\delta_{m,-n} \\
\{ \bar{\psi}_{-r+m}^g, \bpsi_{-r+n}^g\}&=0 
\end{align*} 
\end{corollary*}
From this  we deduce the twisted representations of the symplectic fermions and of its even part, and then again compute the deformed Virasoro structure. 

We find that for generic $g$  there are two simple projective modules, comparable to the lattice vertex algebra modules (see below). The semisimple part of $g$ causes as expected a shift in conformal weight. The unipotent part of $g$ causes a screening operator to appear in the Virasoro action, which is not anymore compatible with the horizontal (i.e. lattice) grading. Instead of a direct sum of Fock/Wakimoto modules we get a massive tower of extensions and $L_0$-Jordan blocks left-to-right, which piece together irreducible modules in different Fock/Wakimoto  modules to Verma modules and Coverma modules of the Viasoro algebra.  The projective modules look similar to the triplet algebra modules (which corresponds to $g=1$), but in general only has two composition factors, matching Figure~\ref{fig_fibre}.

The $g$-twisted modules still possess an action of the centralizer of $g$ in $G$, and this is relevant if one finally wishes to decompose the $g$-twisted module as ordinary modules over $\cW^G$, which is the Virasoro algebra. For $g$ semisimple, which is the well-known case, this decomposes the direct sum of Virasoro modules with respect to the Cartan part into the single Virasoro modules with generic weights in some coset. In the unipotent case we prove in Theorem \ref{thm_tilting} that the resulting Virasoro module has filtrations in terms of Verma and Coverma modules given by kernel and cokernel of the long screening operator, so our unipotently twisted modules give a free field realization of tilting modules. On the other hand, the findings seem contrary to the expectation we have for semisimple finite $G$, where the decomposition of $g$-twisted modules each produces new irreducible $\cW^G$-modules.
\CommentsForMe{What about twisted Virasoro in the generic case?}
\\

In Section \ref{sec_quantumgroups} we match in the case  $(\sl_2,2)$ our findings about the structure of the abelian categories $\cC_g$ for the three types of conjugacy classes $g\in \mathrm{SL}_2(\C)$ with the infinite-dimensional quantum group of Kac-Procesi-DeConcini \cite{DKP92a}, more precisely the fibres over a fixed Borel subgroup. For the principal block of $\tilde{u}_i(\sl_2)$ and nilpotent $g$, the category equivalence is even directly visible as isomorphism to the zero-mode algebra generated by $\psi_0^g,\bpsi_0^g$ and the parity operator~$\pi$.

\begin{question}\label{quest_quantumgroup}
One should construct a $B$-crossed extension of the braided tensor category of representations of the small quantum group $u_q(\g)$, respectively for $q$ of even order a $B^\vee$-crossed extension of the category of representations of the quasi-Hopf algebra $\tilde{u}_q(\g)$. Their equivariantization should correspond to the mixed quantum group \cite{Gait21}, and starting with this in principle gives a construction of the crossed extension.\\

Moreover, we expect that there is a quantum group $U_q^G(\g)$ with central subalgebra the ring of functions on $G$ resp. $G^\vee$, whose category of representations is a respective crossed category, and such that the restriction on any Borel subgroup recovers the extension in the previous paragraph. 
\end{question}

%\begin{question}
%We expect for general $(\g,p)$ an equivalence of $G$-crossed tensor categories between fibres of the vertex algebra with big center, the categories of twisted modules, and the fibres of the quantum group with big center in some version. \\
%If the ENOM uniquess of $G$-crossed extensions were proven in our setting (Question \ref{quest_ENO}), then this would follow from the logarithmic Kazhdan-Lusztig conjecture and Conjecture~\ref{conj_CN} concerning the zero-fibre (both of which we have in the example $(\sl_2,2)$), together with some mild calculations on the Cartan part to fix the choice in  $\mathrm{H}^3(G,\C^\times)$ for the associator. 
%\end{question}

In Section \ref{sec_twistedFreeField} we  present an approach to construct $g$-twisted representation for general $\cW_p(\g)$, which we would call twisted free field realization. The main idea is that, since any  $g$ is conjugate to an element in the standard Borel subgroups, we have for any  $g$ a free-field realization of $\cW_p(\g)$ as subalgebra of a lattice vertex algebra (as kernel of short screenings), such that $g$ is also an automorphism of the lattice vertex algebra, which is now an inner automorphism. For $g$ unipotent, it is given by (exponentiated) long screening and all mode operators commute. We may hence apply the $\Delta$-deformation in \cite{Li97, AM09} to long screening operators to obtain $g$-twisted modules. These modules can now be restricted to obtain twisted modules of $\cW_p(\g)$. In the case $\cW_2(\sl_2)$ we can explicitly observe, that at least the simple $g$-twisted modules can be extended to a module over the lattice vertex algebra  (due to the trivial action of the opposite screening), and this procedure can be easily generalized. Conversely the twisted modules arise from decomposition. In general, we can only conjecture, that this produces all simple twisted modules. Technology to address this question has been developed in \cite{Yang17}.

\begin{question}\label{quest_Liouville}
The appearance of short screening operators $\psi_0^g$ with nonzero-eigenvalue  in the $g$-twisted module of $\cW_p(\sl_2)$, and probably similarly for general $\cW_p(\g)$, seems to suggests a connection to the representations of the $\cW_{-p}(\g)$ \cite{FF96}, where the lattice is negative definite. This negative definiteness causes the $p$-th power of the short screening to be nonzero, and proportional to the long screening, which can act on a module with arbitrary eigenvalue, see  \cite{Len21} Section 6.4. 
\end{question}

\subsection{Content Part 2}

In the second part we turn our attention to the vertex algebras with big center.\\

In Section \ref{sec_LimitAff} we briefly review how the semiclassical limit $\kappa\to \infty$ of an affine Lie algebra $\hat{\g}_\kappa$ is identified with the ring of functions on the space of $\g$-valued connections, see \cite{FBZ04} Section 16.\\

In Section \ref{sec_Vir} we briefly review the Virasoro algebra and its representations at generic central charge $\V_h$ and $\varphi_{m,n}$. We also recall the formula from \cite{BSA88} for the singular vectors of $\varphi_{1,n}$.\\

In Section \ref{sec_LimitVir} we discuss how the semiclassical limit of the Virasoro algebra can be identified with the ring of functions on the space of Sturm-Liouville operators $\frac{\d^2}{\d z^2}+q(z)$, and how for regular Sturm-Liouville operators an additional action of $\sl_2$ is visible. We also propose an interpretation of the limits of $\varphi_{1,n}$ as regular singular Sturm-Liouville operators with singular term $q(z)=-\frac{n^2-1}{4}z^{-2}+\cdots$ and with diagonalizable monodromy,~and of  $\bigoplus_m \varphi_{m,n}$ as ring of functions on such Sturm-Liouville operators with a fixed solution.\\
\CommentsForMe{source?}  

In Section \ref{sec_SturmLiouville} we briefly review some aspects of Sturm-Liouville differential equations, and discuss some examples. We are in particular interested in the monodromy of regular singular solution and in the nontrivial case of Fuchs' theorem: If the leading exponents $z^{s'},z^{s''}$ of the two solutions have integer difference $s'-s''=n$, which precisely correspond to the singular term above, then the monodromy can be a Jordan block, with one solution given by Frobenius' method and another solution of \emph{second kind} involving logarithms. For example, the Bessel functions for integer parameter $\alpha$ exhibit such behavior, while the Bessel functions for half-integer parameter $\alpha$ still has diagonalizable monodromy. In both cases $s'-s''$ is an integer, but the singular vectors in $\varphi_{1,2\alpha}$ predicts the diagonalizability. \\

In Section \ref{sec_couple} we presents the results \cite{CG17,Mor21,CN22}, which establish the existence of a remarkable family of vertex algebras
\begin{theorem*}[\ref{thm_couple}]
For each integer $p>0$ divisible by the lacing number $m$ there exists a  family of vertex algebras, which extends $\rV^\kappa(\g)\otimes \rV^{{\kappa^*}}(\g)$ and decomposes over it as follows
$$\A^{(p)}[\g,\kappa]
=\bigoplus_{\lambda\in Q^+} \bV_\lambda^\kappa(\g)\otimes \bV_{\lambda^*}^{{\kappa^*}}(\g) $$
for a pair of generic levels $\kappa,{\kappa^*}$ satisfying 
$$\frac{1}{\kappa+h^\vee}+\frac{1}{{\kappa^*}+(h^\vee)^\vee}=p$$
where $(h^\vee)^\vee$ denotes the dual Coxeter number of the Langlands dual root system. 
\end{theorem*}
By performing a quantum Hamiltonian reduction on the right side we obtain an algebra $\HA^{(p)}[\g,\kappa]$ extending $\rV^\kappa\otimes \rW^{{\kappa^*}}$ and performing a quantum Hamiltonian reduction on both sides we obtain an algebra $\HHA^{(p)}[\g,\kappa]$ extending $\rW^\kappa\otimes \rW^{{\kappa^*}}$. Note that reversing the roles of $\kappa,\kappa^*$ gives the fourth possibility extending $\rW^\kappa\otimes \rV^{{\kappa^*}}$. Note that \cite{Mor21} shows the existence of the vertex algebra from categorical considerations on the quantum group side, which is structurally interesting but gives not sufficient analytical properties to have a deformable family, on the other hand \cite{CN22} construct the deformable family on the vertex algebra side for the version with Hamiltonian reduction on one side.  \\

 In Section \ref{sec_limit} we take the limit $\kappa\to \infty$ of this form, which follows straightforward from the previous results and the previously discussed limit of $\bV_\lambda^\kappa$ resp. $\bW_\lambda^\kappa$ as bundles over the space of $\g$-connections resp. $\g$-opers. We would expect that these two vertex algebra bundles are related by a classical Hamiltonian reduction. The remarkable conjecture raised in \cite{CN22} is that the zero-fibre of $\HA^{(p)}[\g,\infty]$ is the Feigin-Tipunin algebra $\cW_p(\g)$. On the other hand the limit of $\A^{(p)}[\g,\kappa]$ and of the fourth possibility should have as zero-fibre an interesting vertex algebra, whose Hamiltonian reduction is $\cW_p(\g)$. In the example $(\sl_2,2)$ it is the small N=4 superconformal algebra \cite{CGL20}, see Section \ref{sec_explicit_p2_N4}.

In Section \ref{sec_coset} we turn to the case $p=1$ and $\g$ simply-laced. There we have an explicit realization  $\HA^{(1)}[\g,\kappa]=\rV^{\kappa-1}(\g) \times \rV^1(\g)$. The defining property is based on the diagonal embedding $\hat{\g}_\kappa  \hookrightarrow \hat{\g}_{\kappa-1}\times \hat{\g}_1$ and the GKO-type coset realization of $\rW^{{\kappa^*}}(\g)$ in \cite{ACL19}
$$\rV^\kappa(\g)\otimes \rW^{{\kappa^*}}(\g)  
 \;\hookrightarrow\; 
\rV^{\kappa-1}(\g) \times \rV^1(\g)$$
The limit $\kappa\to \infty$ of this algebra and its modules give  bundles over the space of \text{$\g$-connections} $\d+A(z)$ with $A(z)=\sum_{n\in\Z}A_nz^{-n-1}$. For the algebra itself, the bundle is over regular $\g$-connection, the zero-fibre is $\rV^1(\g)$, which is  the lattice vertex algebra of the root lattice, and all other fibres are $\rV^1(\g)$ as vector spaces with a deformed  Virasoro action
\begin{align*}
L_n^{\d+A}
&=L_n
+\sum_{n'+n''=n} (A_{n''})_{n'}
-\frac{1}{2}\sum_{n'+n''=n}  \langle A_{n'},A_{n''}\rangle
\end{align*}
where $(A_{n''})_{n'}$ denotes the mode operator $n'$ of the element $A_{n''}\in\g$. For arbitrary modules we encounter this deformed Virasoro for all types of singularity of $\d+A(z)$ at $z=0$: For regular connections, the additional terms are only a mild deformation by nilpotent degree-lowering elements. For regular singular connections $A(z)=A_0z^{-1}+\cdots$ we encounter degree-preserving elements, in particular we find $L_0$-Jordan blocks according to the additional summand $A_0$. For irregular singularities, the effect becomes much wilder and we refrain from discussing them in the present article, see however Question \ref{quest_irregular} and the explicit Examples \ref{ex_Affirregular} and in particular Example \ref{ex_SFirregular}. We expect our findings in the case $p=1$ to be related to \cite{ILT15, GIL20,GMS20}.

In Section \ref{sec_W} and \ref{sec_decomposition} we review the quantum Hamiltonian reduction of the coset model and review decomposition formulae for the restriction of certain modules to the subalgebra $\rV^\kappa(\g)\otimes \rW^{{\kappa^*}}(\g)$, essentially following \cite{ACF22}. 

In Section \ref{sec_bun_bV}, \ref{sec_bun_varphi} and \ref{sec_bun_bW} we discuss the limit and fibres of different types of modules and compare the results to the twisted modules. We also find limitations, since the integral form does not transport to the obvious integral form in the decomposition. Recall that for regular singular connections we have proven already established that the fibre is not a direct sum anymore.

In Section \ref{sec_explicit_p1} we use the explicit realization of $\V^\kappa=\HHA^{(1)}[\sl_2,\kappa]$ by joint work of the first author \cite{BFL16} as $\Vir^b\otimes \V_{\sqrt{2}\Z}^{Urod}$ with a modified Virasoro action called Urod algebra, and two explicit commuting Virasoro actions $L_n^{(b_1)},L_n^{(b_2)}$ giving
$$\Vir_{b_1}\otimes \Vir_{b_2}\longrightarrow \Vir^b\otimes \V_{\sqrt{2}\Z}^{Urod}$$
This explicit example $(\sl_2,1)$ was the initial motivation for the study of  $\HHA^{(1)}[\g,\kappa]$  in \cite{ACF22} and also for our present article. We find by explicit calculation:

\begin{corollary*}[\ref{cor_expl_p1}]
The limit of $\V^\kappa$  fibres over the space of regular Sturm-Liouville operators $\frac{\d^2}{\d z^2}+q(z)$ with $q(z)=\sum_{n\in\Z} \ell_nz^{-2-n}$. The fibre $\V^\infty|_{q(z)}$ is nonzero for ${\ell}_n=0,n\geq -1$ and equal to the lattice vertex algebra $\V_{\sqrt{2}\Z}$ with deformed Virasoro action 
\begin{align*}
L_n^{(b_1)}
&=\frac{1}{2\epsilon}f_{n+1} %z^{-1-(n+1)}
+L_n^{\V_{\sqrt{2}\Z}} %\frac{1}{2}(\partial\varphi)^2_n %{z^{-2-n}}
-{\ell}_n
-2\epsilon\sum_{k\in \Z}{\ell}_k e_{n-k-1}  %z^{-2-k-1-(n-k-1)}
\end{align*}
\end{corollary*}
 In particular the zero-fibre $\V^\infty|_0$ is isomorphic to $\V_{\sqrt{2}\Z}$ deformed by a connection $\d+\frac{1}{2\epsilon}fz^0$, which is what (in hindsight) we expect from classical Hamiltonian reduction.\\

In Section \ref{sec_explicit_p2} we turn to the case $(\sl_2, 2)$, where again joint work of the first author \cite{BBFLT13} gives an explicit description of $\HHA^{(2)}[\sl_2,\kappa]$ in terms of the N=1 superconformal algebra $\sVir_{N=1}$ times a free fermion and two explicit Virasoro actions. Recently \cite{CGL20} have established a remarkable description of  $\A^{(2)}[\g,\kappa]$ as the N=4 superconformal algebra and of $\HA^{(2)}[\sl_2,\kappa]$ as the affine vertex algebra $\hat{\osp}(1|2)_\kappa$. We start by the description in the second source, which allows us to proceed more systematically and understand the appearance of the $\osp(2|1)$ symmetry and the influence of the quantum Hamiltonian reduction. Then we turn to the first source, where everything can be checked explicitly and match our findings.
We determine the fibres:

\begin{lemma*}[\ref{lm_deformedModeAlgebra}]
The limit of $\HA^{(2)}[\sl_2,\kappa]$ and its modules fibres over the space of $\sl_2$-connections, with fibre over a point $\d+\sum_{n\in\Z} A_{-n-1}z^n$  given by the deformed mode algebra 
\begin{align*}
  %[(x'/\sqrt{a})_{-n-1},  (x'/\sqrt{a})_{-m-1} 
  %&=-(e'/a)_{-(n+m+1)-1}\\
  \{\psi_m^{\d+A}, \psi_n^{\d+A}\}
  &=-\langle A_{n+m},f'\rangle\\
  \{\bpsi_m^{\d+A}, \bpsi_n^{\d+A}\}
  &=\hphantom{-}\langle A_{n+m},e'\rangle\\
  %[(x'/\sqrt{a})_{-n-1},(y'/\sqrt{a})_m] 
  %&=-\frac{1}{2}(-n-1)1_{-(n+m+2)-0}+\frac{1}{2} (h'/a)_{-(n+m+1)-1}(w)\\
  \{\psi_m^{\d+A},\bpsi_n^{\d+A}\}&=\frac{1}{2}\langle A_{n+m},h'\rangle+
  m \delta_{m,-n}
\end{align*}
\end{lemma*}
The Virasoro action on all fibres is given by the action of $\sum_{i+j} :\bpsi_i\psi_j:$ and we find in Lemma \ref{lm_p2_translationconnection} that $[L_{-1},-]=\frac{\d}{\d z}+A(z)$ if we interpret $A(z)$ as a $\sl_2$-valued field acting on the vector-valued field with components $\psi(z),\bpsi(z)$. For example we hence find that indeed the zero-fibre indeed with symplectic fermions and the fibre over $\d+A_0z^{-1}$ corresponding to an (infinitesimal) automorphism $A_0$ indeed matches the twisted modules that we constructed in the first part.

We then repeat our computation explicitly in the Hamiltonian reduction originally given in the first authors joint work \cite{BBFLT13} as the product of an N=1 superconformal algebra $\sVir_{N=1}$ generated by $L_n,G_r$ and a free fermion $\mathcal{F}$ generated by $f$, and with two explicit commuting Virasoro actions giving the embedding of 
$$\Vir_{b_1}\otimes \Vir_{b_2}\longrightarrow \sVir_{N=1}\otimes \mathcal{F}$$
As in the case $\sl_2,p=1$ the fibration in the limit looks quite natural, namely it corresponds in the limit to the Virasoro algebra contained in  $\sVir_{N=1}$, leaving the generators $G_{-3/2}/b$ and the free fermion. The zero-fibre is again a respective deformation by $\d+fz^0 $ of the algebra of symplectic fermions. Note that in comparison to \cite{CGL20} we work with the  Hamiltonian reduction by $e/\kappa-1$ instead of $e-1$, which we feel has better limit behavior.

 \begin{question}
 The big center construction returns more than vertex algebra modules twisted by an (infinitesimal) automorphism $A_0\in\g$, but rather vertex algebra modules twisted by a $\g$-connection $\d+A(z)$ with $A(z)=\sum_{n\in \Z} A_n z^{-1-n}$, where $\d+A_0z^{-1}$ corresponding to usual twisted modules. Such a theory of twisted modules and intertwiners is certainly to be expected, but to our knowledge not worked out. Abstractly, it should be related to \cite{FBZ04} Section 7 and 9.5 for $\hat{\g}_\kappa\to \V$ 
 %$$\frac{\d}{\d z}\Y(a,z)=[L_{-1}+A(z),\Y(a,z)],\quad \Y(a,z)_{\widetilde{(n)}}\Y(b,z)=\Y(a_{(n)}b,z)$$
 %or 
 %$$\frac{\d}{\d z}\Y(a,z)=\Y(L_{-1}-x_0z^{-1}a,z)] %from d/dz+d/d\zeta
%\Y(a,z)_{\widetilde{(n)}}\Y(b,z)=\Y(a_{(n)}b,z)$$
  The modules we obtain in Section \ref{sec_explicit_p2} and~\ref{sec_explicit_p1} as well as a generalized $\Delta$-deformation as in Section \ref{sec_twistedFreeField} should be examples for such a notion. 

 On the categorical side, it should be clarified what the corresponding notion of a $G$-crossed tensor category is, in particular for a quantum group. In the case of regular singularities, one might expect that the additional freedom (different connections with same monodromy) corresponds to higher categorical structure, and possibly connections to higher cohomology, see \cite{CDGG21,LMSS20}. .
 \end{question}

\begin{question}\label{quest_irregular}
It would be very interesting to continue our discussion to the  case of irregular singularities. In this case, the deformed Virasoro action in Theorem \ref{thm_coset} has degree-raising terms and the modules we would have to discuss must be of Whittaker-type to accommodate $L_n/c \neq 0, n>0$. Such modules involving irregular singularities for the Virasoro algebra have been studied in \cite{GT12}. A theory of wild character varieties has been developed in \cite{Boa14} and versions of the Knizhnik-Zamolodchikov equations with irregular singularities seem to appear in \cite{FMTV00}. Our construction with the big center produces explicit examples for such modules in our cases $(\g,p)$ in Example \ref{ex_Affirregular} using spectral flow, and in our case $(\sl_2,2)$ in Example~\ref{ex_SFirregular}, where we also compute the representations of an irregularly twisted mode algebra and briefly discuss associated differential equations.
\end{question}

\begin{question}
All our discussion was aimed at vertex algebras, which correspond to a punctured disc. It would be interesting to compute globally  $\d+A$-twisted conformal blocks on a surface $\Sigma$ with respect to global $\g$-connection on $\Sigma$. Note that for irregular singularities, one should encounter the Stokes phenomenon, see Example~\ref{ex_SFirregular}.
\end{question}

\begin{acknowledgement*}
We are indebted to T. Creutzig and T. Dimofte for many helpful discussions and suggestions on the topics in this articles and their work. Thanks also for comments of A. Brochier on irregular singularities, of C. Schweigert on infinitesimal deformations of module categories, of J. Yang on  twisted representations and of A. Gainutdinov for suggesting observing the nilpotently deformed modules are tilting modules. We are thankful to the Humboldt foundation for awarding BF the research prize for an extended stay at the University of Hamburg, and for supporting SL with  the Feyodor Lynen program for an extended research stay at the University of Alberta, and we thank both universities and in particular T. Gannon, T. Creutzig and C. Schweigert for their hospitality.
\end{acknowledgement*}

\section{Twisted modules}\label{sec_twisted}

A vertex operator algebra $\V$ is, very roughly, a infinite-dimensional graded vector space with a multiplication map 
$$\Y(-,z): \V\otimes \V\to \V[[z,z^{-1}]]$$
depending on a formal complex variable $z$ as a Laurent series, and which is endowed with an action of the Virasoro algebra spanned by $L_n$ for , which must be compatible with $\Y$ and the action of $L_n$ on the complex variable via $-z^{n+1}\frac{\partial}{\partial z}$. Standard mathematical textbooks on the topic are \cite{FBZ04,Kac98}, to which we refer the reader for background. A module over a vertex operator algebra is similarly a vector space $\cM$ with a map 
$$\Y_\cM(-,z): \V\otimes \cM\to \cM[[z,z^{-1}]]$$
The Virasoro algebra acts also on $\cM$ via $\Y_\cM(L,z)=\sum_n L_n^\cM z^{-2-n}$. Let now $G$ be a group acting on $\V$ as automorphisms. Then a $g$-twisted module $\cM$ of $\V$ for $g\in G$ consists of a map with a more general dependency on $z$ as a multivalued functions (in more physical language, it is a non-local field), and where this multivaluedness is controlled by the action of $g$, namely
\begin{align}\label{twisted_monodromy}
\Y_\cM(ga,z)=e^{2\pi\i\; z\frac{\partial}{\partial z}}\Y_\cM(a,z)
\end{align}
Note that $e^{2\pi\i\; z\frac{\partial}{\partial z}}$ acts as analytic continuation counterclockwise around $z=0$, in particular we find explicitly from the exponential series
\begin{align*}
    e^{2\pi\i\; z\frac{\partial}{\partial z}}\;z^m &=e^{2\pi\i} z^m\\e^{2\pi\i\; z\frac{\partial}{\partial z}}\;\log(z) &=\log(z)+2\pi\i
\end{align*}
and thus $\{1,\log(z)\}$ give a Jordan block. The previous definition is fairly standard when the automorphism acts semisimply, in which case $\Y_\cM(a,z)$ for $a$ a $g$-eigenvector with eigenvalue $\lambda$ is a series of fractional powers $z^{m}$ with $\lambda=e^{2\pi\i m}$. In the nonsemisimple case one allows $\Y_\cM(a,-)$ to have a general regular singularity at $z=0$, which means it is a power series in $z^m$ and polynomially in $\log(z)$. The respective theory has been rigorously developed by Huang \cite{Hu10}, similar to the tensor product theory in the nonsemisimple case using intertwining operators \cite{HLZ}. A somewhat different view by Bakalov \cite{Bak16} has the advantage of being quite explicit and handy, and several examples are discussed.\\

The main application of twisted modules is that they are ordinary modules over the vertex subalgebra $\V^G$ fixed by $G$, called the orbifold. The general principle, which can be proven under certain assumptions, is that the $\V^G$-modules all come from decomposing the restriction of $g$-twisted modules of $\V$ for all $g\in G$. \\

On the categorical side, if $\cC$ is the braided tensor category of representations of $\V$, then the category $\cC_g$ of $g$-twisted modules is a module category over $\cC$ and conjecturally $\bigoplus_{g\in G} \cC_g$ has a $G$-action and a $G$-crossed braiding. The $G$-equivariantization is then a braided tensor category graded by $G$-conjugacy classes and should describe the $\V^G$-modules. For finite groups $G$ it is shown in \cite{ENOM09} that the $G$-crossed extensions are unique up to a choice in  $\mathrm{H}^{3}(G,\mathbb{C}^\times)$. The respective theory for Lie groups is not fully available, to our knowledge, but the second author has recently studied the twisted modules of the orbifold of the 2-dimensional Heisenberg algebra under the action of $\mathrm{SO}_3$.\\ %\cite{GLM22}

We now discuss the description of twisted modules more explicitly, following essentially \cite{BS16}. Assume we can write $g=g_sg_n$ for an automorphism $g_s$ acting semisimply and an automorphism $g_n$ acting locally nilpotently. We find it convenient below to describe both in terms of derivations $s,x$:
$$g=e^{-2\pi\i (s+x)}$$
Then \cite{BS16} Section 5 derives from his characterization the $g$-twisted commutator formula
\[[a_{m}^g,b_{n}^g]=\sum_{j=0}^\infty (\binom{m+x}{j}a)_{j}b)_{m+n-j}^g\]
with $m\in r+\Z,n\in r+\Z$ for $s$-eigenvectors $a,b$ with eigenvalues $r,t$, according to \ref{twisted_monodromy}. Here, $\binom{m+x}{j}$ is to be understood as a polynomial in the endomorphism $x$ and the mode operators $a_{n}$ define the vertex algebra via $\Y(a,z)b=:\sum_{n\in\Z} (a_{n}b)z^{-n-1}$  and  the twisted mode operators $a_{m}^g$ define the vertex module via $\Y_\cM(a,z)=\sum_{m\in r+\Z} a_m^g z^{-1-m}$ plus higher terms containing $\log(x)$. Note that once these terms not containing logarithmic terms are known, the logarithmic terms follow with the action of $x$ from \ref{twisted_monodromy}. 

\subsection{\texorpdfstring{$G$}{G}-twisted modules of \texorpdfstring{$\hat{\g}_\kappa$}{affine g}}\label{sec_twistAffine}

The affine Lie algebra $\hat{\g}_\kappa$ at level $\kappa$ has defining relations 
$$[a_m, b_n]=[a,b]_{m+n}+m\delta_{m,-n}\langle a,b\rangle \kappa$$
with $a,b\in \g$ and with the Killing form $\langle a,b\rangle$ in the normalization usual in vertex algebra theory, which is $\frac{1}{2 h^vee}$ times the trace in the adjoint representation. Modules over the universal vertex algebra are given by (suitable) representations of this Lie algebra, where the vertex operator for $a\in\g$ is 
$$\Y_\cM(a,z)=\sum_{n\in \Z} a_n z^{-n-1}$$
By \cite{Bak16} Section 6 we consider $g=e^{-2\pi\i (s+x)}$. The $g$-twisted modules of the universal affine vertex algebra is given by representations of the $g$-twisted mode algebra $\hat{\g}_{\kappa,g}$ with relations
\begin{align}\label{twisted_commutator}
[a_m^g, b_n^g]=[a,b]_{m+n}^g+\delta_{m,-n}\langle (m+x)a,b\rangle\kappa
\end{align}
and $m,n$ in $\Z$-cosets of the $s$-eigenvalues $r_a,r_b$ for the $s$-eigenvectors $a,b$. The twisted Virasoro action is derived from the action of the familiar Sugawara element $L$ on a twisted module, and we can express the the product  by \cite{Bak16} Section 6.2  in a twisted normally ordered product
\begin{align*}
    L^{\hat{\g}_{\kappa,g}}_n
&=\frac{1}{2(\kappa+h^\vee)} \Y_\cM(v^{(i)}_{(-1)}w^{(i)})_{z^{-n-2}}\\
&=\frac{1}{2(\kappa+h^\vee)}\sum_{m,i} :v^{(i)g}_m w^{(i)g}_{n-m}:-[(s_0+x) v^{(i)},w^{(i)}]_n^g-\kappa\delta_{n,0} \tr\binom{s_0}{2}1
\end{align*}
where $v^{(i)},w^{(i)}$ is a dual pair of $s$-homogeneous bases, and 
where $s_0$ is defined by $s_0a=r_0a$ with the choice of representative  $\Re(r_0)\in(-1,0]$ in the coset $r+\Z$ of the $s$-eigenvalue $r$.

\begin{example}
 We assume for simplicity $s=0$ and compute the commutator of a twisted translation operator with $a_{n}^g$ for any $a\in \g$
\begin{align*}
[L^{\hat{\g}_{\kappa,g}}_{-1},a_{n}^g]
&=\frac{1}{2(\kappa+h^\vee)}\big(\sum_{m\geq 0, i} 
2v^{(i)g}_{-m-1} [w^{(i)g}_{m},a_n^g]
+2[v^{(i)g}_{-1-m},a^g_n] w^{(i)g}_{m}
-[[x.v^{(i)},w^{(i)}]_{-1}^g,a_n^g]\big)\\
&=\frac{1}{2(\kappa+h^\vee)}\big(\sum_{m\geq 0, i} 
2v^{(i)g}_{-m-1} [w^{(i)},a]^g_{m+n}
+2[v^{(i)},a]^g_{-1-m+n}w^{(i)g}_{m}\big.\\
&\big.+2v^{(i)g}_{n-1} \langle (-n+x) w^{(i)},a\rangle\kappa
-[[x.v^{(i)},w^{(i)}],a]^g_{n-1}
-\delta_{-1,-n}(-1)\langle[x.v^{(i)},w^{(i)}],a\rangle\kappa\big)
\end{align*}
Here, the term $2v^{(i)g}_{n-1}\langle  (-n+x)w^{(i)},a\rangle\kappa$ comes from the first or second commutator respectively for $n$ positive or negative, and by the invariance of the form and then the defining property of the dual basis it evaluates to $2((-n-x)a)_{n-1}\kappa$. Assume  for simplicity $n<0$, then by invariance all terms $2v^{(i)g}_{-m-1} [w^{(i)},a]^g_{m+n}$ and $2[v^{(i)},a]^g_{-1-m'+n}w^{(i)g}_{m'}$ for~$m'=m-(-n)$ cancel, leaving only the first term for $m=0,\ldots (-n)-1$. By invariance and symmetry this can be rewritten as $(-n)/2$ identical commutators
\begin{align*}
%[L^{\hat{\g}_{\kappa,g}}_{-1},a_{n}]
\hspace{1.7cm}
&=\frac{1}{2(\kappa+h^\vee)}\big(\sum_{i} 
(-n)[v^{(i)}, [w^{(i)},a]]^g_{n-1}
+2((-n-x)a)_{n-1}\kappa
-[[x.v^{(i)},w^{(i)}],a]^g_{n-1}\big)
\end{align*}
The expression $[v^{(i)}, [w^{(i)},a]]$ is the action of the quadratic Casimir operator in the adjoint representation, which is the identity in the standard normalization and $2h^\vee$ in our normalization. Altogether 
\begin{align*}
[L^{\hat{\g}_{\kappa,g}}_{-1},a_{n}^g]
&=(-n)a_{n-1}^g-(x.a)_{n-1}^g
\end{align*}
We can interpret this result in the language of fields to make the connection visible
$$[L^{\hat{\g}_{\kappa,g}}_{-1},a(z)]=\left(\frac{\d}{\d z}-A(z)\right)a(z)$$
where $a(z)=\sum_n a_nz^{-n-1}$ and in our case $A(z)=xz^{-1}$. For an  arbitrary connection $\d+A(z)$ we recover in section \ref{sec_coset} a similar formula.
\end{example}
The representations of the twisted mode algebra can be induced as modules from the twisted horizontal Lie algebra, spanned by $a=a_0,\;a\in\g$ with relations 
$$[a^g,b^g]=[a,b]^g+\langle x.a,b\rangle$$
We are maybe a bit surprised at first to find that this has the same representation theory as $\g$. We now make explicit in our case the well-known statement that inner automorphisms do not change the category of representations (compare in the logarithmic case  \cite{BS16} Proposition 6.2 and \cite{Yang17}) and compute how the Virasoro action is modified by this:

\begin{lemma}\label{lm_twistedrepaffine}
Consider the inner automorphism $g=e^{-2\pi\i y}$ of $\g$ with $y=s+x$ for $s$ semisimple and $x$ nilpotent in $\g$.

\begin{enumerate}[a)]
\item We have a Lie algebra isomorphism  $f:\hat{\g}_{\kappa,g} \cong \hat{\g}_{\kappa}$ given explicitly by 
$$f:a_n^g\mapsto a_{n-r}-\eta(a)\delta_{n-r,0}$$ 
for any eigenvector $\ad_s(a)=r a$, with the $1$-cochain $\eta(a)=\langle y,a\rangle K$.
\item The image of the twisted Virasoro action is as follows 
$$f(L_n^{\hat{\g}_{\kappa,g}})=L_n^{\hat{\g}_\kappa}
-\frac{2\kappa}{2(\kappa+h^\vee)}y_n+\delta_{n,0}\frac{\kappa^2}{2(\kappa+h^\vee)}\langle y,y\rangle$$
\end{enumerate}
\end{lemma}

\begin{proof} 
a) We check that we have an isomorphism
	\begin{align*}
	&[f(a_m^g),f(b_n^g)]\\
	&=[a_{m-r}-\eta(a)\delta_{m-r_a,0},b_{n-t}-\eta(b)\delta_{n-r_b,0})]\\
	&=[a_{m-r_a},b_{n-r_b}]\\
	&=[a,b]_{m+n-(r_a+r_b)}+\delta_{m-r_a,-(n-r_b)}\langle(m-r_a) a,b\rangle\kappa\\
	&=[a,b]_{m+n-(r_a+r_b)}+\delta_{m-r_a,-(n-r_b)}( m\langle a,b\rangle\kappa-\langle[s,a],b\rangle\kappa)\\
	&f\left([a,b]_{n+m}^g+\delta_{m,-n}\langle(m+x) a,b\rangle \kappa\right)\\
	&=[a,b]_{m+n-(r+t)}-\eta([a,b])\delta_{m+n-(r_a+r_b)}
	+\delta_{m,-n}(m\langle a,b\rangle \kappa+\langle[x,a],b\rangle \kappa)\\
	\end{align*}
	The first summands agree, the summands involving $\langle a,b\rangle\kappa$ agree because the Killing forms vanishes unless $r_a+r_b=0$, and the remaining terms agree with the definition $\eta([a,b])=\langle s+x,[a,b]\rangle\kappa=\langle[s+x,a],b\rangle\kappa$.\\

b) This isomorphism $f:\hat{\g}_{\kappa,g}\to \hat{\g}_\kappa$ does not preserve the Sugawara vector giving the Virasoro action. We now compute the new Virasoro action on $\hat{\g}$:\\

 We take a dual pair of $s$-homogeneous bases $v^{(i)},w^{(i)}$ and consider the action of the Sugawara element $L$ on the twisted module, which \cite{BS16} Section 6.2 can be expressed using the twisted normally ordered product
\begin{align*}
    L^{\hat{\g}_{\kappa,g}}_n
&=\frac{1}{2(\kappa+h^\vee)} \Y_\cM(v^{(i)}_{(-1)}w^{(i)})_{z^{-n-2}}\\
&=\frac{1}{2(\kappa+h^\vee)}\sum_{m,i} :v^{(i)g}_m w^{(i)g}_{n-m}:-[(s_0+x) v^{(i)},w^{(i)}]_n^g-\kappa\delta_{n,0} \tr\binom{s_0}{2}1
\end{align*}
Transporting this formula to $\hat{\g}$ using $f$ we get 
\begin{align*}
 f(L_n^{\hat{\g}_{\kappa,g}})
&=L_n^{\hat{\g}_\kappa}
-\frac{2\kappa}{2(\kappa+h^\vee)} \eta(v^{(i)})w^{(i)}_n 
%m-r=0, dualbasis r+t=0, then n-m-t=n 
+\delta_{n,0}\frac{\kappa^2}{2(\kappa+h^\vee)}\eta(v^{(i)})\eta(w^{(i)})
%m-r=0, n-m-t=0, r+t=0
\end{align*}
where $f$ has no effect on the additional terms, so the additional terms can be reabsorbed into the definition of $L_n^{\hat{\g}_\kappa}$. Now: For $\eta(a)=\langle s+x,a\rangle$ we have by the defining property of the dual basis of the Killing form $\eta(v^{(i)})w^{(i)}_n=(s+x)_n$ and $\eta(v^{(i)})\eta(w^{(i)})=\eta(s+x)=\langle s+x,s+x\rangle=\langle s,s\rangle$, since $x$ acts nilpotently. This concludes the proof of the second assertion. 

\begin{comment}
We tried to explicitly write down mode operators, but the twisted normally ordered product makes this more complicated than the following sketch
\begin{align*}
2(\kappa+h^\vee)\; f(L_n^{\hat{\g}_g})
&=2(\kappa+h^\vee)\; L_n^{\hat{\g}}
-2\kappa \eta(v^{(i)})w^{(i)}_n 
%m-r=0, dualbasis r+t=0, then n-m-t=n 
+\delta_{n,0}\kappa^2\eta(v^{(i)})\eta(w^{(i)})\\
%m-r=0, n-m-t=0, r+t=0
&-[(s_0+x) v^{(i)},w^{(i)}]_n^g
+\delta_{n,0}\kappa\eta([(s_0+x) v^{(i)},w^{(i)}])\\
&-\delta_{n,0}\kappa\tr\binom{s_0}{2}\\
\end{align*}
For $\eta(a)=\langle s+x,a\rangle$ we have by the defining property of the dual basis of the Killing form $\eta(v^{(i)})w^{(i)}_n=(s+x)_n$ and $\eta(v^{(i)})\eta(w^{(i)})=\eta(s+x)=\langle s+x,s+x\rangle=\langle s,s\rangle$. Moreover we have $[[s+x,v^{(i)}],w^{(i)}]=(s+x)$, as can be seen from evaluating for arbitrary $a$ the Killing form $\langle a,[[s+x,v^{(i)}],w^{(i)}] \rangle=\langle [w^{(i)},a],[s+x,v^{(i)}] \rangle=\langle w^{(i)},[a,[s+x,v^{(i)}]] \rangle=\tr([a,[s+x,-]])=\langle a,s+x\rangle$. Hence altogether, with the additional assumption that $s=s_0$ and using that  $\tr(s)=0$ by symmetry, we have the asserted result.
\end{comment}
\end{proof}

\begin{comment}
\begin{example}
	We now discuss the case $\mathfrak{sl}_2$, which has basis $e,h,f$ and Killing form 
	$$(X,Y)=4\mathrm{Tr}(XY)=4\begin{pmatrix} 0 & 0 & 1 \\ 0 & 2 & 0 \\ 1 & 0 & 0\end{pmatrix}$$
	so $v^i\otimes v_i=\frac{1}{4}(e\otimes f+\frac{1}{2}h\otimes h+f\otimes e)$ and $tr(sh(sh-1))=\frac{2s(2s-1)}{2}+0+\frac{-2s(-2s-1)}{2}=4s^2$. Hence for $x=e$ we have $\V_{A_1}$ with 
	$$L_n^{\varphi} =L_n+t\cdot e_n$$
	
	Indeed we check that this expression defines a deformed Virasoro action for every field with self-commuting modes $e_n$: 
	We check relations for $L_n^{\varphi} =L_n+t\cdot e_{n+r}$
	\begin{align*}
	&[L_m^ \varphi,L_n^\varphi]-[L_m,L_n]\\
	&=t [L_m,e_{n+r}]-t [L_n,e_{m+r}]+t^2 [e_m,e_n] \\
	(selfcommuting)
	&=t [L_m,e_{n+r}]-t [L_n,e_{m+r}] \\
	% -q=-2-m, -p=-1-n-r
	&=t\sum_{l<0}{-(-2-m)-1\choose -l-1}Y(Y(T)_l e)_{(-2-m)+(-1-n-r)-l}\\
	&-t\sum_{l<0}{-(-2-n)-1\choose -l-1}Y(Y(T)_l e)_{(-2-n)+(-1-m-r)-l}\\
	(primary)
	&=t\left({-(-2-m)-1\choose 0}-{-(-2-n)-1\choose 0}\right) 
	Y(L_{-1}e)_{-2-m-n-r}\\ %(-1-m-n-r)Y(e)_{-1-m-n-r}
	&+t\left({-(-2-m)-1\choose 1}-{-(-2-n)-1\choose -1}\right) 
	Y(L_{0}e)_{-1-m-n-r}\\
	(weight\;one)&=t((-1-m-n-r)+(m+n+2))Y(e)_{-1-m-n-r}\\
	&=t(1-r)Y(e)_{-1-m-n-r}\\
	\end{align*} 
	For $r=0$ this is equal to $L_{m+n}^\varphi-L_{m+n}$.
\end{example}
\end{comment}

\newcommand{\cone}[4]{
	\draw[fill=black] (#1,#2) circle[radius=0.02cm];
	\draw (#1,#2)--(#1+0.3*#3,#2+#3);
	\draw (#1,#2)--(#1-0.3*#3,#2+#3);
	\draw (#1,#2) node[anchor=west]{\small #4};
}
\newcommand{\ncone}[4]{
	\draw[fill=black] (#1,#2) circle[radius=0.02cm];
	%\draw (#1,#2)--(#1+0.3*#3,#2+#3);
	%\draw (#1,#2)--(#1-0.3*#3,#2+#3);
	\draw (#1,#2) node[anchor=west]{\small #4};
}
\newcommand{\graycone}[4]{
	\draw[fill=black] (#1,#2) circle[radius=0.02cm];
	\draw[gray] (#1,#2)--(#1+0.3*#3,#2+#3);
	\draw[gray] (#1,#2)--(#1-0.3*#3,#2+#3);
	\draw (#1,#2) node[anchor=west]{\small #4};
}

\newcommand{\startendoff}{0.06}
\newcommand{\arrL}[6] {
	\draw [opacity=0.3, -stealth](#1-\startendoff*#1+\startendoff*#3,#2-\startendoff*#2+\startendoff*#4) -- (#3-\startendoff*#3+\startendoff*#1,#4-\startendoff*#4+\startendoff*#2);
	\draw (#6*#1+#3-#6*#3,#6*#2+#4-#6*#4) node[anchor=mid]{\tiny #5};
}
\newcommand{\arrLg}[6] {
	\draw [opacity=0.3,darkgreen, -stealth](#1-\startendoff*#1+\startendoff*#3,#2-\startendoff*#2+\startendoff*#4) -- (#3-\startendoff*#3+\startendoff*#1,#4-\startendoff*#4+\startendoff*#2);
	\draw (#6*#1+#3-#6*#3,#6*#2+#4-#6*#4) node[anchor=mid]{\tiny #5};
}
%Not used
\newcommand{\arrLb}[6] {
	\draw [opacity=0.3,blue, -stealth](#1-\startendoff*#1+\startendoff*#3,#2-\startendoff*#2+\startendoff*#4) -- (#3-\startendoff*#3+\startendoff*#1,#4-\startendoff*#4+\startendoff*#2);
	\draw (#6*#1+#3-#6*#3,#6*#2+#4-#6*#4) node[anchor=mid]{\tiny #5};
}
\newcommand{\arrLgd}[6] {
	\draw [opacity=0.6,green,dotted, -stealth](#1-\startendoff*#1+\startendoff*#3,#2-\startendoff*#2+\startendoff*#4) -- (#3-\startendoff*#3+\startendoff*#1,#4-\startendoff*#4+\startendoff*#2);
	\draw (#6*#1+#3-#6*#3,#6*#2+#4-#6*#4) node[anchor=mid]{\tiny #5};
}

\subsection{\texorpdfstring{$\SL_2$}{SL2}-twisted modules of the triplet algebra}\label{sec_twistedTriplet}

We now turn to the triplet algebra $\cW_{\sl_2,p}$ which is a $C_2$-cofinite vertex algebra with non-semisimple category of modules, which is as a braided tensor category equivalent to the category of modules over a certain quasi-variant $\tilde{u}_q(\sl_2)$ of the small quantum group $u_q(\sl_2)$ \cite{FGST06,CGR20,GLO18}. It has a conjectural action of $\SL_2$ in terms of long screening and divided power of short screening. The triple algebra is a vertex subalgebra of the lattice vertex algebra $\V_{\sqrt{2p}\Z}$ as kernel of a short screening operator. \\

For $p=2$ the triplet algebra is isomorphic to the even part of the vertex superalgebra of symplectic fermions, in this case there is an evident action of $\SL_2=\mathrm{Sp}_2$. \cite{Kau00, GR17}. The symplectic fermions are defined by generating fields  $\psi(z)=\sum_{n\in \Z} \psi_n z^{-1-n}$ and $\bpsi(z)=\sum_{n\in \Z} \bpsi_n z^{-1-n}$ with mode operators fulfilling 
$$\{ \psi_{m}, \bpsi_{n}\}=\delta_{m,-n} m$$
with $\{-,-\}$ denoting the anticommutator. The twisted mode algebra for a nilpotent automorphism is discussed in \cite{BS16} Section 3 and 4, and for general $g$ one can proceed similarly: Up to automorphisms resp. conjugation in $\SL_2$  we have $g=e^{-2\pi\i (s+x)}$ for 
$$s=\begin{pmatrix} r & 0 \\ 0 & -r \\\end{pmatrix},\qquad 
x=\begin{pmatrix} 0 & 0 \\ 2t & 0 \\\end{pmatrix},\qquad 
$$
in the basis $\psi,\bpsi$. From the twisted commutator formula (\ref{twisted_commutator}) we get 
\begin{corollary}\label{cor_twistedSF}
The algebra of twisted mode operators is generated by $\psi_{r+n},\bpsi_{-r+n}$ for $n\in \Z$ with the following relations
\begin{align*}
\{ \psi_{r+m}^g, \psi_{r+n}^g\}&= 2t\delta_{m,-n} \\
\{ \psi_{r+m}^g, \bpsi_{-r+n}^g\}&= (r+m)\delta_{m,-n} \\
\{ \bar{\psi}_{-r+m}^g, \bpsi_{-r+n}^g\}&=0 
\end{align*} 
The action of the Virasoro algebra is by \cite{BS16} Proposition 4.1 with shifted indices
\CommentsForMe{More...Shifted indices make exactly binomial term?}
\begin{align*}
L_n
&=\begin{cases}
\sum_{j>0} \bpsi_{-j}^g\psi_{j}^g
+\frac{1}{2}(\bpsi_{0}^g\psi_{0}^g-\psi_{0}^g\bpsi_{0}^g)
-\sum_{i>0} \psi_{-i}^g\bpsi_{i}^g +\binom{s_0}{2},\quad &\text{ for }n=0\\
\sum_{i+j=n} \bpsi_{i}^g\psi_{j}^g&\text{ for }n\neq 0\\
\end{cases}
\end{align*}
\end{corollary}
We now discuss the representation theory depending on $g=e^{-2\pi\i (s+x)}$:
\begin{lemma}\label{lm_repTwistedClifford}
The abelian category of representations of the mode algebra in Corollary \ref{cor_twistedSF} is equivalent to the category of representations of the zero-mode algebra generated by $\psi_0^g,\bpsi_0^g$ resp. to the trivial algebra for $r\not\in \Z$. The bijection is given by generalized highest-weight representations in the sense that a representation of the zero-mode algebra is trivially extended to a representation of $\psi_n^g,\bpsi_n^g,\;n\geq 0$ and then induced up.
\end{lemma}
\begin{proof}
    The algebra of mode operators is apparently an infinite-dimensional complex  Clifford algebra. 
    All nontrivial commutation relations only appear for pairs of indices $r\pm m$, so it suffices to study representations of the zero mode algebra for $r=m=0$, which is the Clifford algebra associated to the quadratic form 
    $$\begin{pmatrix}
        0 & 2t \\
        2t & 0 
    \end{pmatrix}$$
    and on the other hand representations of $\psi_{r+m}^g,\bpsi_{-r+m}^g, \psi_{r-m}^g,\bpsi_{-r-m}^g$ for fixed integer $m$ with $\pm r +m>0$. This is the Clifford algebra associated to the quadratic form
    $$\begin{pmatrix}
        0 & 0 & 2t & r+m \\
        0 & 0 & r-m & 0 \\
        2t & r-m & 0 & 0 \\
        r+m & 0 & 0 & 0
    \end{pmatrix}$$
    For $r\pm m\neq 0$ this quadratic form is nondegenerate, so the corresponding Clifford algebra is a simple algebra. More explicitly, after a base change we have  two mutually commuting base pairs 
    $\psi_{r+m}^g, \bpsi_{-r-m}^g$ and $\bpsi_{-r+m}^g,\psi_{r-m}^g-\frac{2t}{r+m}\bpsi_{-r-m}^g$. The unique $4$-dimensional simple module of this algebra is generated by a vector $v$ with $\psi_{r+m}^gv=\bpsi_{-r+m}^gv=0$. 
\end{proof}

We now determine the modules induced by modules of the zero-mode algebra $\psi_0,\bpsi_0$, which also possess a compatible structure of a super vector space, given by an action of the parity operator $\pi$ with $\pi\psi_n^g=-\psi_n^g\pi$ and $\pi\bpsi_n^g=-\bpsi_n^g\pi$. These induced modules correspond to twisted modules over the Vertex subalgebra of symplectic fermions. The action of the Virasoro algebra on the generating zero-mode module fulfills 

\begin{align*}
L_n|_{\cM_0}
&=\begin{cases}
0 &\text{ for }n> 0\\
\frac{1}{2}(\bpsi_{0}^g\psi_{0}^g-\psi_{0}^g\bpsi_{0}^g)
+\binom{s_0}{2},\quad &\text{ for }n=0\\
\sum_{i+j=n,i,j<0} \bpsi_{i}^g\psi_{j}^g&\text{ for }n<0\\
\end{cases}
\end{align*}

To compute the twisted modules over the even vertex subalgebras, we have to determine also the $\pi$-twisted sectors, which correspond to the matrix $s+\frac{1}{2}$. The subsequent decomposition over the even vertex subalgebra, which corresponds to taking the $\Z_2$-equivariant objects, is simply a decomposition into the eigenspaces of $\pi$. \\ \CommentsForMe{irreduciblility of induced module?}

{\bf Case 1:} $s+x=\begin{pmatrix} 0 & 0 \\ 0 & 0\end{pmatrix}$.\\

This is the case of untwisted modules and the result is well-known, but it is instructive to briefly repeat the computation: The zero-mode algebra is 
$$\psi_{0}^2=0,\quad \bpsi_{0}^2=0,\quad \{ \psi_{0}, \bpsi_{0}\}=0$$
This superalgebra of dimension $4$ has two irreducible modules, generated by $v^\pm$ with $\psi_0v^\pm=\bpsi v^\pm=0$ and parity $\pi v^\pm=\pm v^\pm$, and projective covers

\begin{center}
\begin{tikzcd}%[column sep=tiny]
 &  v^\pm \arrow{dl}[above]{\psi_0\quad} \arrow{dr}[above]{\quad\bpsi_0} &  \\ %[left] etc
v'^{\;\mp} \arrow{dr}[below]{\bpsi_0\quad} & & v''^{\;\mp} \arrow{dl}[below]{\quad\psi_0} \\
& 1^\pm & 
\end{tikzcd}    
\end{center}

The Virasoro algebra acts on these zero-mode modules with $L_0 1^\pm=0$ resp. on the projective covers with $L_0 v^\pm=1^\pm$. 

To illustrate an example, we draw we draw the module induced from the extension $1^+\to \cM\to 1^-$, which is usually called Wakimoto module or Verma module. It is in fact a lattice vertex superalgebra $\V_\Z$ which serves as free field realization as discussed in the next section. In the diagram, where cones denote Virasoro modules, and as usual the Y-axis is the $L_0$-grading and the X-axis denotes the $\sl_2$-grading resp. the $\Z$-lattice grading.

\begin{center}
\begin{tikzpicture}
\begin{scope}[yscale=-2.5,xscale=2.5]
	\cone{0}{0}{1.5}{$1^{+}$};
	\cone{1}{0}{1.5}{$1^{-}$};
	%\arr{1}{0}{0}{0}{$\psi_{0}$;}
	
	\cone{-1}{1}{1.5}{$\psi_{-1}1^{+}$};
	\cone{0}{1}{1.5}{$\psi_{-1}1^{-}$};
	\cone{1}{1}{1.5}{$\bpsi_{-1}1^{+}$};
	\cone{2}{1}{1.5}{$\bpsi_{-1}1^{-}$};

	\arrL{0}{1}{0}{0}{$L_{1}$}{0.5};
	\arrL{1}{0}{1}{1}{$L_{-1}$}{0.5};
	 
	 \cone{-2}{3}{1.5}{$\psi_{-2}\psi_{-1}1^{+}$};
	 \cone{-1}{3}{1.5}{};
	 \cone{0}{3}{1.5}{};
	 \cone{1}{3}{1.5}{};
	 \cone{2}{3}{1.5}{};
	 \cone{3}{3}{1.5}{$\bpsi_{-2}\bpsi_{-1}1^{-}$};

    \arrL{-1}{3}{-1}{1}{$L_{2}$}{0.5};
    \arrL{0}{1}{0}{3}{$L_{-2}$}{0.5};
    \arrL{1}{3}{1}{1}{$L_{2}$}{0.5};
    \arrL{2}{1}{2}{3}{$L_{-2}$}{0.5};
\end{scope}
\end{tikzpicture}
\end{center}
 
We now discuss the $\pi$-twisted sector with $s+x=\begin{pmatrix} 1/2 & 0 \\ 0 &  1/2\end{pmatrix}$.\\

Due to the half-integer grade it has a trivial zero-mode algebra, so there are two irreducible modules generated by $1^\pm$ with $\psi_0^\pi,\psi_0^\pi$. The Virasoro algebra acts on these zero-mode modules with $L_0=\binom{s_0}{2}=-\frac{1}{8}$ and the next nontrivial terms are $L_0\psi_{-1/2}1^\pm = (-\psi_{-1/2}\bpsi_{1/2}-\frac{1}{8})\psi_{-1/2}1^\pm=\frac{3}{8}\psi_{-1/2}1^\pm$ and similarly $L_0\bpsi_{-1/2}1^\pm =\frac{3}{8}\bpsi_{-1/2}1^\pm$.

\begin{comment}
%IMAGES
\begin{center}
\hspace*{-2cm}
	\begin{tikzpicture}
	\begin{scope}[yscale=-6,xscale=4]
	
	\cone{0}{0}{1.5}{$1$};
	
	%Parallel \arrLg{0.02}{0.02}{-0.26}{1}{\textcolor{darkgreen}{$L_{-2}\quad$}}{.5};
	\arrLg{0.018}{-0.03}{-0.26}{1}{\textcolor{darkgreen}{$L_{-2}\quad$}}{.5};
	\arrLg{-0.0}{0.0}{-0.22}{0.25}{\textcolor{darkgreen}{$L_{-1}$}}{.3};
	\arrLb{0}{0.25}{0}{0.11}{\textcolor{blue}{$\hspace{0.4cm} L_{1}$}}{.3};
	\arrLb{0}{0.2}{0.0}{1}{\textcolor{blue}{$\hspace{0.4cm} L_{-1}$}}{.3};

	\cone{-1}{0.25}{1.5}{$e^{-\alpha}$};
	\graycone{0}{0.25}{1.5}{$\;\;\,\partial\varphi$};
	\cone{1}{0.25}{1.5}{$e^\alpha$};

	\cone{-2}{1}{1}{$e^{-2\alpha}$};
	\graycone{-1}{1}{1}{$\partial\varphi e^{-\alpha}$};
	\ncone{-0.269}{1}{1}{$(\partial\varphi)^2$};
	\graycone{0.0}{1}{1}{$\partial^2\varphi$};
	\graycone{1}{1}{1}{$\partial\varphi e^{\alpha}$};
	\cone{2}{1}{1}{$e^{2\alpha}$};

	\end{scope}
	\end{tikzpicture}
\end{center} 
\end{comment}
\enlargethispage{1cm}

\begin{center}
	\begin{tikzpicture}
	\begin{scope}[yscale=-2.5,xscale=2.5]
	
	\cone{0}{-0.125}{1.5}{$1^\pm$};
	
	\arrL{0}{-0.125}{0}{1.875}{$L_{2}$}{0.6};
	
	\cone{-1}{0.375}{1.5}{$\psi_{-\frac{1}{2}}1^\pm$};
	\cone{1}{0.375}{1.5}{$\bpsi_{-\frac{1}{2}}1^\pm$};
	%\arrL{-1}{0.375}{1}{0.375}{\textcolor{darkgreen}{$L_0$}}{0.5};
	%\arrL{-1}{0.375}{1}{1.375}{\textcolor{darkgreen}{$L_{-1}$}}{0.6};
	%\arrL{-1}{1.375}{1}{0.375}{\textcolor{darkgreen}{$L_{1}$}}{0.6};
	
	\ncone{0}{0.875}{1.5}{$\psi_{-\frac{1}{2}}\bpsi_{-\frac{1}{2}}1^\pm$};
	
	\ncone{-1}{1.375}{1.5}{$\psi_{-\frac{3}{2}}1^\pm$};
	\ncone{1}{1.375}{1.5}{$\bpsi_{-\frac{3}{2}}1^\pm$};
	
	\cone{-2}{1.875}{1.5}{$\psi_{-\frac{3}{2}}\psi_{-\frac{1}{2}}1^\pm$};
	\cone{0}{1.875}{1.5}{};
	\ncone{0.1}{1.875}{1.5}{};
	\cone{2}{1.875}{1.5}{$\bpsi_{-\frac{3}{2}}\bpsi_{-\frac{1}{2}}1^\pm$};
	%\arrL{-2}{1.875}{0}{1.875}{\textcolor{darkgreen}{$L_0$}}{0.5};
	%\arrL{-2}{1.875}{0}{-0.125}{\textcolor{darkgreen}{$L_2$}}{0.6};
		
	\end{scope}
	\end{tikzpicture}
\end{center} 

We also discuss the decomposition behaviour under restriction to the invariants $\V^{G}$, which is the Virasoro algebra, via the action of the centralizer $\mathrm{Cent}(g)=G=\mathrm{SL}_2$: The action of $\sl_2$ preserves the (vertical) $L_0$-grading and the irreducible module above apparently decompose into the $\sl_2$-representations of dimension $1,3,5,\ldots$ resp. $2,4,6,\ldots$ tensored with irreducible Virasoro representations $\varphi_{1,n}$.
\\
    
{\bf Case 2:} $s+x=\begin{pmatrix} 0 & 0 \\ 2t & 0\end{pmatrix},\quad t\neq 0$.\\

The zero-mode algebra is 
$$(\psi_{0}^g)^2=t,\quad (\bpsi_{0}^g)^2=0,\quad \{ \psi_{0}^g, \bpsi_{0}^g\}=0$$    
One-dimensional representations with $\psi_{0}^g=\pm\sqrt{2t}$ do not possess a compatible structure of a supervector space. Instead this only leaves a unique irreducible representation $\cW_{\begin{psmallmatrix} 0 & 0 \\ 2t & 0\end{psmallmatrix}}$
$$\langle 1^{+},1^{-}\rangle,\qquad 
\psi_{0}^g1^{-}=1^{+},\quad
\psi_{0}^g1^{+}=t\,1^{-},\quad
\bpsi_{0}^g1^{\pm}=0,\quad
\pi 1^{\pm}=\pm 1^{\pm}
$$
\begin{center}
\begin{tikzcd}
1^+ \arrow[r, "\psi_0^g"]  & 1^- \arrow[l]  
\end{tikzcd}    
\end{center}
and a unique indecomposable self-extension
\begin{center}
%\begin{tikzcd}
%v^+ \arrow[d, "\bpsi_0^g"]\arrow[r, "\psi_0^g"]  & v^- \arrow[d, %"\bpsi_0^g "]\arrow[l]  \\ 
%1^+ \arrow[r, "\psi_0^g"]  & 1^- \arrow[l]  \\ 
%\end{tikzcd}    
\begin{tikzcd}
    & v^+ 
    \ar[dl,"\psi_0^g", shift right=-0.5ex]
    \ar[dr,swap,"\bpsi_0^g", shift left=-0.5ex] & \\
    v^-
    \ar[dr, "\bpsi_0", shift right=-0.5ex] 
    \ar[ur, "\psi_0^g", shift right=-0.5ex] && 
    1^+ 
    \ar[dl, swap, "\psi_0^g", shift left=-0.5ex] \\
    %\ar[ul, swap, "\psi_1", shift left=-0.5ex] \\
    & 1^-
    %\ar[ul, "\psi_1", shift right=-0.5ex]
    \ar[ur, swap, "\psi_0^g", shift left=-0.5ex]
\end{tikzcd}  
\end{center}

The structure of $\cW_{\begin{psmallmatrix} 0 & 0 \\ 2t & 0\end{psmallmatrix}}$ as an algebra over the Virasoro algebra is very rich. For illustration we now compute some instances, where we use green color to highlight the new terms in the twisted case. In particular we observe Jordan blocks for $L_0$ and new extensions between Virasoro modules.

\begin{align*}
L_0\, 1^{\pm} 
&=0 
\qquad
& L_0\,\bpsi^g_{-1}1^{\pm}
%&=\bpsi^g_{-1}\psi^g_{1}\bpsi^g_{-1}1^{\pm}\\
&=\bpsi^g_{-1}1^{\pm}
\qquad
&L_0\,\psi^g_{-1}1^{\pm}
%&=(\bpsi^g_{-1}\psi^g_{1}
%-\psi^g_{-1}\bpsi^g_{1})\psi^g_{-1}1^{\pm}\\
&=\textcolor{darkgreen}{2t\bpsi^g_{-1}1^{\pm}}+\psi^g_{-1}1^{\pm}\\
&\\
L_{-1}\,1^{\pm}
%&=\bpsi^g_{-1}\psi^g_{0}1^{\pm}\\
&=\begin{cases}
\textcolor{darkgreen}{t\,\bpsi^g_{-1}1^{-}}\\
\bpsi^g_{-1}1^{+}\\
\end{cases}
\qquad
&L_1\,\bpsi^g_{-1}1^{\pm}
&=\bpsi^g_{1}\psi^g_{0}\bpsi^g_{-1}1^{\pm}=0
\qquad
&L_1\,\psi^g_{-1}1^{\pm}
&=\bpsi^g_{1}\psi^g_{0}\psi^g_{-1}1^{\pm}
=\begin{cases}
\textcolor{darkgreen}{t\,1^{-}}\\
1^{+}\\
\end{cases}\\
\end{align*}

We again draw these results in a diagram. Observe however, that the $L_0$-action is not semisimple and that the $\sl_2$-grading is not preserved.

%\begin{center}
%	\begin{tikzpicture}
%	\draw (0,0)--(1.4,0) (1.6,0)--(3,0);
%	\draw (-0.1,0) circle[radius=0.1cm] node[anchor=south]{$ q^{4}$} node[anchor=north]{$ 4m$}
%	(1.5,0) circle[radius=0.1cm] node[anchor=south]{$ -1$}node[anchor=north]{$ 1$}
%	(3.1,0) circle[radius=0.1cm] node[anchor=south]{$ -q^{-2}$};
%	\draw (0.7,0) node[anchor=south]{$ q^{-4}$}node[anchor=north]{$ -4m$};
%	\draw (2.3,0) node[anchor=south]{$ q^{4}$} node[anchor=north]{$ -2+4m$};
%	\draw (3.5,0) node[anchor=north]{$ 1-2m$};
%	\end{tikzpicture}
%\end{center}

\begin{center}
\begin{tikzpicture}
\begin{scope}[yscale=-2.5,xscale=2.5]
	\cone{0}{0}{1.5}{$1^{+}$};
	\cone{1}{0}{1.5}{$1^{-}$};
	%\arr{1}{0}{0}{0}{$\psi_{0}$;}
	
	\cone{-1}{1}{1.5}{$\psi_{-1}1^{+}$};
	\cone{0}{1}{1.5}{$\psi_{-1}1^{-}$};
	\cone{1}{1}{1.5}{$\bpsi_{-1}1^{+}$};
	\cone{2}{1}{1.5}{$\bpsi_{-1}1^{-}$};

	\arrL{0}{1}{0}{0}{$L_{1}$}{0.5};
	\arrL{1}{0}{1}{1}{$L_{-1}$}{0.5};
	
	\arrLgd{-1}{1-0.01}{1}{1-0.01}{\textcolor{darkgreen}{$L_0$}}{0.7};
	\arrLgd{0}{1+0.01}{2}{1+0.01}{\textcolor{darkgreen}{$L_0$}}{0.7};
	\arrLg{0}{0}{2}{1}{\textcolor{darkgreen}{$L_{-1}$}}{0.7};
	\arrLg{-1}{1}{1}{0}{\textcolor{darkgreen}{$L_{1}$}}{0.7};
	 
	 \cone{-2}{3}{1.5}{$\psi_{-2}\psi_{-1}1^{+}$};
	 \cone{-1}{3}{1.5}{};
	 \cone{0}{3}{1.5}{};
	 \cone{1}{3}{1.5}{};
	 \cone{2}{3}{1.5}{};
	 \cone{3}{3}{1.5}{$\bpsi_{-2}\bpsi_{-1}1^{-}$};

    \arrL{-1}{3}{-1}{1}{$L_{2}$}{0.5};
    \arrL{0}{1}{0}{3}{$L_{-2}$}{0.5};
    \arrL{1}{3}{1}{1}{$L_{2}$}{0.5};
    \arrL{2}{1}{2}{3}{$L_{-2}$}{0.5};

	 \arrLg{-2}{3}{0}{1}{\textcolor{darkgreen}{$L_{2}$}}{0.7};
	 \arrLgd{-2}{3}{0}{3}{\textcolor{darkgreen}{$L_0$}}{0.6};
	 \arrLg{1}{1}{3}{3}{\textcolor{darkgreen}{$L_{-2}$}}{0.7};	 \arrLgd{1}{3}{3}{3}{\textcolor{darkgreen}{$L_0$}}{0.6};
\end{scope}
\end{tikzpicture}
\end{center}

The $\pi$-twisted modules are again induced from the two irreducible representations of the trivial zero-mode algebra

%\begin{align*}
%\{ \psi_{m+\frac{1}{2}+\Nil}, %\psi_{n+\frac{1}{2}+\Nil}\}&=\textcolor{darkgreen}{\delta_{m,-n} t} \\
%\{ \psi_{m+\frac{1}{2}+\Nil}, \bpsi_{n+\frac{1}{2}+\Nil}\}&=\delta_{m,-n} (m+\frac{1}{2}) \\
%\{ \bar{\psi}_{m+\frac{1}{2}+\Nil}, \bpsi_{n+\frac{1}{2}+\Nil}\}&=0 
%\end{align*} 
%and Virasoro action 
%\begin{align*}
%L_n
%&=\begin{cases}
%\sum_{j>0} \bpsi_{-j+\Nil}\psi_{j+\Nil}
%+\frac{1}{2}(\bpsi_{0+\Nil}\psi_{0+\Nil}-\psi_{0+\Nil}\bpsi_{0+\Nil})
%-\sum_{i>0} \psi_{-i+\Nil}\bpsi_{i+\Nil}-\frac{1}{8},\quad &\text{ for }n=0\\
%\sum_{i+j=n} \bpsi_{i+\Nil}\psi_{j+\Nil}&\text{ for }n\neq 0\\
%\end{cases}
%\end{align*}

\begin{center}
	\begin{tikzpicture}
	\begin{scope}[yscale=-2.5,xscale=2.5]
	
	\cone{0}{-0.125}{1.5}{$1^\pm$};
	
	\arrL{0}{-0.125}{0}{1.875}{$L_{2}$}{0.6};
	
	\cone{-1}{0.375}{1.5}{$\psi_{-\frac{1}{2}}1^\pm$};
	\cone{1}{0.375}{1.5}{$\bpsi_{-\frac{1}{2}}1^\pm$};
	\arrLgd{-1}{0.375}{1}{0.375}{\textcolor{darkgreen}{$L_0$}}{0.65};
	%\arrL{-1}{0.375}{1}{1.375}{\textcolor{darkgreen}{$L_{-1}$}}{0.6};
	%\arrL{-1}{1.375}{1}{0.375}{\textcolor{darkgreen}{$L_{1}$}}{0.6};
	
	\ncone{0}{0.875}{1.5}{$\psi_{-\frac{1}{2}}\bpsi_{-\frac{1}{2}}1^\pm$};
	
	\ncone{-1}{1.375}{1.5}{$\psi_{-\frac{3}{2}}1^\pm$};
	\ncone{1}{1.375}{1.5}{$\bpsi_{-\frac{3}{2}}1^\pm$};
	
	\cone{-2}{1.875}{1.5}{$\psi_{-\frac{3}{2}}\psi_{-\frac{1}{2}}1^\pm$};
	\cone{0}{1.875}{1.5}{};
	\ncone{0.1}{1.875}{1.5}{};
	\cone{2}{1.875}{1.5}{$\bpsi_{-\frac{3}{2}}\bpsi_{-\frac{1}{2}}1^\pm$};
	\arrLgd{-2}{1.875}{0}{1.875}{\textcolor{darkgreen}{$L_0$}}{0.5};
	\arrLg{-2}{1.875}{0}{-0.125}{\textcolor{darkgreen}{$L_2$}}{0.7};
	\arrLgd{0}{1.875}{2}{1.875}{\textcolor{darkgreen}{$L_0$}}{0.5};
	\arrLg{0}{-0.125}{2}{1.875}{\textcolor{darkgreen}{$L_2$}}{0.7};	
	\end{scope}
	\end{tikzpicture}
\end{center}

\begin{comment}
\begin{align*}
L_0\,\psi_{-\frac{1}{2}}1_{\sigma\Nil}
&=\left(\bpsi_{-\frac{1}{2}}\psi_{\frac{1}{2}}-\psi_{-\frac{1}{2}}\bpsi_{\frac{1}{2}}-\frac{1}{8}\right) \psi_{-\frac{1}{2}}1_{\sigma\Nil}\\
&=\textcolor{darkgreen}{t\bpsi_{-\frac{1}{2}}1_{\sigma\Nil}}+(\frac{1}{2}-\frac{1}{8}) \psi_{-\frac{1}{2}}1_{\sigma\Nil}\\
L_0\,\psi_{-\frac{1}{2}}1_{\sigma\Nil}
&=\left(\bpsi_{-\frac{1}{2}}\psi_{\frac{1}{2}}-\psi_{-\frac{1}{2}}\bpsi_{\frac{1}{2}}-\frac{1}{8}\right) \psi_{-\frac{1}{2}}1_{\sigma\Nil}\\
&=\textcolor{darkgreen}{t\bpsi_{-\frac{1}{2}}1_{\sigma\Nil}}+(\frac{1}{2}-\frac{1}{8}) \psi_{-\frac{1}{2}}1_{\sigma\Nil}\\
%
L_0\,\psi_{-\frac{3}{2}}\psi_{-\frac{1}{2}}1_{\sigma\Nil}
&=\left(\bpsi_{-\frac{3}{2}}\psi_{\frac{3}{2}}+\bpsi_{-\frac{1}{2}}\psi_{\frac{1}{2}}-\psi_{-\frac{1}{2}}\bpsi_{\frac{1}{2}}-\psi_{-\frac{3}{2}}\bpsi_{\frac{3}{2}}
-\frac{1}{8}\right)\psi_{-\frac{3}{2}}\psi_{-\frac{1}{2}}1_{\sigma\Nil}\\
&=\textcolor{darkgreen}{t\bpsi_{-\frac{3}{2}}\psi_{-\frac{1}{2}}1_{\sigma\Nil}
+t\psi_{-\frac{3}{2}}\bpsi_{-\frac{1}{2}}1_{\sigma\Nil}
}+(\frac{1}{2}+\frac{3}{2}-\frac{1}{8})\psi_{-\frac{3}{2}}\psi_{-\frac{1}{2}}1_{\sigma\Nil}\\
L_2\,\psi_{-\frac{3}{2}}\psi_{-\frac{1}{2}}1_{\sigma\Nil}
&=\left(\bpsi_{\frac{3}{2}}\psi_{\frac{1}{2}}+\bpsi_{\frac{1}{2}}\psi_{\frac{3}{2}}\right)\,\psi_{-\frac{3}{2}}\psi_{-\frac{1}{2}}1_{\sigma\Nil}\\
&=\textcolor{darkgreen}{(-\frac{3}{2}t-\frac{1}{2}t) 1_{\sigma\Nil}}
L_{-2}\,1_{\sigma\Nil}
&=
\end{align*}
\end{comment}

We now discuss the decomposition behaviour under restriction to the invariants $\V^{G}$, which is the Virasoro algebra, via the action of the centralizer $\mathrm{Cent}(g)$, which is the Borel part generated by $g$ itself, respectively by its Lie algebra generated by $x$.

\begin{lemma}\label{lm_tilting}
The kernel of $x$ on the $\pi$-untwisted module of parity $\pm1$ are Verma modules generated by $1^\pm$. More precisely the singular vectors are
$$1^\pm,\;\bpsi^g_{-1}1^\mp,\bpsi^g_{-2}\bpsi^g_{-1}1^\pm,\ldots$$
and the cokernel of $x$ are contragradient Verma modules with cosingular vectors 
$$1^\pm,\;\psi^g_{-1}1^\mp,\psi^g_{-2}\psi^g_{-1}1^\pm,\ldots$$
\end{lemma}
\begin{proof}
It is clear that the kernel resp. cokernel are generated by the respective vectors. Be induction and the formulas for $L_0$ we see that 
\begin{align*}
  L_n\,\psi^g_{-n}\psi^g_{-(n-1)}\cdots \psi^g_{-1}1^{\pm}
&=(\bpsi^g_{n}\psi^g_{0})\,\psi^g_{-n}\psi^g_{-(n-1)}\cdots \psi^g_{-1}1^{+}\\
&+\sum_{i=1}^{n-1}(\bpsi^g_{i}\psi^g_{n-i})\,\psi^g_{-n}\psi^g_{-(n-1)}\cdots \psi^g_{-1}1^{\pm}
\end{align*}
The sum has nonzero contribution but these cancel each other out
\begin{align*}
&\sum_{i=1}^{n-1}(\bpsi^g_{i}\psi^g_{n-i})\,\psi^g_{-n}\psi^g_{-(n-1)}\cdots \psi^g_{-1}1^{\pm}\\
&=\textcolor{darkgreen}{\sum_{i=n/2}^{n-1} (-i)(2t)(-1)^{i+(n-i)} \bpsi^g_{i}\,\psi^g_{-n}\psi^g_{-(n-1)}\cdots \cancel{\psi^g}_{-i} \cdots\cancel{\psi^g}_{-(n-i)} \cdots \psi^g_{-1}1^{+}}\\
&\textcolor{darkgreen}{+\sum_{i=1}^{n/2} (-i)(2t)(-1)^{(i-1)+(n-i)} \bpsi^g_{i}\,\psi^g_{-n}\psi^g_{-(n-1)}\cdots \cancel{\psi^g}_{-(n-i)}\cdots  \cancel{\psi^g}_{-i} \cdots \psi^g_{-1}1^{+}=0}
\end{align*}
which leaves  only the first term, which is depending on the sign
\begin{align*}
  L_n\,\psi^g_{-n}\psi^g_{-(n-1)}\cdots \psi^g_{-1}1^{-}
&=\hphantom{t}{(-n)(-1)^{n}\psi^g_{-(n-1)}\cdots \psi^g_{-1}1^{+}}\\
  L_n\,\psi^g_{-n}\psi^g_{-(n-1)}\cdots \psi^g_{-1}1^{+}
&=\textcolor{darkgreen}{t(-n)(-1)^{n}\psi^g_{-(n-1)}\cdots \psi^g_{-1}1^{-}}
\end{align*}
\end{proof}

We also note for later use that
\begin{align*}
    L_0\,\psi^g_{-n}\psi^g_{-(n-1)}\cdots \psi^g_{-1}1^{+}
&=-\sum_{k=1}^n (\psi^g_{-k}\bpsi^g_{k})\,\psi^g_{-n}\psi^g_{-(n-1)}\cdots \psi^g_{-1}1^{+}\\
&+\sum_{k=1}^n (\bpsi^g_{-k}\psi^g_{k})\,\psi^g_{-n}\psi^g_{-(n-1)}\cdots \psi^g_{-1}1^{+}\\
&=-\left(\sum_{k=1}^n (-k)\right)\psi^g_{-n}\psi^g_{-(n-1)}\cdots \psi^g_{-1}1^{+}\\
&+\textcolor{darkgreen}{t\sum_{k=1}^n \psi^g_{-n}\psi^g_{-(n-1)}\cdots \bpsi^g_{-k}\cdots  \psi^g_{-1}1^{+}}\\
\end{align*}

\begin{theorem}\label{thm_tilting}
    The modules deformed by a nilpotent element have a filtration by Verma modules and a filtration by Coverma modules given by kernel resp. cokernel of the long screening operator. 
\end{theorem}
\begin{proof}
    This follows from the respective computation in the previous Lemma \ref{lm_tilting} together with the known structure of the Fock module under the undeformed Virasoro action and long screening operator action. To this end, we rewrite everything in terms of (both undeformed and deformed) creation operators $\psi_{-j},\bpsi_{-j},j>0$ and undeformed annihilation operators $\frac{\partial}{\partial \bpsi_{-j}},\frac{\partial}{\partial \psi_{-j}},j>0$. Then the deformed Virasoro action is equal to the undeformed Virasoro action (which vanishes on highest weight vectors by definition) plus the additional terms $2t\bpsi_{-j-n}\frac{\partial}{\partial \psi_{-j}}$ that shift to a different Fock module. Clearly, these terms commute with the undeformed long screening operators, which are similar linear combinations of $\bpsi_{-j}\frac{\partial}{\partial \psi_{-j}}$. We proved in the previous Lemma \ref{lm_tilting} that the left-most irreducible modules combine to a Coverma module of the Virasoro algebra. Since the structure commutes with the undeformed action of the long screening operators, and it is known that the irreducible modules in the undeformed Fock modules form $\sl_2$ representations generated by the left-most irreducibles, this assertion also holds for all other irreducible modules, up to more-left terms. Hence we have a filtration in terms of Coverma modules, as was already suggested by the picture. The filtration in terms of Verma  modules is analogous.  
\end{proof}

{\bf Case 3}: If $r\not\in \Z$, then the zero-mode algebra is trivial, so there are two induced modules, each generated by a vector $v^\pm$ with fixed parity $\pi v^\pm=\pm v$. Since the $\pi$-twisted sectors  For $r\not\in \frac{1}{2}+\Z$ the $\pi$-twisted sectors have the same description, otherwise it it discussed in the previous bullet. Note that for $r\not\in \frac{1}{2}\Z$ we have a conjugate with $t=0$.

\subsection{Quantum groups}\label{sec_quantumgroups}

There should be a quantum group setup that reproduces the categorical situation, more precisely a $G^\vee$-crossed extension of the category of $u_q(\g)$-representations, see Question \ref{quest_quantumgroup}. We now compare the abelian category of $g$-twisted representation for $g\in B$ a fixed Borel subgroup  to the  Kac-Procesi-DeConcini quantum group $U^{\mathcal{K}}_q(\g)$ or unrestricted specialization with a big center, for which we refer to \cite{DKP92a} and \cite{CP94} Chapter 9.2 and  Chapter 11. 

For $\ord(q)=\ell$ odd and larger than all $(\alpha,\alpha)/2=d_\alpha$ there is a central subalgebra $Z_0\subset U^{\mathcal{K}}_q(\g)$, which contains in particular $E_i^\ell,F_i^\ell,K_i^\ell$ and which is isomorphic to the ring of functions on the Poisson dual Lie group. In particular, the representations of $U^{\mathcal{K}}_q(\g)$ form a graded tensor category with zero-fibre representations of $u_q(\g)$.  For $\ell$ even, the Lie group has to be replaced by its Langlands dual $G^\vee$, see \cite{Beck92, Tan13} for work in this direction.

\begin{question}
What is the abelian (module-) category associated to a conjugacy class $[g]$ of $G$ resp. $G^\vee$?  
\end{question}

It is expected that this is connected to the geometric properties of the conjugacy class $[g]$, for example an influential conjecture by DeConcini, Kac, Procesi states that the dimensions of the irreducible representations are divisible by $\ell^{\frac{1}{2}\dim([g])}$, where $\dim([g])$ is the dimension as a complex manifold, for example for regular elements of maximal dimension $|\varphi^+|$ it is true. Considerably less is known about the full structure as an abelian category, and to our knowledge there is no work in the case of even $\ell$ and the corresponding quantum group in \cite{CGR20,GLO18}, which is the one expected to correspond to the Feigin-Tipunin algebra. 

As a consequence of standard algebra techniques it follows, see  \cite{CP94} Theorem 11.1.1 and Lemma 11.1.3

\begin{theorem}
There is a proper closed subvariety $\mathcal{D}$ of the spectrum of the (full) center, such that outside $\mathcal{D}$ there is a unique irreducible representation with prescribed central character, which has dimension. $\ell^{|\varphi^+|}$ with $\varphi^+$ the set of positive roots.
\end{theorem}

We now concentrate on the case $\sl_2$. The full center consist of $Z_0$, generated by $E^\ell,F^\ell,K^\ell$ which we identify with the ring of functions on the big cell of $G=\mathrm{SL}_2$, together with Casimir element $C$. The conjugacy classes in $G$ consist of $\pm 1$, the regular semisimple elements and the regular unipotent element (possibly with sign). In this example \cite{CP94} Example 11.1.7 discuss all irreducible representations: There are cyclic representations $V(a,b,\lambda)$ for generic points of the spectrum of the full center and the familiar representations $V^\pm_k,k\leq \ell$ for the $\ell-1$ special points with central character $\chi$ 
$$\chi(E^\ell)=\chi(F^\ell)=0,\quad \chi(K^\ell)=1, \quad \chi(C)=\pm \frac{q^{r+1}-q^{-r-1}}{(q-q^{-1})^2}$$

We want to further derive the full abelian category structure of  the fibre over any point $(e,f,k)$ in the spectrum of $Z_0$, which is the  category of representation of   
$$u_q(\sl_2)^{(e,f,k)} := U^{\mathcal{K}}_q(\sl_2)/(E^\ell-e,\;F^\ell-f,K^\ell-k)$$
\begin{lemma}
The fibre $\Rep(u_q(\sl_2)^{(e,f,k)})$ over $(e,f,k)=(0,0,\pm 1)$ is the familiar nonsemisimple category of representations of $u_q(\g)$, the fibre over a regular semisimple element (conjugate to $(0,0,k),k\neq \pm 1$) is a semisimple category with $\ell$ representations of dimension $\ell$, and the fibre over a regular unipotent element (conjugate to $(1,0,\pm 1)$) is a nonsemisimple category with a unique simple representation of dimension $\ell$ and a projective cover with composition series of length $\ell$.
\end{lemma}
\begin{proof}
We use the results in \cite{CP94} on irreducible modules. We first have to analyze the fibers of the map of parameters 
\begin{align*}
    V(a,b,\lambda) &\longmapsto (e,f,k)=\left(a\prod_{j=1}^{\ell-1}\left(ab+\frac{(\lambda q^{1-j}-\lambda^{-1}q^{-1+j})(q^j-q^{-j})}{(q-q^{-1})^2}\right),\;b,\;\lambda^\ell\right) \\
    V^\pm_k &\longmapsto
    (e,f,k)=(0,0,1)
\end{align*}  
The preimage of $(e,f,k)\neq (0,0,1)$ consists generically of $\ell$ solutions $\lambda^\ell=k$ and for each $\ell$ solutions for the polynomial equation for $a$. If $f=0$ and hence $b=0$ this polynomial equation becomes a linear equation $e=aP$ with $P=\prod_{j=1}^{\ell-1}\frac{(\lambda q^{1-j}-\lambda^{-1}q^{-1+j})(q^j-q^{-j})}{(q-q^{-1})^2}$. The product $P$ is zero iff $\lambda=\pm q^{j-1}, j=1\ldots \ell-1$, which is the case for all but one solution $\lambda=q^{\ell-1}$ of $\lambda^\ell=\pm 1$. Hence for $k\neq \pm 1$ there are $\ell$ preimages of $(e,0,k)$, and they are  nonisomorphic, since there is a unique lowest-weight-vector. On the other hand for $k=\pm 1$ there is a unique preimage of  $(e,0,\pm 1)$ of $e\neq 0$. All other $(e,f,k)$ are $\mathrm{SL}_2$-conjugate to these cases. In particular the points  $(e,0,k),k\neq \pm 1$ are conjugate to $(0,0,k)$, in which case the preimages are simply Verma-modules with generic highest weights $\lambda$ with $\lambda^\ell=k$. \\

Now since $u_q(\sl_2)^{(e,f,k)})$ is an algebra of dimension $\ell^3$ it is clear from standard algebra arguments (recall the decomposition of the regular bimodule into products of irreducibles and their projective cover) that if there are $\ell$ irreducible modules of dimension $\ell$, then the category of representations is semisimple, and if there is a unique irreducible module, then its projective cover has a composition series of length $\ell$. 
\end{proof}

Since this result does not cover the case of even order root of unity and quantum group variant $\tilde{u}_q(\sl_2)$ relevant to this article, we also directly discuss the case $q^4=1$ and find it is in complete correspondence to our previous findings. 

\begin{example}[Case $\sl_2,p=2$]
     The quantum group $\tilde{u}_q(\sl_2)$ in \cite{CGR20} is as an algebra isomorphic to $u_q(\sl_2)$ for $\ord(q)=4$. We consider the  quantum group with a big center $U^{\mathcal{K}}_q(\sl_2)$ for $\ord(q)=4$, which has central elements $e:=E^2,f:=F^2, k=K^4$, which correspond to the Langlands dual group, and define again
     $$u_q(\sl_2)^{(e,f,k)} := U^{\mathcal{K}}_q(\sl_2)/(E^2-e,\;F^2-f,K^4-k)$$
     Since $K^2$ is also central, the fiber $\Rep(u_q(\sl_2)^{(e,f,k)}$ decomposes into a direct sum (which corresponds to $\pi$-untwisted and $\pi$-twisted sectors in the previous section). 
     
     Over $(0,0,k)$ all modules are highest-weight modules and hence quotients of baby Verma modules, which are irreducible for $k\neq 1$ and for $K^2=-1$. % $(EF-FE)v=\frac{K-K^{-1}}{q-q^{-1}}$ 
     In particular the fibres over $(0,0,k),k\neq 1$ are semisimple with two irreducible representations of dimension $2$.
     
     Over $(e,f,1)$ on the other hand we have $[E,F]=0$ resp. $[E,FK]_+=0$ and a fibre functor to super vector spaces via the action of $K$. Hence the fibre of $(e,f,1)$ are representations of the following algebra in the category of super vector spaces
     $$\langle E,FK \rangle / (E^2-e,F^2-f)$$
     In particular over $(e,0,1)$ we have a $2$-dimensional cyclic module with $F$ acting identically zero and a $2+2$-dimensional self-extension with $F$ acting nilpotently, which is the projective cover, just as the unipotent fibre in the previous section. Identifying $E,FK$ in the previous algebra description with $\psi_0^g,\bpsi_0^g$ shows directly that it is isomorphic to the zero-mode algebra, and hence the equivalence of abelian categories.
\end{example}

\subsection{Twisted free field realization}\label{sec_twistedFreeField}

The twisted representations of $\cW_{\sl_2,2}$ in the previous section, in particular the logarithmically twisted representations associated to nilpotent automorphisms, were found by explicitly computing the twisted mode algebra, and they rely on the description as symplectic fermions. We now give a construction of twisted module, which relies on the definition of $\cW_{\g,p}$ as subalgebra of a lattice vertex algebra $\V_{\sqrt{p}\Lambda^\vee}$ as kernel of a set of short screening operators. We cannot realize all twisted modules, in particular not the projectives, and we cannot prove that we can realize all irreducible twisted modules in this way, although we would conjecture so.  \\

Consider the lattice vertex algebra $\V_\Lambda$ for the rescaled coroot lattice $\Lambda=\sqrt{p}\Lambda_\g^\vee$. It has an action of the Borel subgroup $B^\vee\subset G^\vee$ on the by exponentials of semisimple resp. nilpotent elements  $$s^{\alpha^\vee}_0=\mathrm{Res}(\Y(s^{\alpha^\vee},z)),\qquad  x^{\alpha^\vee}_0=\mathrm{Res}(\Y(x^{\alpha^\vee},z))$$
(the latter the  so-called long screening operators) where $\alpha^\vee$ runs over the basis of coroots associated to the choice of $B$ and 
$$s^{\alpha^\vee}=\partial\varphi_{\alpha_i^\vee},\qquad x^{\alpha^\vee}=e^{\sqrt{p}\alpha^\vee}$$
an we have chosen a Virasoro action shifted by a background charge so that the $x^{\alpha^\vee}$ have conformal dimension $1$. If we assume $\g$ simply-laced, then this action is rigorously established in \cite{FT10,Sug21a}, but we expect the general proof to be similar.\\

Let $g\in G^\vee$. Up to conjugation, we can assume $g\in B^\vee$. We can now construct in general $g$-twisted modules of $\V_{\sqrt{p}\Lambda_\g^\vee}$, by using a method initiated in \cite{DLM96,Li97} in the semisimple setting and developed in \cite{FFHST02,AM09} for applications to short screening operators. To apply this method to the long screening operators, we collect the following properties:

\begin{itemize}
    \item The vertex algebra infinitesimal automorphisms (or derivations) $s^{\alpha^\vee}_0,x^{\alpha^\vee}_0$ are inner infinitesimal automorphism, meaning that we have obtained them as as $z^{-1}$-mode of vertex algebra elements $s^{\alpha^\vee},x^{\alpha^\vee}$. Hence, suitable mode operators $x_n,n\in\Z$ are available to us.  
    \item $s^{\alpha^\vee},x^{\alpha^\vee}$ are primary of conformal dimension $1$, meaning that 
    $$L_ns^{\alpha^\vee}=0,n>0,\qquad L_0s^{\alpha^\vee}=s^{\alpha^\vee}$$ 
    $$L_nx^{\alpha^\vee}=0,n>0,\qquad L_0x^{\alpha^\vee}=x^{\alpha^\vee}$$ 
    \item The mode operators $x_n^{\alpha^\vee}$ are self-commuting by the commutator formula
    \begin{align*}
    %[x(z^{q})^{\alpha^\vee},x(z^{p})^{\alpha^\vee}]
    %&=\sum_{l<0} {-q-1 \choose -l-1} \Y(x(z^{l})^{\alpha^\vee}x^{\alpha^\vee})(z^{p+q-l})\\ %q=-n-1, p=-m-1
    [x_n^{\alpha^\vee},x_m^{\alpha^\vee}]
    &=\sum_{k\geq 0} \binom{n}{k} \Y(x^{\alpha^\vee}_k x^{\alpha^\vee})_{m+n-k}\\%z^{-m-n+k-1}=0
    \end{align*}
    because $x_k^{\alpha^\vee}x^{\alpha^\vee}=0$ for $-k-1<|\sqrt{p}\alpha^\vee|$. The mode operators $s_n^{\alpha^\vee}$ span the Heisenberg algebra, so they are not self-commuting, but sufficiently controlled, see \cite{DLM96,Li97}.
\end{itemize}
Since twisted modules for semisimple elements are well-known, and on the other hand the last bullet requires modifications in this case, we now assume a unipotent element $g$ resp. a nilpotent derivation $x$, and study the logarithmically twisted modules as described. We now state the main result in \cite{Hu10} Section 5: 
\begin{theorem}\label{thm_Delta}
Assume that $\V$ is a vertex operator algebra and $x$ an element, which is a primary field of conformal dimension $1$ and which has self-commuting mode operators, as discussed above. Let $\cM$ be a (weak) $\V$-module, then the following defines on the vector space $\cM$ the structure of an (weak) $x_0$-twisted module $\widetilde{\cM}$:
$$\Y_{\widetilde{\cM}}(a,z)
:=\Y_{\cM}(\Delta(x)a,z),\qquad
\Delta(x,z)=\exp\left(x_0\log(z)+\sum_{n=1}^\infty \frac{x_n}{-n}(-z)^{-n}\right)
$$
\end{theorem}

Note that the source states this more generally for non-selfcommuting mode operators, on the other hand it is demanded in \cite{AM09} Theorem 2.1. and in this case we find the exponential more difficult to handle.\CommentsForMe{TC: Where does it exactly enter?} The defining property for $g=e^{2\pi\i x_0}$ holds by construction in the first term of $\Delta$
\begin{align*}
e^{2\pi\i\;z\frac{\partial}{\partial z}}
\Y_{\widetilde{\cM}}(a,z)
&=\Y_{{\cM}}(\exp\left(x_0(\log(z)+2\pi\i)+\sum_{n=1}^\infty \frac{x_n}{-n}(-z)^{-n}\right)a,z)\\
&=\Y_{\widetilde{\cM}}(e^{2\pi\i x_0}a,z)
\end{align*}
and we have 
\begin{align*}
\frac{\partial}{\partial z}
\Y_{\widetilde{\cM}}(a,z)
&=\frac{\partial}{\partial z}\Y_{{\cM}}(\exp\left(x_0\log(z)+\sum_{n=1}^\infty \frac{x_n}{-n}(-z)^{-n}\right)a,z)\\
&+\Y_{{\cM}}((x_0z^{-1}+\sum_{n=1}^\infty x_nz^{-n-1})\exp\left(x_0\log(z)+\sum_{n=1}^\infty \frac{x_n}{-n}(-z)^{-n}\right)a,z)\\
&=\Y_{{\cM}}((L_{-1}+\Y(x,z))(\exp\left(x_0\log(z)+\sum_{n=1}^\infty \frac{x_n}{-n}(-z)^{-n}\right)a,z)\\
&=\Y_{\tilde{\cM}}(L_{-1}a,z)
\end{align*}

Now we calculate the deformed Virasoro action in analogy to \cite{Hu10} Section 5 or \cite{AM09} Thm 2.2:
\begin{example}
Since $L_0x=x$ and $L_nx=0$ for $n>0$ we have by skew-symmetry
\begin{align*}
    \Y(x,z)L
    &=e^{zL_{-1}}\Y(L,-z)x
    &=e^{zL_{-1}}\left((-z)^{-2}x+(-z)^{-1}L_{-1}x+\sum_{n\leq-2} z^{-n-2}L_nx\right)
\end{align*}
in particular the singular terms are $x_0L=-L_{-1}x+L_{-1}x=0$ and $x_1L=x$, and further $x_1x_1L=0$. Hence 
\begin{align*}
    \Y_{\tilde{\cM}}(L,z)
    &=\Y_{\cM}(\Delta(x,z)L,z)
    =\Y_{\cM}(L+z^{-1}x,z) \\
    \tilde{L}_n
    &=L_n+x_n
\end{align*}
\end{example}

Next we consider a set of short screening operators $y^{-\beta}_0=\mathrm{Res}(\Y(y^{-\beta},z))$ with $y^{-\beta}=e^{-\alpha/\sqrt{p}}$ where $-\beta$ runs over a set of simple roots of the opposite Borel. It is known that $x^{\alpha^\vee}_0$ and $y^{-\beta}_0$ commute, due to the commutator formula: For $\alpha\neq \beta$ this follows from $\langle\alpha,-\beta\rangle\geq 0$, whence the sum over $k$ is empty, and for $\alpha=\beta$ by explicit calculations of the summands $k=1,2$ and cancellation. 

\begin{corollary}
The action of $B^\vee$ on $\V_{\Lambda}$ for $\Lambda=\sqrt{p}\Lambda^\vee$ by vertex algebra automorphism descends to an action of $B$ on the vertex subalgebra $\cW_{\g,p}$ defined as intersection of the kernel of all short screenings $y^{-\beta}_0$. 
\end{corollary}

\begin{corollary}
For $\lambda\in \Lambda^*/\Lambda$ and $\V_\lambda$ the simple module of $\V_\lambda$ generated by $e^\lambda$ and $x:=x^{\alpha^\vee}_0$, we obtain an  $x$-twisted module $\widetilde{\V}_\lambda$ by Theorem \ref{thm_Delta} and by restriction an $x$-twisted module of $\cW_{\g,p}$, which we call the $x$-twisted Verma module for weight $\lambda$.
\end{corollary}

We now discuss this more explicitly in the case $\cW_{\sl_2,2}$ discussed in the previous subsection. Recall that for a derivation $x\psi=2t\bpsi$ we had ($\pi$-untwisted) $x$-twisted modules generated by for every representation of the zero-mode algebra 
$$(\psi_{0}^g)^2=t,\quad (\bpsi_{0}^g)^2=0,\quad \{ \psi_{0}^g, \bpsi_{0}^g\}=0$$  
with compatible super vector space structure, more precisely a unique irreducible module with $\bpsi_0^g=0$ and a unique self-extension, 
and unique $\pi$-twisted modules, each decomposed into $\pi$-eigenspaces. We now ask: Which of these modules can be extended to a representation of $\V_{\sqrt{2}\Z}\supset \cW_{\sl_2,2}$, where in this free-field realization $\psi=e^{-\alpha/\sqrt{2}}$ and $\bpsi=L_{-1}e^{\alpha/\sqrt{2}}$. In fact $\V_{\sqrt{2}\Z}$ as a module over $\cW_{\sl_2,2}$ is an indecomposable extension with two composition factors generated by $1$ and $e^{\alpha/\sqrt{2}}$, so the question is, whether the module can be extended by an action of the second generator. Now because
$$\frac{\partial}{\partial z}\Y(e^{\alpha/\sqrt{2}},z)=
\Y(L_{-1}e^{\alpha/\sqrt{2}},z)
=\sum_{n\in\Z} \bpsi_n z^{-n-1}$$
this requires necessarily that $\bpsi_0=0$, because the derivation cannot produce a term $z^{-1}$. Conversely, under this assumption we can define easily
$$\Y(e^{\alpha/\sqrt{2}},z):=\sum_{n\in\Z}\frac{\bpsi_n}{-n}z^{-n}$$
This means that the unique irreducible $\pi$-untwisted and $\pi$-twisted $x$-twisted module of $\cW_{\sl_2,2}$ extends to a $x$-twisted module of $\V_\Lambda$, while the indecomposable  $\pi$-untwisted $x$-twisted module does not extend.\\

We can compare this explicitly to the twisted free field realization using $\Delta$, with
$$x=e^{\alpha\sqrt{2}},\qquad 
\psi=e^{-\alpha/\sqrt{2}},\qquad 
\bpsi=L_{-1}e^{\alpha/\sqrt{2}}$$
For example, since 
$$
x_0e^{-\alpha/\sqrt{2}}=\partial\varphi_{\alpha\sqrt{2}}e^{\alpha/\sqrt{2}}
=2L_{-1}e^{\alpha/\sqrt{2}},\qquad 
x_1e^{-\alpha/\sqrt{2}}=e^{\alpha/\sqrt{2}},\qquad 
x_ne^{-\alpha/\sqrt{2}}=0,n>1$$ 
and all successive actions with $x_n$ are zero, we get upon expanding the exponential the following explicit expression for $\psi^g(z)$
%\exp\left(x_0\log(z)\sum_{n=1}^\infty \frac{x_n}{-n}(-z)^{-n}\right)
\begin{align*}
    \Y_{\widetilde{\cM}}(e^{-\alpha/\sqrt{2}}, z) 
    &=\Y_\cM(\Delta(x,z)e^{-\alpha/\sqrt{2}},z) \\
    &=\Y_\cM((1+x_0\log(z)+x_1z^{-1})e^{-\alpha/\sqrt{2}},z)\\
    &=\Y_\cM(e^{-\alpha/\sqrt{2}},z)
    +z^{-1}\Y_\cM(e^{\alpha/\sqrt{2}},z)
    +2\log(z)\frac{\partial}{\partial z}\Y_\cM(e^{\alpha/\sqrt{2}},z)
\end{align*}
As intended, the logarithmic term has by construction the correct monodromy behavior \ref{twisted_monodromy}, if we recall again that in our case  $(x_0)^2=0$. The non-logarithmic operator has the additional term $z^{-1}\Y_\cM(e^{-\alpha/\sqrt{2}},z)$, which we claim corresponds to the additional (green) terms in the previous section. For example, we recover the representation of the zero-mode algebra on the groundstates:
\begin{align*}
 \psi_0 1^+
 &=
 \mathrm{Res}\left(\Y_{\tilde{\cM}}(-e^{\alpha/\sqrt{2}},z)\right)1
 =
\mathrm{Res}\left(\textcolor{darkgreen}{z^{-1}\Y_\cM(-e^{\alpha/\sqrt{2}},z)}\right)1=
\textcolor{darkgreen}{e^{\alpha/\sqrt{2}}}=\textcolor{darkgreen}{1^-} \\
\psi_0 1^-
&=
 \mathrm{Res}\left(\Y_{\tilde{\cM}}(-e^{\alpha/\sqrt{2}},z)\right)e^{\alpha/\sqrt{2}}
 =\mathrm{Res}\left(\Y_{{\cM}}(-e^{\alpha/\sqrt{2}},z)\right)e^{\alpha/\sqrt{2}}=
 1=1^+\\
 \bpsi_0 1^+
 &= \mathrm{Res}\left(\Y_{\tilde{\cM}}(e^{\alpha/\sqrt{2}},z)\right)1
 =\mathrm{Res}\left(\Y_{{\cM}}(e^{\alpha/\sqrt{2}},z)\right)1
=0\\
 \bpsi_0 1^-
 &= \mathrm{Res}\left(\Y_{\tilde{\cM}}(e^{\alpha/\sqrt{2}},z)\right)e^{\alpha/\sqrt{2}}
 =\mathrm{Res}\left(\Y_{{\cM}}(e^{\alpha/\sqrt{2}},z)\right)e^{\alpha/\sqrt{2}}
=0
\end{align*}
\begin{remark}
In \cite{FT10,Sug21a} (again, simply-laced) the modules of $\cW_{\g,p}$ where constructed by the cohomology of $\V_\Lambda\times_{B^\vee} G^\vee$, which is a vertex algebra bundle over $G^\vee/B^\vee$ with fiber $\V_\Lambda$. 
The fact that we find a single $g$-twisted module of every untwisted module of  $\V_\Lambda$, in conjunction with the results of the second part, indicates some connection between these approaches.
\end{remark}

\begin{question}
It would be nice to have explicit twisted intertwining operators for the tensor product.
\end{question}

\section{Semiclassical limits and Sturm-Liouville operators}

To speak about limits of vertex algebras for $\epsilon\to 0$ (or respectively $\epsilon\to\infty$), we need to consider a vertex algebra $\V_\epsilon$ over the field of rational functions $\C(\epsilon)$ and fix a choice of an integral form, which is $\C[\epsilon]$-submodule $\V_\epsilon^{\C[\epsilon]}$ such that $\V_\epsilon^{\C[\epsilon]}\otimes_{\C[\epsilon]}\C(\epsilon)\cong \V_\epsilon$ and $\V_\epsilon^{\C[\epsilon]}$ is closed under the vertex operator $\Y$. Differently spoken, we fix a $\C[\epsilon]$-basis such that the vertex operator  on $\V_\epsilon^{\C[\epsilon]}$ only consists of nonnegative (respectively nonpositive) powers of $\epsilon$, so that they are well defined in the limit $\epsilon\to 0$ resp. $\epsilon\to \infty$ as the $\epsilon^0$-term. This procedure has been introduced in \cite{CL15} under the name \emph{deformable family}. The semiclassical limit of a vertex algebra acquires the additional structure of a Poisson vertex algebra, as discussed e.g. in  \cite{FBZ04} Sec. 16 \\

\subsection{Limit of affine Lie algebras: Connections}\label{sec_LimitAff}

We first briefly review the semiclassical limit of affine Lie algebras: Let  $\hat{\g}_\kappa$ be an affine Lie algebra 
$$[a_m, b_n]=[a,b]_{m+n}+m\delta_{m,-n}\langle a,b\rangle\kappa$$
with $a,b\in \g$ and Killing form $\langle a,b\rangle$. Then with the choice of an integral form generated by $a_n\kappa^{-1}$ we have the semiclassical limit $\kappa\to \infty$, which due to
$$[a_m\kappa^{-1}, b_n\kappa^{-1}]=\kappa^{-1}\cdot [a,b]_{m+n}\kappa^{-1}+\kappa^{-2}\cdot m\delta_{m,-n}\langle a,b\rangle\kappa$$
becomes commutative and can be identified with the ring of functions on the space of $\g$-valued connections 
$$\d+A,\qquad A=\sum_{n}A_n z^{-n-1}$$
where $a_n^{\kappa+p}\kappa^{-1}$ corresponds to the function 
$$\d+A\;\longmapsto \langle A_n,a\rangle$$
Clearly, the limit of the vacuum module (the Verma module for weight zero) can be identifies with the ring of functions on the space of regular $\g$-valued connections. As stated in \cite{FBZ04} Section 16.3 this is in fact an isomorphism respecting the natural Poisson structure.\\

%There are different integral forms and hence deformable families we can consider: Take as generators of the Verma module $a_0,a_m\kappa^{-1},m>0$ for $a\in \g$. Then $a_0\kappa^{-1}$ acts by zero, and there is still a nontrivial action of $\g$. In this sense, the limit $\kappa\to \infty$ of $\bV^\kappa_\lambda(\g)$, with $h$ not scaling with $b$ (so $\lambda b^{-2}\to 0$), is a bundle over the ring of functions on the space of regular $\g$-connections, and the zero-fibre is the irreducible representation of $\g$ with highest weight $\lambda$.
%The limit of the direct sum 
%$$\bigoplus_{\lambda=0}^\infty \bV^\kappa_\lambda\to \bigoplus_{\lambda=0} \LieL_\lambda(\sl_2)$$
%can be identified with the ring of functions on the space of regular connections with a fixed solution, where $L_\lambda(\sl_2)$ are the polynomials of degree $\lambda$ in some basis of two solutions under the natural $\sl_2$-action 

%\begin{question}
%Presumably, under this identification the limit Poisson vertex algebra structure matches some natural Poisson structure on the bundle of solutions over the variety of regular connections.
%\end{question}

\subsection{The Virasoro algebra}\label{sec_Vir}

We now discuss in more detail the semiclassical limit of $\rW^\kappa(\sl_2)$, which is equal to the Virasoro algebra $\Vir^b$ for $b=\sqrt{-\kappa-2}$ defined by 
$$[L_m,L_n]=(m-n)L_{m+n}+\frac{c}{12}(m^3-m)\delta_{m,-n}$$
where we parameterize central charge and conformal weights ($L_0$-eigenvalues) by
$$c=1+6(b^{-1}+b)^2,\qquad h_P=\frac{(b^{-1}+b)^2}{4}-P^2$$
The Verma module $\V_{c,h}$ is the module with highest weight vector
$$L_0v=hv,\qquad L_nv=0,\; n>0$$
For generic $b^2\not\in\mathbb{Q}$ the Verma module is irreducible except for the following values $h_{m,n}$ for $m,n\in \mathbb{N}$, in which case there is a singular vector at level $mn$.
$$h_{m,n}=\frac{(b^{-1}+b)^2}{4}-\frac{(mb^{-1}+nb)^2}{4}
=-b^{-2}\frac{m^2-1}{4}-\frac{mn-1}{2}-b^{2}\frac{n^2-1}{4}
$$
In particular the vacuum module $h_{1,1}=0$ has a singular vector $L_{-1}v=0$. For later use we spell out
\begin{align*}
    h_{m,1}&=-b^{-2}\frac{m^2-1}{4}-\frac{m-1}{2}\\
    h_{1,n}&=-\frac{n-1}{2}-b^{2}\frac{n^2-1}{4}
\end{align*}
and note that in in the limit $b\to \infty$ we have $h_{m,1}\to -\frac{m-1}{2}$ and $h_{1,n}$ scales like $-b^{2}\frac{n^2-1}{4}$.\footnote{I physics literature, the former limit is called light, the latter heavy.}

It is proven in \cite{BSA88} that the singular vector for $h_{1,n}$, and after replacing $b\leftrightarrow b^{-1}$ also $h_{m,1}$, is given by
\begin{align}\label{formula_BSA}
S_{1,n}=(b^{-2})^n\sum_{n_1+\cdots n_k=n} \frac{(n-1)!}{\prod_{l=1}^{k-1}(\sum_{i=1} L_{n_i})(n-\sum_{i=1} L_{n_i})} L_{-n_1}b^{-2}\cdots L_{-n_2}b^{-2}
\end{align}
For example $S_{1,2}b^4=(L_{-1}b^{-2})^2+L_{-2}b^{-2}$ and $S_{2,1}b^{-4}=(L_{-1}b^2)^2+L_{-2}b^{2}$. 
We note that in the naive limit $b\to \infty$ we have for $h_{m,1}$ the singular vector $L_{-1}^n$ and for $h_{1,n}$ the singular vector $L_{-n}$. However, the result depends on the choice of integral form, see below.\\

The tensor structure on the representation category of the Virasoro vertex algebra has been obtained in \cite{CJORY21}.

\subsection{Limit of the Virasoro algebra}\label{sec_LimitVir}

There are several interesting choices of integral forms and hence deformable families. First, let us choose a system of generators ${\ell}_n=L_n b^{-2}$, then the limit of the Virasoro algebra $b\to \infty$ becomes commutative. It can be identified with the space of functions on the space of Sturm-Liouville operators
$$\frac{\d^2}{\d z^2}+q(z),\qquad q(z)=\sum_{n\in \Z} q_n z^{-n-2}$$
where ${\ell}_n$ corresponds to the function 
$$\frac{\d^2}{\d z^2}+q\;\longmapsto q_n$$

There is also a different integral form generated by $L_1,L_0,L_{-1}$ and ${\ell}_{n}$ else, whence ${\ell}_1={\ell}_0={\ell}_{-1}=0$. In this case $L_1,L_{0},L_{-1}$ generate as usual the Lie algebra $\sl_2$. The action on ${\ell}_{-n},n\geq 2$ can be read off the Virasoro algebra relations
\begin{align*}
[L_1,{\ell}_{n}]=(-n+1){\ell}_{n-1}
[L_0,{\ell}_{n}]=(-n){\ell}_{n}
[L_{-1},{\ell}_{n}]=(-n-1){\ell}_{n+1}
\end{align*}
so ${\ell}_{n},n\leq -2$ resp. $n\geq 2$ span an $\sl_2$ Verma module with lowest weight vector ${\ell}_{-2}$ resp. highest weight vector ${\ell}_2$.

\begin{remark}
Again, the limit acquires an additional Poisson structure, and the correspondence above relates this with a respective natural Poisson structure on the space of Sturm-Liouville operators. A similar principle relates the quantum Hamiltonian reduction of arbitrary $\cW_{\g,k}$ to higher differential operators, which is related to the KdV-hierarchy and also puts more context on the formula for the singular vectors $S_{1,n}$, see the lecture \cite{Zub91}.
\end{remark}

We give some examples:
\begin{example} The limit of the Verma module $\V_h$ at conformal weight $h=b^2a$ for a fixed value $a$ is the ring of functions on the space of  Sturm-Liouville operators with regular singularity with fixed singular coefficient $q_0=a$, because $\ell_0=L_0b^{-2}$ has eigenvalue $a$ by construction and on the other side $\ell_0$ corresponds to the function that assigns to a Sturm-Liouville operator the coefficient $q_0$. 
\end{example}
\begin{example}
    The limit of the vacuum representation $\varphi_{1,1}=\V_0/L_{-1}1$ is the ring $\C[{\ell}_{-2},{\ell}_{-3},\ldots]$ of functions on the space of  Sturm-Liouville operators that are regular in $z=0$.\\
    
    More generally, the limit of $\varphi_{1,n}$ is the ring of functions on the space of Liouville operators with ${\ell}_0=-\frac{n^2-1}{4}$ with diagonalizable monodromy. We discuss this in more depth in the next section.
\end{example}
\begin{example}
    The limit of $\varphi_{m,1}$ is, in view of ${\ell}_0\to 0$ and  $S_{m,1}=b^{4m}{\ell}_{-1}^n+\cdots$ the truncated polynomial ring
    $$\C[{\ell}_{-1},{\ell}_{-2}\ldots]/({\ell}_{-1}^n)$$
    viewed as a module over $\C[{\ell}_{-2},{\ell}_{-3},\ldots]$. Note that in contrast to the the previous case $\varphi_{1,n}$ now $L_0$ acts with a nonzero value in $-\frac{m-1}{2}+\mathbb{N}$. \\
    
    Let us now choose a different integral form on this module, generated by $L_{-1},{\ell}_{-2},{\ell}_{-3}\ldots$. Recall formula \eqref{formula_BSA} for $S_{m,1}$ and assume that the partitions $(n_1,\ldots,n_k)$ we sum over are ordered and let $l$ be the index such that $n_i=1$ iff $i\leq l$. Then the corresponding summand, rewritten in the chosen integral form, is 
    $$(L_{-n_1}b^2)\cdots (L_{-n_k}b^2)
    =b^{2l+4(k-l)} L_{-1}^{l}{\ell}_{-n_{k-l}}$$
    and the highest scaling $b^{2m}$ is attained precisely for the partition $(1,\ldots 1,2,\ldots 2)$. Hence this limit of $\varphi_{m,1}$ is 
        $$\C[L_{-1},{\ell}_{-2}\ldots]/S_{m,1}$$
    with $S_{m,1}$ a certain polynomial in $L_{-1},{\ell}_{-2}$ starting with $L_{-1}^m$.  By acting with $L_{-1},L_0,L_1$ there seems to be an $\sl_2$-action on this space.
\end{example}

\begin{question}~
\begin{itemize}
    \item 
    Presumably, the $\sl_2$ action is the action of the Möbius transformations on the space of Sturm-Liouville operators. Note that the trivial operator $\frac{\d^2}{\d z^2}$ where all ${\ell}_{n}=0$ is invariant under $\sl_2$, otherwise we have a  conjugated $\sl_2$-action for regular $q(z)$.
    \item Presumably, the limit of the direct sum $\bigoplus_{m\geq 1}\varphi_{m,1}$ is the ring of functions on the space of regular Sturm-Liouville operators together with a fixed solution $f_1$. Presumably, the $\sl_2$-action on the space of solutions corresponds to the $\sl_2$-action by $L_{-1},L_0,L_1$. Is there also a natural matching Poisson structure?
    \item What is the respective statement for $\varphi_{m,n}$? The most naive guess would be that this is simply a combination of the cases $\varphi_{m,1}$ and $\varphi_{1,n}$. This is, that the direct sum $\bigoplus_{m\geq 1}\varphi_{m,n}$ is the ring of functions on the space of Sturm-Liouville operators with $q_0=-\frac{n^2-1}{4}$ and diagonal monodromy, as discussed below for $\varphi_{1,n}$, together with a fixed solution $f_1$.
\CommentsForMe{Do $2,n$ explicitly}
%EXPLICIT FORMULA FOR VIRASORO SINGULAR VECTOR DMITRY V. MILLIONSHCHIKOV}
    \item Similarly, one could presume that the semiclassical limit of the triplet algebra $\bigoplus \varphi_{m,1}\otimes \C^{m}$ is the ring of functions on the space of regular Sturm-Liouville operators together with a fixed basis of solutions $f_1,f_2$, and the according $\sl_2$-action. 
    \item One could also consider an integral form with ${\ell}_0,L_{-1}$ and $L_0,{\ell}_{-1}$ and then only action of part of $\sl_2$. What does this correspond to?
    \item Similarly, for the $W$-algebra resp. the affine Lie algebra for arbitrary $\g$ there should be a modified integral form on the Verma module, such that the limit is the bundle given by opers resp. connections with a fixed solution, together with an action of $\g$.
\end{itemize}
\end{question}

\subsection{Sturm-Liouville operators}\label{sec_SturmLiouville}

The structure of the representations of the Virasoro algebra deeply reflects the solution properties of the associated differential equation, for example its monodromy. Let us briefly review some aspects of Sturm-Liouville differential equations:

$$\left(\frac{\d^2}{\d z^2}+q(z)\right)f(z)=0$$

The Frobenius method makes a power series ansatz $f(z)=\sum_{n\in\Z} a_nz^n$ around $z=0$. If $q(z)$ is regular then $z=0$ is called a regular point. If $q(z)z^2$ is regular, say $q(z)=\frac{{\ell}_0}{z^2}+\cdots$, then $z=0$ is called a regular singular point. In this case, we can start by solving the differential equation for the lowest term 
$$\left(\frac{\d^2}{\d z^2}+\frac{{\ell}_0}{z^2}\right)f_0(z)=0,\qquad f_0(z)=z^s,\qquad s(s-1)+{\ell}_0=0$$
and obtain a recursive formula for the power series coefficients of a solution $f(z)=\sum_{n=0}^\infty a_nz^{n+s}$. The equation $s(s-1)+{\ell}_0=0$ is called indical equation. Indeed, the theorem of Fuchs states, that the solutions of a Sturm-Liouville differential equation given by a convergent Frobenius series as above, or one of the solutions is of second kind:
$$f_2=f_1\log(x)+\sum_{n=0}^\infty a_n z^{n+s}$$
where $f_1$ is a solution of the first kind. \\

To understand more precisely, when this can happen, we consider the monodromy of the $2$-dimensional space of solutions around $z=0$. Suppose the two solutions of the indicial equation $s_1,s_2$ differ not by an integer, then the monodromy is diagonalizable 
$$\begin{pmatrix} e^{2\pi\i s_1} & 0 \\ 0 &  e^{2\pi\i s_2} \end{pmatrix}$$
in a basis of two distinguished solutions $f_1,f_2$ with Frobenius series for $s_1,s_2$. If $s_1-s_2\in\Z$, it may happen that the monodromy matrix is a Jordan block
$$\begin{pmatrix} e^{2\pi\i s_1} & 2\pi\i \\ 0 &  e^{2\pi\i s_2} \end{pmatrix}$$
with a solution $f_1$ with a Frobenius series for $s_1$ and a solution $f_2$ of the second type. Indeed, the classification theorem states that Sturm-Liouville operators with regular singular points are up to coordinate transformations  classified by their monodromy operator, see e.g. \cite{Ov01} Section 4.2.\\

From the perspective of the Virasoro algebra, ${\ell}_0$ is the rescaled conformal dimension, and the indicial equations $s(s-1)+{\ell}_0=0$ with $s_1-s_2=n$ precisely correspond to the values ${\ell}_0=-\frac{n^2-1}{4}$ of the series $\varphi_{1,n}$. A striking observation (say, in the next examples) is that the singular vector $S_{1,n}$, as a polynomial equation in ${\ell}_{-k}$, seems to predict precisely those coefficients that lead to diagonalizable monodromy.  In this sense, the limit of $\varphi_{1,n}$ could be identified with the ring of functions on the space of diagonalizable Sturm-Liouville operators with $s_1-s_2=n$.

\begin{question}
We could not find sources in literature for the statements in the previous paragraph, although we would strongly assume that this behavior is known to the experts. We would be happy for any reference, otherwise this is an interesting question to settle.  
\end{question}

For Sturm-Liouville differential equations with irregular singularity, one observes the Stokes phenomenon that the vicinity of the singularity consists of angular sectors, called Stokes sectors, each with a different asymptotic expansion, see for example \cite{Boa14}. It would be very interesting to continue the present discussion in this context, see Question~\ref{quest_irregular}.\\

\begin{example}
The Euler equation 
$$\left(\frac{\d^2}{\d z^2}+\frac{a}{z^2}+\frac{0}{z}\right)f(z)=0$$
is solved by $z^{s}$ for the two solutions $s(s-1)+a=0$, meaning
$s=\frac{1\pm\sqrt{1-4a}}{2}$. For $\varphi_{1,2}$ we have ${\ell}_0=-\frac{3}{4}$ and thus $s=\frac{3}{2},-\frac{1}{2}$. In particular the monodromy is 
$$\begin{pmatrix} e^{2\pi\i(\frac{3}{2})} & 0 \\ 0 & e^{2\pi\i(-\frac{1}{2})}  \end{pmatrix}=\begin{pmatrix} -1 & 0 \\ 0 & -1 \end{pmatrix}$$
This diagonalizability is consistent with ${\ell}_{-1}=0$ being the singular vector in $\varphi_{1,1}$.
\end{example}

\begin{example}
The Bessel differential equation 
$$ \left(z^2 \frac{\d^2 z}{\d z^2} + z \frac{\d}{\d x} + (z^2 - \alpha^2 )\right)F_\alpha(z)=0$$
becomes under substitution $F_\alpha(z)=z^{-1/2}f_\alpha(z)$ 
$$\left(\frac{\d^2 z}{\d z^2} + \frac{1-4\alpha^2}{4z^2}+\frac{0}{z}+1\right)F_\alpha(z)$$
with a regular singularity at $z=0$, solved by the Bessel function $J_\alpha(z)$. Typically $J_{\pm \alpha}(z)$ are two linearly independent solutions. However, if $\alpha$ is an integer than some coefficient vanish and we have $J_{-\alpha}(z)=(-1)^\alpha J_\alpha(z)$. Then, another linearly independent solution can be obtained by a limit procedure, called the Bessel function of the second kind
\begin{align*}
    Y_n(z) =
\frac{2}{\pi} J_n(z) \ln \frac{z}{2}
&-\frac{\left(\frac{z}{2}\right)^{-n}}{\pi}\sum_{k=0}^{n-1} \frac{(n-k-1)!}{k!}\left(\frac{z^2}{4}\right)^k \\ &-\frac{\left(\frac{z}{2}\right)^n}{\pi}\sum_{k=0}^\infty (\psi(k+1)+\psi(n+k+1)) \frac{\left(-\frac{z^2}{4}\right)^k}{k!(n+k)!}
\end{align*}
We now discuss this from our perspective: The indicial equation predicts solutions with exponent $s=\frac{1}{2}\pm \alpha$, which is consist with the behavior $z^{\pm \alpha}$ for the Bessel functions $J_\alpha$ and $J_{-\alpha}$ at the singularity. For $\alpha$ half-integer, the exponents differ by an integer. Accordingly, the value ${\ell}_0=\frac{1-4\alpha^2}{4}$ matches ${\ell}_0=-\frac{n^2-1}{4}$ for $n=2\alpha$. Moreover 
$${\ell}_{-1}=0, \quad {\ell}_{-2}=1,\quad {\ell}_{-n}=0,n>2$$
Thus the only contribution to $S_{1,n}$ is $(n_1,n_2,...)=(2,2,...)$ and $S_{1,n}=0$ iff $n$ odd. Accordingly the monodromies are for $n=2\alpha$ odd resp. even are in the basis $J_\alpha,$
$$\begin{pmatrix} 1 & 0 \\ 0 & 1  \end{pmatrix},\qquad \begin{pmatrix} -1 & 4\i \\ 0 & -1 \end{pmatrix}$$
According to their appearance in applications, the first type are called spherical Bessel functions and the second type cylindrical harmonics.    
\end{example}

\begin{example}
The hypergeometric differential equation 
$$\left(z(1-z)\frac {\d^2}{\d z^2} + (c-(a+b+1)z) \frac {\d}{\d z} - ab\right)F(z) = 0$$
has regular singular points at $0,1,\infty$, and it is up to coordinate transformation the unique differential equation with three regular singular points. It can be transformed to the differential equation 
$$\left(\frac{\d^2}{\d z^2}+q(z)\right)f(z)=0$$
\begin{align*}
q(z)&=\frac{z^2(1-(a-b)^2) +z(2c(a+b-1)-4ab)+c(2-c)}{4z^2(1-z)^2}\\
&=\frac{c(2-c)}{4}z^{-2}
+\frac{(2c(a+b-1)-4ab)+2c(2-c)}{4}z^{-1}+\cdots \\
%&+\frac{(1-(a-b)^2))+2(2c(a+b-1)-4ab)+3c(2-c)}{4}z^{0}
%+\cdots\\
f(z)&=z^{c/2}(1-z)^{-(c-a-b-1)/2}\cdot F(z)
\end{align*}
For $c\not\in\Z$ the hypergeometric equation has two solution
$$ {_2}F_1(a,b,c;z), \qquad 
z^{1-c} \, {_2}F_1(1+a-c,1+b-c;2-c;z)
$$
leading to solutions with exponents $s=c/2$ and $s=1-c/2$ at $z=0$. For integer $c$, either one of the solutions is a rational function, or a hypergeometric function with an additional logarithms and a digamma function. This phenomenon is reviewed for example in  \cite{Vid07} and \cite{DK17}, which also nicely discusses this effect on a toy model ${_0}F_{1}$ with one regular singular point.\\

The case of positive integer $c$ (resp. $2-c$) corresponds by  $c=n+1$ we have $\frac{c(2-c)}{4}=-\frac{n^2-1}{4}$ to $\varphi_{1,n}$. \\

For example, $c=0$ corresponds to $\varphi_{1,1}$ and $z^{-2}$-term in $q(z)$ disappears. The relation ${\ell}_{-1}$ means that the $z^{-1}$-term in $q(z)$ also disappears and $b=0$ resp. $a=0$. In the hypergeometric equation in this case the third term vanishes and 
$$\left(z(1-z)\frac {\d^2}{\d z^2} -(a+1)z \frac {\d}{\d z}\right)F(z) = 0$$
and after canceling $z$ and substituting $F'$ we find solutions
$$F=1,\quad (1-z)^{a+2},\qquad \text{resp.}\qquad F=1,\quad \log(1-z)$$
with the respective diagonalizable monodromy with eigenvalues $+1$

On the other hand the hypergeometric function for $a+b=1$ is related to the Legendre function $P_{-a}^{1-c}(1-2z)$. For integer $c$ the solutions are Legendre polynomials, together with the Legendre functions of the second kind, which involve logarithms. 
\end{example}

We finally discuss the action by coordinate transformation and an $\sl_2$-action:

If $f_1,f_2$ are solutions of the differential equation, then $f_1/f_2$ has Schwarzian derivative $2q(z)$. Here, the Schwarzian derivative $\mathcal{S}$ has the following definition and defining property
$$(\mathcal{S}f)(z)  = \left( \frac{f''(z)}{f'(z)}\right)'  - \frac{1}{2}\left(\frac{f''(z)}{f'(z)}\right)^2 
 = \frac{f'''(z)}{f'(z)}-\frac{3}{2}\left(\frac{f''(z)}{f'(z)}\right)^2$$
$$\mathcal{S}(f \circ g) = \left( \mathcal{S}(f)\circ g\right ) \cdot(g')^2+\mathcal{S}(g)$$
and is zero for Möbius transformations. We have the following action of precomposing with a diffeomorphism, 
$$\frac{\d^2}{\d z^2} + f^\prime(z)^2 \,(q\circ f)(z) + \tfrac{1}{2} \mathcal{S}(f)(z)$$
and the following remarks:
\begin{itemize}
    \item The space of Sturm-Liouville operators can be interpreted as the coadjoint representation of the Virasoro algebra, see 
    \cite{Ov01}. The Schwartz derivative appears here a $1$-cocycle for the Virasoro algebra.
    \item From the analytic perspective, this transformation formula is the canonical transformation behavior of a linear map between spaces of tensor densities, see \cite{Ov01} Section 1.7. and \cite{FBZ04} Section 8.2 and 16.5.
    $$\mathcal{F}_{-1/2}\to \mathcal{F}_{3/2}$$
    \item From a conformal field theory perspective, this transformation behavior with Schwarzian derivative is precisely the transformation behavior of the Virasoro field $\Y(T,z)$ at nonzero or zero central charge.
    \item The functions $f_1^2, f_1f_2,f_2^2$, which also appear as a basis of $\varphi_{3,1}$, are solutions to a third order differential equation, this should be viewed in the sense of the KdV hierarchy, see \cite{Ov01,Zub91}
\end{itemize}

Suppose $f_1,f_2$ are solutions of a Sturm-Liouville problem. Then the transformation to the coordinate $t=f_1/f_2$ brings the differential equation to the trivial form $\frac{\d^2}{\d t^2} f(t)=0$. \\

One has for the trivial form a Lie algebra $\sl_2$ via $\frac{\d}{\d z},z\frac{\d}{\d z},z^2\frac{\d}{\d z}$. Following Kirillov, we see that via the coordinate transformation, this leads to an action of $\sl_2$ by vector fields 

$$f_1^2 \frac{\d}{\d z},\quad  f_1f_2\frac{\d}{\d z},\quad f_2^2\frac{\d}{\d z}$$

We check the commutator relations
\begin{align*}
    [f_1f_2 \frac{\d}{\d z},f_1^2 \frac{\d}{\d z}]
    &=2f_1^2f_1'f_2\frac{\d}{\d z}
    -f_1^2f_1'f_2\frac{\d}{\d z}
    -f_1^3f_2'\frac{\d}{\d z} =W\cdot f_1^2 \frac{\d}{\d z}\\
    [f_1f_2 \frac{\d}{\d z},f_2^2 \frac{\d}{\d z}]
    &=2f_2^2f_2'f_1\frac{\d}{\d z}
    -f_2^2f_2'f_1\frac{\d}{\d z}
    -f_2^3f_1'\frac{\d}{\d z} =-W\cdot f_1^2 \frac{\d}{\d z}\\
    [f_1^2 \frac{\d}{\d z},f_2^2 \frac{\d}{\d z}]
    &=2f_1f_1'f_2^2\frac{\d}{\d z}
    -2f_1^2f_2'f_2\frac{\d}{\d z}=2W\cdot f_1f_2 \frac{\d}{\d z}
\end{align*}
where the Wronski determinant $W=f_1'f_2-f_1f_2'$ is constant and nonzero. \CommentsForMe{We also check that these vector fields leave the space of solutions invariant}
%\begin{align*}
%f_1^2\frac{\d}{\d z}f_1=....0?
%\left(\frac{\d^2}{\d z^2}+q(z)}\right)f_1^2 \frac{\d}{\d z}f_1
%&=(f_1^2f_1')''+(-f_1'')f_1^2f_1'
%\mathrm{ad}_{f_1^2}(\frac{\d^2}{\d z^2}+q(z))
%&=f_1^2q'+2(f_1^2)'q-(f_1^2)'''\\
%&=f_1^2(-f_1''')+4f_1'f_1(-f_1'')-2f_1f_1'''(-f_1'')-3f_1'f_1''
%directly....?
%[ \frac{\d^2}{\d z^2}+q(z),f_1f_2 \frac{\d}{\d z}]
%&=
%2(f_1'f_2+f_1f_2')\frac{\d^2}{\d z^2}
%+f_1''f_2+f_1'f_2'+f_1f_2''-f_1f_2 q(z)'
%....
%\end{align*}
Since, at least for regular $q(z)$, we may assume after a coordinate change a trivial Sturm-Liouville operator, it suffices to study $\sl_2$-action in this case
$$\frac{\d}{\d z},\quad z\frac{\d}{\d z},\quad z^2\frac{\d}{\d z} $$
For singular $q(z)$ such a coordinate change is no longer holomorphic and we do not get an action of $\sl_2$ by regular vector fields. 

\section{Vertex algebras with big center}

Denote by $\rV^\kappa(\g)$ the universal affine vertex algebra at level $\kappa$ and by $\bM^\kappa_\lambda$ the induced  module at dominant weight $\lambda$, by $\bW^\kappa_\lambda$ the Wakimoto module, and by $\bV^\kappa_{\lambda}$ their irreducible quotient for integral weight $\lambda$. Denote by  $\bL_{\kappa^*}(\g)$ the irreducible affine vertex algebra for integral level ${\kappa^*}$ and by $\bL^{\kappa^*}_\alpha$ the irreducible quotients modules $\bL^\kappa_\lambda$ for integral weight $\lambda$. For $\g=\sl_2,{\kappa^*}=1$ we have $\bL^1(\sl_2)=\V_{\sqrt{2}\Z}$ a lattice vertex algebra and $\bL^1_0=\V_0,\;\bL^1_{1/2}=\V_{1/2}$ their irreducible modules.

Denote by $\cW^\kappa(\g)$ the quantum Hamiltonian reduction by a principal nilpotent element $f$, and $\bV_{\lambda,f}$ and more generally $T_{\lambda,\mu}^\kappa$ the irreducible module obtained by reduction resp. $\mu$-twisted reduction of $\bV^\kappa_\lambda$, for $\lambda$ a dominant weight and $\mu$ a dominant coweight \cite{AF19}. For $\g=\sl_2$ we have $\bV_{\lambda,f}=\varphi_{m,1}$ and $T_{\lambda,\mu}=\varphi_{m,n}$ with Feigin-Frenkel duality $\varphi^b_{m,n}=\varphi^{b^{-1}}_{n,m}$. Denote by $\pi^\kappa_\lambda$ the Fock module, which is the reduction of the Wakimoto module and for generic $\lambda$ coincides with the Verma module.\\

\subsection{Coupled vertex algebras}\label{sec_couple}

Predecessors of the following general construction were the explicit constructions for $\sl_2,p=1$ in \cite{BFL16} and $\sl_2,p=2$ in \cite{BBFLT13} and $\g,p=1$ in \cite{ACF22}, whose conventions we use. On the other hand the construction builds for $p=0$ on the algebra of chiral differential operators $\mathcal{D}_{G,\kappa}$, which can be seen as a chiralization of the space of regular function on a finite group or Lie group.

The following construction first appeared in \cite{CG17} from physical considerations, on a category level it has been proven in \cite{Mor21} using quantum groups and Kazhdan-Lusztig correspondence 
\begin{theorem}\label{thm_couple}
For each integer $p>0$ divisible by the lacing number $m$ there exists a  family of vertex algebras, which extends $\rV^\kappa(\g)\otimes \rV^{{\kappa^*}}(\g)$ and decomposes over it as follows
$$\A^{(p)}[\g,\kappa]
=\bigoplus_{\lambda\in Q^+} \bV_\lambda^\kappa(\g)\otimes \bV_{\lambda^*}^{{\kappa^*}}(\g) $$
for a pair of generic levels $\kappa,{\kappa^*}$ satisfying 
$$\frac{1}{\kappa+h^\vee}+\frac{1}{{\kappa^*}+(h^\vee)^\vee}=p$$
where $(h^\vee)^\vee$ denotes the dual Coxeter number of the Langlands dual root system. 
\end{theorem}
It is a major open question if this has all necessary properties of a deformable family. For the quantum Hamiltonian reduction by a principal element on at least one side, this is the first result in \cite{CN22}:   
\begin{corollary}
By quantum Hamiltonian reductions on either or both sides, with respect to the principal nilpotent orbit, we obtain deformable families of vertex algebras  extending $\mathrm{V}^\kappa\otimes \rW^{{\kappa^*}}$
$$\HA^{(p)}[\g,\kappa]=\bigoplus_{\lambda\in Q^+} \bV_\lambda^\kappa(\g)\otimes T_{\lambda^*,0}^{{\kappa^*}}(\g)$$
(or sides reversed) and deformable families extending $\rW^\kappa\otimes \rW^{{\kappa^*}}$
$$\HHA^{(p)}[\g,\kappa]=\bigoplus_{\lambda\in Q^+} T_{\lambda,0}^{\kappa}(\g)\otimes T_{\lambda^*,0}^{{\kappa^*}}(\g)$$
\end{corollary}

%\begin{question}~
%\begin{enumerate}
%    \item In these works also the case of admisible level is discussed.
%    \item What is the formula for Verma modules at nongeneric weight? Here we don't have a direct sum, but a massive tower of extensions, as our result in the next section of taking covariants on the left-hand side shows in comparisment to the twisted modules.   
%\end{enumerate}
%\end{question}

\subsection{Limits of these vertex algebras}\label{sec_limit}

The limits for $\kappa\to\infty$ or ${\kappa^*}\to \infty$ are, at least as modules, easily read off the defining decomposition formulas. They are by Section \ref{sec_LimitAff} resp. Section \ref{sec_LimitVir} bundles over the regular $\g$-connections resp. regular $\g$-opers with zero-fibres
$$\bigoplus_{\lambda\in Q^+} \bV_{\lambda}^{\kappa}(\g) \otimes \LieL_{\lambda^*}(\g),\qquad 
\bigoplus_{\lambda\in Q^+} T_{\lambda,0}^{\kappa}(\g) \otimes \LieL_{\lambda^*}(\g)
$$ 
The quantum Hamiltonian reduction on the factor where the limit is taken are after the limit presumably related by a classical Hamiltonian reduction. \\

A conjecture communicated to us in the case $\sl_2$ by T. Creutzig, and raised in this generality in his joint work \cite{CN22} is

\begin{conjecture}\label{conj_CN}
The zero-fibre $\bigoplus_{\lambda\in Q^+} T_{\lambda,0}^{\kappa} \otimes \LieL_{\lambda^*}(\g)$ is the Feigin-Tipunin algebra $\cW_p(\g)$. Note that this is true for $p=1$ and $\g$ simply laced, and also in type $C$ using the large level limit of $\mathfrak{osp}(1|2n)$. 
\end{conjecture}

Note that the divisibility $m\mid p$ is precisely the condition we found in the non-simply-laced case in \cite{FL18}.

\begin{question}\label{quest_AffineTriplet}
The zero-fibre $\bigoplus_{\lambda\in Q^+} \bV_{\lambda}^{\kappa} \otimes \LieL_{\lambda^*}(\g)$, whose quantum Hamiltonian reduction is presumably the Feigin-Tipunin algebra, seems to be a very interesting vertex algebra. What is known about it? Is there a free-field realization using nonlocal screening operators? For $\sl_2$ this has a nice recent answer \cite{CNS24}. 
\end{question}

\section{The case of simply-laced \texorpdfstring{$\g$}{g} at \texorpdfstring{$p=1$}{p=1}}

For $p=1$ the algebra in the previous section has an explicit realization, which appeared in physics literature first for $\sl_2$ under the name GKO-construction \cite{GKO86} and was established in \cite{ACL19} for $\g$ simply-laced, see also \cite{ACF22}. %\cite{CN22} Theorem 1.3. and references therein:

\begin{theorem}
Let $\g$ be simply-laced, then we have the deformable family of vertex algebras extending $\rV^\kappa(\g)\otimes \rW^{{\kappa^*}}(\g)$
$$\rV^{\kappa-1}(\g)\otimes \bL^1(\g)=
\bigoplus_{\lambda\in Q^+} \bV_\lambda^\kappa\otimes T_{\lambda^*,0}^{{\kappa^*}},\qquad
\frac{1}{\kappa+h^\vee}+\frac{1}{{\kappa^*}+(h^\vee)^\vee}= 1$$
More precisely, the full commutant of $\rV^\kappa(\g)$ is $\rW^{{\kappa^*}}(\g)$. 
\end{theorem}

\begin{remark}
For non-simply-laced cases, we would have to consider as smallest case $p=m$ in order to have a well-defined Feigin-Tipunin algebra \cite{FL18}, corresponding to the largest $p$ for which the quantum group $u_q(\g)$ collapses.  
\end{remark}

\subsection{The affine coset}\label{sec_coset}
The explicit construction is based on the diagonal embedding of affine Lie algebras
$$\hat{\g}_{\kappa}\to \hat{\g}_{\kappa-1} \times \hat{\g}_1$$ 
$$a_n^{(\kappa)}\mapsto a_n^{(\kappa-1)}+a_n^{(1)}$$
and the commuting Virasoro embedding based on the three Sugawara elements $$L_n^{Coset}
=L_n^{(\kappa-1)}+L_n^{(1)}-L_n^{(\kappa)}$$
to which the $\rW^{{\kappa^*}}(\g)$ reduces to in the case $\sl_2$. The overall deformed action of $\rW^{{\kappa^*}}(\g)$ could be calculated similarly.

\begin{theorem} \label{thm_coset}
The limit $\kappa\to\infty$ of the affine Lie algebra $\hat{\g}_{\kappa-1} \times \hat{\g}_1$ is a bundle over the  space of $\g$-connections.
The fibers over a point $\d+A$, or the coinvariants with respect to $a_n^{(\kappa)}/\kappa-\langle A_n,a\rangle$ are isomorphic as a vector space to $\hat{\g}_1$ and $L^{Coset}$ descends to a deformed  Virasoro action
\begin{align*}
L_n^{\d+A}
&=L_n
+\sum_{n'+n''=n} (A_{n''})_{n'}
-\frac{1}{2}\sum_{n'+n''=n}  \langle A_{n'},A_{n''}\rangle
\end{align*}
where $(A_{n''})_{n'}$ denotes the mode operator $n'$ of the element $A_{n''}\in\g$. 
\end{theorem}

\begin{proof}
Now we take the limit $\kappa\to \infty$ and compute the coinvariants with respect to $a_n^{(\kappa)}/\kappa-\lambda_n$ where we express the eigenvalues in terms of the connections by $\lambda_n=\langle A_n,a\rangle$. This expression is by definition equal to the relation 
$$a_n^{(\kappa-1)}/\kappa+a_n^{(1)}/\kappa \sim \langle A_n,a\rangle$$
Hence the coinvariants are given with the following isomorphism of vector spaces
$$\left(\hat{\g}_{\kappa-1} \times \hat{\g}_1\right)/{\sim} \stackrel{\sim}{\longrightarrow} \hat{\g}_1$$
In particular, if $a_n^{(1)}$ acts on modules not scaling with $\kappa$, then taking coinvariants $a_n^{(\kappa)}/\kappa-\langle A_n,a\rangle\kappa$ gives the same as taking coinvariants  $a_n^{(\kappa-1)}/\kappa-\langle A_n,a\rangle$. The Virasoro action of the coset, as defined, is explicitly
\begin{align*}
L_n^{Coset}
&=L_n^{(\kappa-1)}+L_n^{(1)} - L_n^{(\kappa)}\\
&=-\frac{1}{2(\kappa-1+h^\vee)}\sum_{n'+n''=n} \colon v^{(i)(\kappa-1)}_{n'}w^{(i)(\kappa-1)}_{n''}\colon
-\frac{1}{2(1+h^\vee)}\sum_{n'+n''=n} \colon v^{(i)(1)}_{n'}w^{(i)(1)}_{n''}\\
&+\frac{1}{2(\kappa+h^\vee)}\sum_{n'+n''=n}  \colon(v^{(i)(\kappa-1)}_{n'}+v^{(i)(1)}_{n'})(w^{(i)(\kappa-1)}_{n''}+w^{(i)(1)}_{n''})\colon
\end{align*}
for $v^{(i)(\kappa)},w^{(i)(\kappa)}$ a dual basis of $\g$ corresponding to the level $\kappa$. Then the image under the previous isomorphism is
\begin{align*}
L_n^{\d+A(z)}
&=L_n-\frac{1}{2(\kappa-1+h^\vee)}\sum_{n'+n''=n} \colon (-v_{n'}^{(i)(1)}+\langle A_{n'},v^{(i)(1)}\rangle\kappa)(-w_{n''}^{(i)(1)}+\langle A_{n''}, w^{(i)(1)}\rangle\kappa)\colon\\
&+\frac{1}{2(\kappa+h^\vee)}\sum_{n'+n''=n} \colon \langle A_{n'},v^{(i)(1)}\rangle\kappa\langle A_{n''},w^{(i)(1)}\rangle\kappa\colon
\end{align*}
In the limit $\kappa\to \infty$ only terms involving at least one $A_n$ survive. Then taking normally ordered products is not necessary anymore. The quadratic term in $A_n$ survives with a prefactor $$\left(-\frac{1}{2(\kappa-1+h^\vee)}+\frac{1}{2(\kappa+h^\vee)}\right)\kappa^2
%=(-\kappa-p-h^\vee+\kappa+h^\vee)\kappa^2/(2(\kappa+h^\vee)(\kappa+p+h^\vee))
\to -1/2$$ 
The two mixed terms are equal, because the dual basis is symmetric.  %Altogether this gives 
%\begin{align*}
%L_n^{Coset}
%&=L_n^p
%+\sum_{n'+n''=n} 
% v_{n'}^{(i)p}\lambda(w^{(i)p},n'')
%-\frac{p}{2}\sum_{n'+n''=n}  \lambda(v^{(i)p},n')\lambda(w^{(i)p},n'')
%\end{align*}
With the property of the dual basis the 
\begin{align*}
L_n^{\d+A(z)}
&=L_n
+\sum_{n'+n''=n} (A_{n''})_{n'}
-\frac{1}{2}\sum_{n'+n''=n}  \langle A_{n'},A_{n''}\rangle
\end{align*}
\end{proof}

\begin{corollary}
For $a\in \g$ we find 
\begin{align*}
    [L^{\d+A(z)}_n, a_k]
    &=[L_n, a_k] %(-k)\kappa a_{k+n} + 
    +\sum_{n'+n''=n} [ (A_{n''})_{n'},a_k] \\
    &=[L_n, a_k] %(-k)\kappa a_{k+n} + 
    +\sum_{n'+n''=n} [A_{n''},a]_{n'+k}
    + \sum_{n'+n''=n} n'\delta_{n',-k}\langle A_{n''},a\rangle\kappa
\end{align*}
and the quadratic term in $L_n^{\d+A}$ is central. In particular the modified Virasoro action has 
$$L^{\d+A(z)}_{-1}=L_{-1}+A(z)+\langle A(z)',-\rangle $$ 
%with k'=k+l+1 is A_lz^{-1-l}a_kz^{-k-1}=A_la_{-l-1+k'}z^{-k'-1}
%and Res((-1-n'')A_n'',(-k-1)a_k)z^{-2-n''}z^{-1-k}
where $A(z)$ acts as $\g$-valued function with the $\g$-adjoint action on $\g$-valued functions $a(z)$ and $\langle,\rangle$ is evaluated degree-wise. Using an isomorphism like $f$ in Lemma \ref{lm_twistedrepaffine} should remove this last term. 
%e note that for
%$$f:a_k\mapsto a_k-\langle A_{l},a\rangle\delta_{k,l}$$
%we have
%$$f^{-1}L^{\d+A(z)}_{-1}f=$$
\end{corollary}

\begin{example}
We observe that for the connection 
$\d+A=\d+z^{-1}X$ with $X\in\g$ we have
\begin{align*}
L_n^{\d+A(z)}
&=L_n
+X_{n}
-\delta_{n,0}\frac{1}{2}\sum_{k\in\Z}  \langle X,X\rangle
\end{align*}
This coincides with with the formula for the twisted Virasoro algebra in Lemma \ref{lm_twistedrepaffine} with $\kappa=1$.
\CommentsForMe{were there a factor $\frac{1}{1-h^\vee}$?}
\end{example}

\subsection{The quantum Hamiltonian reduction}\label{sec_W}

As suggested in \cite{BBFLT13, BFL16} for $\sl_2$ and conceptually established \cite{ACF22} we now take the coset model and  perform a quantum Hamiltonian reduction $\mathrm{H}^0(-)$ with respect to the diagonally embedded $$\hat{\g}_{\kappa}\hookrightarrow \hat{\g}_{\kappa-p}\times \hat{\g}_{p}$$ 
A general principle  \cite{ACF22} Theorem 1 states that for any vertex algebra $\V$ equipped with a homomorphism $\rV^{k}(\g)\to \V$ we have an isomorphism 
$$\mathrm{H}^\bullet(\V\otimes \bL_p(\g))\cong \mathrm{H}^\bullet(\V)\otimes \bL_p(\g)^{Urod}$$
where the first reduction is with respect to the diagonal embedding and the isomorphism only becomes an isomorphism of conformal vertex algebras if the Virasoro structure on $\bL_p(\g)$ is modified, which we denote $\bL_p(\g)^{Urod}$. This is called a generalization of the previously known Urod algebra is the case $\sl_2,p=1$ \cite{BFL16}.\footnote{They introduce a parameter $\epsilon$ which interpolates between a reduction of the first and the second factor for $\epsilon\to 0,\infty$} A similar statement holds of course for reductions of any module over $\V\otimes \bL_p(\g)$. 

Precomposing this with the diagonal embedding $\rV^{\kappa}(\g)\hookrightarrow \rV^{\kappa-p}(\g)\otimes \bL_p(\g)$.
we obtain a functor called $\rW$-algebra translation, which is proven to be exact
$$\rW^{\kappa}(\g)
\to 
\mathrm{H}^0(\rV^{\kappa-p}(\g)\otimes \bL_p(\g))\cong \rW^{\kappa-p}(\g)\otimes \bL_p(\g)^{Urod}$$

For $p=1$ this is the reduction of the coset construction and gives a construction of the algebra $A(\g,\kappa,p)$ for $p=1$ in Section \ref{sec_couple}
$$\rW^{\kappa}(\g)\otimes \rW^{\kappa^*}(\g)
\hookrightarrow \rW^\kappa(\g)\otimes \bL_p(\g)^{Urod}=A(\g,\kappa,1)$$
$$\frac{1}{\kappa+h^\vee}+\frac{1}{{\kappa^*}+(h^\vee)^\vee}=m,\qquad \kappa+h^\vee\not\in\mathbb{Q}_{\leq 0} $$

The structure of the reduction $A(\g,\kappa,1)$ and also of the reduction of any module over $\rV^\kappa(\g)\otimes \bL_1(\g)$ can be computed if one knows the decomposition of modules if restricted to $\rV^\kappa(\g)\otimes \rW^{{\kappa^*}}(\g)$. This is achieved with the decomposition formulas in the next section.

\subsection{Representations and decomposition formulas}\label{sec_decomposition}

It is important to know the decomposition of a $\rV^\kappa(\g)\otimes \bL_1(\g)^{Urod}$-module over the coset subalgebra $\rV^{\kappa}(\g)\otimes \rW^{\kappa^*}(\g)$. In particular this allows to compute the quantum Hamiltonian reduction with respect to  $\rV^{\kappa}(\g)$ and obtain then similar decomposition of the corresponding $\rW^\kappa(\g)\otimes \bL_1(\g)^{Urod}$-module over the coset subalgebra  $\Vir^{b_1}\otimes \Vir^{b_2}$, which  can be found in  \cite{BFL16} Theorem 2.1 and Theorem 3.3.  and  we will discuss more explicitly in the next section. For general $\g$ and $p=1$ this is in \cite{ACF22} Section 8, which we now quote and we try to keep close to their notation.\\

%b_1\to infty, h_{1,n}\to \infty
\begin{theorem}\label{thm_decomposition}
For generic level $\kappa$ we have the following decompositions of $\rV^\kappa(\g)\otimes \bL_1(\g)^{Urod}$-module over the coset subalgebra $\rW^{\kappa}(\g)\otimes \rW^{\kappa^*}(\g)$
%decompositions of modules over $\rV^\kappa(\sl_2)\otimes \Vir^{b_2}$ resp. its reduction $\Vir^{b_1}\otimes \Vir^{b_2}$:
\begin{enumerate}
\item We have for integral weight $\mu$ \CommentsForMe{$\lambda^*\;?$}
$$\bV^{\kappa-1}_\mu\otimes \bL^{1}_\nu
=\bigoplus_{\substack{ \lambda\in \mu+\nu+ Q \\ \lambda\in P^+ }}  \bV^\kappa_{\lambda}\otimes T^{\kappa^*}_{\lambda,\mu}$$
and reduction resp. reduction after spectral flow $\mu'\neq 0$
$$T^{\kappa-1}_{\mu,\mu'}\otimes \bL^1_\nu
=\bigoplus_{\substack{ \lambda\in \mu+\mu'+\nu+ Q \\ \lambda\in P^+ }}   T^\kappa_{\lambda,\mu'}\otimes T^{\kappa^*}_{\lambda,\mu}$$
%which translates in our notation for sl2 to \marginpar{careful $b_2^{-1}$ etc}
%$$\varphi_{m,n}^{b}\otimes \V_\nu
%=\bigoplus_{l\in \nu+m+n+\mathbb{N}_0}\varphi^{b_1}_{l,n}\otimes %\varphi^{b_2}_{l,m}
%$$
\item We have  for generic weight $\mu$ 
$$\bW^{\kappa-1}_\mu \otimes \bL^1_\nu
=\bigoplus_{\lambda\in \mu+\nu+ Q}  \bW^\kappa_{\lambda}\otimes \pi^{{\kappa^*}+(h^\vee)^\vee}_{\lambda-({\kappa^*}+2)\mu}$$
and reduction 
$$\pi^{\kappa-1+h^\vee}_\mu \otimes \bL^1_\nu
=\bigoplus_{\lambda\in \mu+\nu+ Q}  \pi^{\kappa+h^\vee}_{\lambda}\otimes \pi^{{\kappa^*}+(h^\vee)^\vee}_{\lambda-({\kappa^*}+(h^\vee)^\vee)\mu}$$
\CommentsForMe{TC: In ACF steht $\mu$ $\lambda$ vertauscht. Ist $h^\vee$ in nur einem Index richtig?}
\end{enumerate}
\end{theorem}

%\begin{question}
%The formula is not at all symmetric in the two indices in the right hand %side. So the limit $b\to 0$ would give a very different result.
%\end{question}
\begin{question}
For modules from the Wakimoto module and non-generic weight $\mu$ (for example $\mu=0$) we expect similar fomulae but there can (and will) be nontrivial extensions between the summands. It would be very helpful to have results in this direction. In our context, these extensions are also visible for the twisted modules with nilpotent part, and connections with regular singular connections with nilpotent part.
\end{question}
\begin{question}
There are also decomposition formulae for admissible level. It would be very interesting to repeat out considerations in this case.
\end{question}

\subsection{Bundle associated to the representations \texorpdfstring{$\bV^{\kappa-1}_\mu\otimes \bL^{1}_\nu$}{V}}\label{sec_bun_bV}

The affine Lie algebra $\hat{\g}_\kappa$ has by Section \ref{sec_LimitAff} as limit $\kappa\to \infty$ the ring of functions over the space of $\g$-connections on the punctured disc. Similarly, the affine vertex algebra $\rV^\kappa(\g)$ has as limit the ring of functions over the space of $\g$-connections  and its modules $\bV^{\kappa}_\mu$, with integral $\mu$ not scaled by $b$, have as limit a bundle over the space of regular connections with fibre the irreducible $\g$-module $\LieL_\mu(\g)$. This can be presumably identified with polynomials of degree $\mu$ in the solutions, such that the limit of $\bigoplus_\mu  \bV^{\kappa-1}_\mu$ can be identified with the ring of functions over the space of regular connections together with a solution of the associated differential equation. Applying Theorem \ref{thm_coset} directly yields\\

\begin{corollary}
For integral $\mu$, not scaled with $\kappa$, the limit of the $\rV^{\kappa-1}\otimes \bL^{1}$-module  
$$\bV^{\kappa-1}_\mu\otimes \bL^{1}_\nu$$ 
is a bundle over the space of regular connections with zero-fibre $\LieL_\mu(\g)\otimes \bL^1_\nu$ and modified Virasoro structure.
\end{corollary}
On the other hand, naively applying the limit to the decomposition formula in Theorem \ref{thm_decomposition} yields the following matching result. {\bf A warning} should be placed however, that the decomposition is not necessarily preserving the chosen integral form, as we will see. The direct expression for $p=1$ and the decomposition formula give us two in general different integral forms to consider. 
\begin{corollary}
The limit of the decomposition of $\bV^{\kappa-1}_\mu\otimes \bL^{1}_\nu$ as $\rV^{\kappa}\otimes \rW^{\kappa^*}$-module
$$\bigoplus_{\substack{ \lambda\in \mu+\nu+ Q \\ \lambda\in P^+ }}  \bV^\kappa_{\lambda}\otimes T^{\kappa^*}_{\lambda,\mu}$$ 
is a bundle over the space of regular connections with zero-fibre $\bigoplus_{\substack{ \lambda\in \mu+\nu+ Q \\ \lambda\in P^+ }}  \LieL_\lambda \otimes T^{\kappa^*}_{\lambda,\mu}$
\end{corollary}

For $\mu=\nu=0$ this is the vertex algebra 
$\rV^{\kappa-1}(\g)\otimes \bL^1(\g)$, 
which acquires as big center the ring of functions over regular connections. The vertex algebra with big center becomes a bundle over the space of regular connections with zero-fibre the vertex algebra 
$$\bL^1(\g)= \bigoplus_{\substack{ \lambda\in Q \\ \lambda\in P^+ }}  \LieL_\lambda \otimes T^{\kappa^*}_{\lambda,0}$$
which it the Feigin-Tipunin algebra for simply-laced $\g$ in the limiting case $p=1$.\\

We now discuss the case $\mu\neq 0$. On the Grothendieck ring resp. on the level of characters, we have indeed an equality
\begin{align*}
\LieL_\mu(\g)\otimes \bL^1_\nu(\g)
&= \bigoplus_{\substack{ \lambda\in \mu+\nu+ Q \\ \lambda\in P^+ }} \LieL_\lambda \otimes T^{\kappa^*}_{\lambda,\mu}\\
\intertext{as we can see from computing with $\g$ fusion rules}
    \LieL_\mu(\g)\otimes \bL^1_\nu(\g)
    &=
    \LieL_\mu \otimes \left( \bigoplus_{\substack{ \eta\in \nu+ Q \\ \eta\in P^+ }} \LieL_\eta \otimes T^{\kappa^*}_{\eta,0}\right)\\
    &=\bigoplus_{\substack{ \lambda\in \mu+\nu+ Q \\ \lambda\in P^+ }} \LieL_\lambda \otimes
    \left(
    \bigoplus_{\eta\in P^+ } \binom{\mu\; \eta}{\lambda} T^{\kappa^*}_{\eta,0} \right)\\
\intertext{where we use the fusion rules coefficients $\binom{\mu\; \eta}{\lambda}$, which are nonzero only for $\mu+\eta\in \lambda+Q$, and comparing this to the formula in the limit $\kappa^*+(h^\vee)^\vee \to 1$}
T^{\kappa^*}_{\lambda,\mu}
&=\bigoplus_{\eta\in P^+} \binom{\mu\; \eta}{\lambda} T^{\kappa^*}_{\eta,0}
\end{align*}
Note however that the limit of $T^{\kappa^*}_{\lambda,\mu}$ depends on the choice of an integral form, that has to be compatible with the decomposition formula. We can read off from the previous calculation, that apparently this limit sends $T^{\kappa^*}_{\lambda,\mu}$ to a direct sum of irreducible  $T^{\kappa^*}_{\eta,0}$, and not the  respective indecomposable. 

\begin{example}
For $\sl_2$ we have $\rW^{\kappa^*}$ the Virasoro algebra at central charge $1$ and $T_{\lambda,\mu}=\varphi_{2\lambda+1,2\mu+1}$. On the other the Clebsch-Gordan coefficients are $\binom{\mu\; \eta}{\lambda}=1$ for $|\mu-\eta|\leq \lambda\leq \mu+\eta$, which is equivalent to $|\mu-\lambda|\leq \eta\leq \mu+\lambda$, and $\mu+\eta\in \lambda+2\Z$. Altogether the previous formula reads
\begin{align*}
\varphi_{2\lambda+1,2\mu+1}
&=\bigoplus_{\substack{|\mu-\lambda|\leq \eta\leq \mu+\lambda \\ \mu+\eta\in \lambda+2\Z}}  \varphi_{2\eta+1,1}
\end{align*}
Here, it is clear that there are no non-split extensions on the right-hand side.
\end{example}

\begin{question}
Which integral forms are compatible with the decomposition formula? Can we obtain similar results as the left-hand side for $p>1$? This should be attempted by carefully following the ideas in \cite{CN22}.
\end{question}
\begin{question}
We would like to study the case where $\mu/b^2$ is fixed in the limit. These correspond to regular singular connections.
\end{question}

\subsection{Bundle associated to the spectral flow representations \texorpdfstring{$\sigma^\ell(\bV^{\kappa-1}_\mu)\otimes \bL^{1}_\nu$}{}}\label{sec_bun_spectralflow}

We also consider spectral flow $\sigma^\ell$, which is an automorphism of the affine Lie algebra associated to an element $\ell\in P$, viewed as subgroup of the affine Weyl group. As such it does not preserve the choice of Borel algebra used to define highest-weight modules.

Consider the previous decomposition
%$$\rV^\kappa(\sl_2)\otimes \Vir^{\kappa^*}\hookrightarrow \rV^{\kappa-1}(\sl_2)\otimes \bL^1(\sl_2)$$
$$\bV^{\kappa-1}_\mu(\sl_2)\otimes \bL^1_\nu(\sl_2)=\bigoplus_{\substack{ \lambda\in \mu+\nu+ Q \\ \lambda\in P^+ }}  \bV^\kappa_{\lambda}\otimes T^{\kappa^*}_{\lambda,\mu}$$ 
We apply $\sigma^\ell$ and use that coset is preserved \CommentsForMe{TC?} and that $\sigma^\ell(\bL^1_\nu)=\bL^1_{\nu-\ell}$ and substitute $\nu$
$$\sigma^\ell(\bV^{\kappa-1}_\mu)\otimes \bL^1_{\nu}=\bigoplus_{\substack{ \lambda\in \mu+\nu+\ell+ Q \\ \lambda\in P^+ }}  \sigma^\ell(\bV^\kappa_{\lambda})\otimes T^{\kappa^*}_{\lambda,\mu}$$ 

 For simplicity we now restrict ourselves to the case $\sl_2$: The spectral flow automorphism is defined as follows, which we quote from \cite{CG17} Section 9.1.1\CommentsForMe{Relates to extra-Term in g-twisted module?}

$$\sigma^\ell(e_n)=e_{n-\ell},\qquad
\sigma^\ell(f_n)=f_{n+\ell},\qquad
\sigma^\ell(h_n)=h_n-\delta_{n,0}\ell K
$$
$$\sigma^\ell(L_0)=L_0-\frac{1}{2}\ell h_0+\frac{1}{4}\ell^2 k$$

\begin{example}
The first spectral flow of the vacuum module $\sigma^1(\bV^{\kappa-1}_0)$ is generated by a vector $v$ with 
$$e_nv=0,\;n\geq 1,\qquad 
f_nv=0,\;n\geq -1,\qquad 
h_nv=-\delta_{n,0}\kappa,\;n\geq 0$$
For the obvious integral form, the limit $\kappa\to\infty$ is hence the ring of functions on the space of regular singular $\sl_2$-connections 
$$\d+A_0z^{-1}+A_{-1}z^{0}+\sum_{n\leq -1} A_{n} z^{-1-n},\qquad A_0\in \sl_2^{\geq 0},\qquad A_{-1}\in \sl_2^{\geq 0} $$
The first spectral flow of other modules with finite-dimensional groundstates have $e_1$ nonzero, but nilpotent, so this again corresponds to regular singular connections. 
\end{example}

\begin{example}\label{ex_Affirregular}
The first spectral flow of a Verma module $M_\mu^{\kappa}$ corresponds to an irregular singular connection 
$$\d+A_1z^{-2}+\cdots ,\qquad A_1\in \sl_2^{>0}$$
Note that for non-generic weights the direct decomposition formula does presumably not hold. A second spectral flow on the vacuum module corresponds to an irregular singular connection 
$$\d+A_1z^{-2}+\cdots ,\qquad A_1\in \sl_2^{>0}$$
and the irreducible quotient corresponds to certain equations to be fulfilled by the $A_n$. It is to be expected that irreducible representations lead to connections with a particular easy behavior in contrast to the generic case, somewhat similar to Section \ref{sec_LimitVir}. For example, $A_1\in\sl_2^{\geq 0}$ nilpotent leads to a collapse of the essential singularity $\exp(-A_1z^{-1})$.
\end{example}

\subsection{Bundle associated to the representations \texorpdfstring{$T^{\kappa-1}_{\mu,\mu'}\otimes \bL^1_\nu$}{T}}\label{sec_bun_varphi}

The quantum Hamiltonian reduction resp. reduction after spectral flow $\mu'\neq 0$ of the decomposition in the previous section is by Theorem \ref{thm_decomposition} 
$$T^{\kappa-1}_{\mu,\mu'}\otimes \bL^1_\nu
=\bigoplus_{\substack{ \lambda\in \mu+\mu'+\nu+ Q \\ \lambda\in P^+ }}   T^\kappa_{\lambda,\mu'}\otimes T^{\kappa^*}_{\lambda,\mu}$$

We have no analog of Theorem \ref{thm_coset} for the reduction, but we can naively read off the limit of the right-hand side: It is a bundle over the $\g$-opers with diagonalizable monodromy $\mu'$ and with zero-fibre $$\bigoplus_{\substack{ \lambda\in \mu+\mu'+\nu+ Q \\ \lambda\in P^+ }}  \LieL_\lambda\otimes T^{\kappa^*}_{\lambda,\mu}$$
We observe that this is the very same decomposition as for $\mu'=0$ with $\nu+\mu'$ in order to account for the shift in the summation condition. The $\mu\neq 0$ has a similar effect as in the last section, so we can concentrate on the case $\mu=0$.

\begin{example}
For $\sl_2$ the limit of $\varphi_{2\lambda+1,2\mu'+1}$ is for $\lambda=0$ the ring of functions on the space of Sturm-Liouville operators $\frac{\d^2}{\d z^2}+q(z)$ with $q_0=-\frac{n^2-1}{4}$ for $n=2\mu'+1$ and diagonal monodromy, which is then the monodromy of $z^s$ for $s=\frac{1\pm n}{2}$. More generally for $m=2\lambda+1$ it is a bundle over this set with zero-fibre $\C^{m}=\LieL_\lambda$ as $\sl_2$-module. So for $\mu=0$ the limit of the decomposition is, setting also $l=2\nu+1$,
$$\bigoplus_{\substack{ \lambda\in m+n+l\in 2\Z \\ m\geq 0 }}  \C^m\otimes \varphi_{m,1}$$
\end{example}

Again, as a warning, it is not clear whether the limit of the decomposition agrees with the limit of the actual reduction (which we cannot compute) as modules or just on the level of Grothendieck rings.

\subsection{Bundle associated to the representations \texorpdfstring{$\bW^{\kappa-1}_\mu \otimes \bL^1_\nu$}{W}, generic \texorpdfstring{$\mu$}{mu}}\label{sec_bun_bW}

For generic $\mu$ the Wakimoto module $\bW^{\kappa-1}_\mu$ coincides with the Verma module 
$$\bM^{\kappa-1}_\mu = \hat{\sl}_2 \otimes_{\langle h_0,e_0,\;f_{1},h_{1},e_1,\ldots \rangle} v\C$$ 
with action $h_0v=\mu v$ and all others zero. Assume as first case that $\mu$ does not scale with $b^2$, and as second case that $\mu/b^2$ is kept constant. Then by Theorem \ref{thm_decomposition} we have: The limit of $\bW^{\kappa-1}_\mu \otimes \bL^1_\nu$ is a bundle over the space of connection 
$$\d+A_{0}z^{-1}+\sum_{k\geq 0} A_{-k-1}z^k,\qquad A_0=t\cdot e_0
\hphantom{+\mu/b^2\cdot h_0}$$
$$\d+A_{0}z^{-1}+\sum_{k\geq 0} A_{-k-1}z^k,\qquad A_0= t\cdot e_0+\mu/b^2\cdot h_0$$
for any $t\in\C$ and with fiber $\bL^1_\nu=\V_{\sqrt{2}\Z}$ with modified conformal structure 
$$\tilde{L}_n^{Coset}=L^{\sqrt{2}\Z}_n+\left(t\cdot e_n+\mu/b^2\cdot h_n\right)-\delta_{n,0}\frac{1}{2}\frac{(\mu/b^2)^2}{8}+\sum_{k\geq 0} (A_{-k-1})_{n+k+1}$$
On the other hand we may employ the decomposition formula in Theorem \ref{thm_decomposition}
$$\bW^{\kappa-1}_\mu \otimes \bL^1_\nu=\bigoplus_{\lambda\in \mu+\nu+ Q}  \bW^\kappa_{\lambda}\otimes \pi^{\kappa^*}_{\lambda-({\kappa^*}+2)\mu}$$
In the limit the right-hand side is a bundle over the same space of connections and with fiber
$$\bigoplus_{\lambda\in \mu+\nu+ Q} \pi^{{\kappa^*}+(h^\vee)^\vee}_{\lambda-({\kappa^*}+2)\mu}$$
\CommentsForMe{TC: should be $\lambda-\mu$ for $\sl_2$, then the limit is fine. TC?}

\begin{remark}
We could also choose a different integral form generated by $e_0$ and $e_n/\kappa,f_n/\kappa$ as before. Then instead of a regular singular connection we again get a regular connection, since $e_0/\kappa,f_0/\kappa\to 0$, but the bundle associated to $\mathcal{M}_\lambda^\kappa$ is suddenly the Verma module $\mathcal{M}_\lambda$ over $\g$ and we get as limit
$$\bigoplus_{\lambda\in \mu+\nu+ Q}  \mathcal{M}_{\lambda}\otimes \pi^{{\kappa^*}+(h^\vee)^\vee}_{\lambda-({\kappa^*}+2)\mu}$$
\end{remark}

\subsection{A second limit and the affine variant of the Feigin-Tipunin algebra}

In Question \ref{quest_AffineTriplet} we asked for the affine variant of the Feigin-Tipunin algebra, which should extend $\rV^p(\g)$, should have a decomposition
$$\bigoplus_{\lambda\in Q^+} \bV^{-h^\vee+p}_\lambda\otimes \LieL_{\lambda^*}, $$
have as Hamiltonian reduction the Feigin-Tipunin algebra, and should appear in the present work as the zero-fibre of the respective bundle in the limit of $\HA^{(p)}[\g,\kappa]$.\\

In the case $p=1$ we can calculate this in our explicit coset model as the second possible limit of 
$$\rV^{\kappa-1}(\g)\otimes \bL^1(\g)=\bigoplus_{\lambda\in Q^+} \bV^k_\lambda\otimes T_{\lambda^*,0}^{\kappa^*}$$
$$\kappa\to -h^\vee+1,\quad {\kappa^*}\to \infty$$
The result is clearly $\rV^{crit}(\g)\otimes \bL^1(\g)$ being a bundle over the regular $\g$-opers via the coset embedding of $\rW^{\kappa^*}$ and extending $\rV^{crit}(\g)$ via the diagonal embedding. Note that as discussed in the introduction, the center of $\rV^{crit}(\g)$ is also the ring of functions on regular $\g$-opers. \\

From the explicit formulas in the previous subsection we see that $L_n^{Coset}/\ell=\frac{1}{2}L_n^{crit, rescaled}$ plus terms of order $\ell^{-1}$ if all weights scale with order $1$.

\begin{corollary}
The fibre of $\rV^{crit}(\g)\otimes \bL^1(\g)$ over a regular $\g$-oper is, as a vertex algebra or twisted module, identified with the fibre of  $\rV^{crit}$ tensored with $\bL^1(\g)$. 
\end{corollary}
As already stated, these fibres decompose over $\rV^{-h^\vee+1}$ as  
$$\bigoplus_{\lambda\in Q^+} \bV^{-h^\vee+1}_\lambda\otimes \LieL_{\lambda^*} $$

\begin{question}
Compute the (twisted) representations of this algebra and compare it to the respective limits $\kappa^*\to \infty$ in the general setting of the previous section, and in particular for $p=1$.
\end{question}

In the further case $(\sl_2,p=2)$, the result should be the small N=4 superconformal algebra as discussed in the end of section \ref{sec_explicit_p2_N4}.

\section{Explicit computations for \texorpdfstring{$\sl_2,p=1$}{sl2,p=1}}\label{sec_explicit_p1}

In this section we follow closely the explicit computations in \cite{BFL16}, but to be consistent with the general formulae in Section \ref{sec_couple} we let $\kappa$ parametrizes $b_2$ instead of $b$, so our $\kappa$ corresponds to $\kappa+1$ in \cite{BFL16}. Note also that the sides and the roles of $b_1,b_2$ and are reversed and the factor with $b_1$ is dualized in comparison to the general formulae in Section \ref{sec_couple}. \\

The quantum Hamiltonian reduction of $\rV^\kappa(\sl_2)$ is the Virasoro algebra with parameter (see Section \ref{sec_Vir})
$$b_2=\sqrt{-\kappa-2}$$
Accordingly, the coset and its reduction are
$$\rV^{\kappa^*}(\sl_2)\otimes \Vir^{\kappa}(\sl_2)\hookrightarrow \rV^{\kappa}(\sl_2)\otimes \bL_1(\sl_2)$$
$$\Vir^{b_1^{-1}}\otimes \Vir^{b_2}\hookrightarrow \Vir^{b}\otimes \bL_1(\sl_2)$$
$$ b_1^{-1}=\sqrt{-{\kappa^*}-2}=\sqrt{-\frac{\kappa+2}{\kappa+1}},\qquad {\kappa^*}=-\frac{\kappa}{\kappa+1},\qquad 
b=\sqrt{-\kappa-1},\qquad$$

%OLD CONVENTIONS
%$$\rV^{\kappa+1}(\sl_2)\otimes \Vir^{\kappa^*}(\sl_2)\hookrightarrow \rV^{\kappa}\otimes \bL_1(\sl_2)$$
%$$\Vir^{b_1}\otimes \Vir^{b_2}\hookrightarrow \Vir^{b}\otimes \bL_1(\sl_2)$$
%$$b_1=\sqrt{-\kappa-3},\qquad  b_2=\sqrt{-{\kappa^*}-2}=\sqrt{-\frac{\kappa+3}{\kappa+2}},\qquad {\kappa^*}=-\frac{\kappa+1}{\kappa+2}$$
We check the following relations hold: 
\begin{align*}
    1&=\frac{1}{\kappa+2}+\frac{1}{\kappa^*+2} \\
    1&=b_1^{2}+b_2^{-2}\\
    b_1&=b/\sqrt{1-b^2}\\
    b_2&=\sqrt{b^2-1}
\end{align*}
The Urod algebra is the vertex algebra $\bL_1(\sl_2)$ with generators $f_{-1},e_{-1},h_{-1}$, or equivalently the vertex algebra $\V_{\sqrt{2}\Z}$ with generators $e^{-\sqrt{2}},e^{\sqrt{2}},\partial\varphi_{\sqrt{2}}$, with modified conformal element
$$T^{Urod}=\frac{1}{4}(\partial\varphi\sqrt{2})^2+\frac{1}{2}\partial^2\varphi\sqrt{2}+\epsilon \partial^2e^{\sqrt{2}\varphi}$$
\cite{ACF22} state in Lemma 6.1 and Example 6.2 similar formulae for an Urod algebra $\bL_1(\g)$ with modified Virasoro action, corresponding to the case $\g,p=1$.

\begin{remark}
It might be interesting to try to understand the Urod algebra as affine Lie algebra at level $1$ again deformed by a connection. 
\end{remark}

 In \cite{BFL16}  Theorem 3.3 a) and  Theorem 2.1 b) explicit instances of the defining decompositions in Section \ref{sec_couple} are proven for this model
 $$\Vir^b \otimes \V_{\sqrt{2}\Z}
    =\bigoplus_{n\geq 0} \varphi_{n,1}^{b_1}\otimes \varphi_{1,n}^{b_2} 
    $$
 In \cite{BFL16} Remark 3.1 explicit conformal elements realizing the embedding of $\Vir^{b_1},\;\Vir^{b_1}$ in $M^b\otimes \V_{\sqrt{2}\Z}$ are computed 
 \begin{align*}
    T^{(b_1)}
    &=\frac{b+b^{-1}}{2(b-b^{-1})\epsilon}e^{-\sqrt{2}\varphi}
    +\frac{b}{2(b-b^{-1})}(\partial\varphi)^2
    -\frac{b^{-1}}{\sqrt{2}(b-b^{-1})} \partial^2\varphi
    -\frac{(1+2b^{-2})\epsilon}{b^2-b^{-2}}(\partial\varphi)^2e^{\sqrt{2}\varphi}\\
    &-\frac{\sqrt{2}b^{-1}\epsilon}{b-b^{-1}} \partial^2\varphi e^{\sqrt{2}\varphi}
    -\frac{2\epsilon^2}{b^2-b^{-2}}e^{2\sqrt{2}\varphi}
    -\frac{b^{-1}}{b-b^{-1}}T^b
    -\frac{2\epsilon}{b^2-b^{-2}}T^b e^{\sqrt{2}\varphi}\\
    T^{(b_2)}
    &=-\frac{b+b^{-1}}{2(b-b^{-1})\epsilon}e^{-\sqrt{2}\varphi}
    -\frac{b^{-1}}{2(b-b^{-1})}(\partial\varphi)^2
    +\frac{b}{\sqrt{2}(b-b^{-1})} \partial^2\varphi
    +\frac{(1+2b^{2})\epsilon}{b^2-b^{-2}}(\partial\varphi)^2e^{\sqrt{2}\varphi}\\
    &+\frac{\sqrt{2}b\epsilon}{b-b^{-1}} \partial^2\varphi e^{\sqrt{2}\varphi}
    +\frac{2\epsilon^2}{b^2-b^{-2}}e^{2\sqrt{2}\varphi}
    +\frac{b}{b-b^{-1}}T^b
    +\frac{2\epsilon}{b^2-b^{-2}}T^b e^{\sqrt{2}\varphi}
\end{align*}

The limit $b,b_2\to \infty$ acquires a big central subalgebra generated by $L_n^{b_2}/b_2^2$. In the formula for $T^{b_2}/b_2^2$ all term except $T/b^2$ vanish in the limit (assumed that we do not rescale the module with $b$), and we hence recover a similar result as in the proof of Theorem \ref{thm_coset}
$$L_n^{(b_2)}/b_2^2=L_n^{(b)}/b^2$$

\begin{corollary}\label{cor_expl_p1}
The limit of $\V^\kappa$  fibres over the space of regular Sturm-Liouville operators $\frac{\d^2}{\d z^2}+q(z)$ with $q(z)=\sum_{n\in\Z} \ell_nz^{-2-n}$. the fibre $\V^\infty|_{q(z)}$ is nonzero for ${\ell}_n=0,n\geq -1$  and equal to the lattice vertex algebra $\V_{\sqrt{2}\Z}$ with  deformed Virasoro action 
\begin{align*}
L_n^{(b_1)}
&=\frac{1}{2\epsilon}f_{n+1} %z^{-1-(n+1)}
+L_n^{\V_{\sqrt{2}\Z}} %\frac{1}{2}(\partial\varphi)^2_n %{z^{-2-n}}
-{\ell}_n
-2\epsilon\sum_{k\in \Z}{\ell}_k e_{n-k-1} \\ %z^{-2-k-1-(n-k-1)}
\end{align*}
 Note that for regular connections all contributions are degree-lowering and hence not essentially change the structure as a Virasoro module. \\
 
In particular the zero-fibre $\V^\infty|_0$ is isomorphic to $\V_{\sqrt{2}\Z}$ deformed by a connection $\d+\frac{1}{2\epsilon}fz^0$, which is what (in hindsight) we expect from the classical Hamiltonian reduction.\\

We also have an $\sl_2$-action by $L_n^{(b_2)},n=-1,0,1$, which defined only for  ${\ell}_{-1},{\ell}_0,{\ell}_{1}=0$

\begin{align*}
L_n^{(b_2)}
&=-\frac{1}{2\epsilon}f_{n+1}
+\frac{1}{\sqrt{2}}\partial^2\varphi_n %{z^{-2-n}}
%+\epsilon Y(\partial^2e^{+\sqrt{2}\varphi})_{-2-n}
%$(-k-1)(-k-2)e_k z^{-k-3}=Y(\partial^2 e^\sqrt{2},z), n=k+1
+(-n)(-n-1)e_{n+1}
+L_n^{(b)}
+2\epsilon\sum_{k\in \Z}{\ell}_k e_{n-k-1}\\
\end{align*}
\end{corollary}

%%%%%%
%We find an eigenvector for $(\frac{1}{4},\pm\frac{1}{2})$
%$$\left(\frac{1}{2\epsilon}f_{1}
%+\frac{1}{2}(\partial\varphi)^2-\frac{1}{4}
%-\lambda_n
%-2\epsilon\sum_{k=-2}^\infty\lambda_k e_{-k-1}\right)v=\frac{1}{4}v$$
%$$\left(-\frac{1}{2\epsilon}f_{1}
%+\frac{1}{\sqrt{2}}\partial^2\varphi
%+Y(\partial^2e^{+\sqrt{2}\varphi})_{-2}
%+L_0
%+2\epsilon\sum_{k=-2}^\infty\lambda_k e_{-k-1}\right)v=\pm\frac{1}{2}v
%$$
%All terms $e_1,f_1$ are zero so the eigenvectors are simply as follows.
\begin{example}
In particular the conformal weights resp. $L_0^{(b_1)}$-eigenvalues coincide with the standard choice of $\V_{\sqrt{2}\Z}$, and the $\sl_2$-weights of $L_0^{(b_2)}$ are according to the $\sqrt{2}\Z$-grading. $0,\;\frac{1}{4},\frac{1}{4},\;1,1,1,\ldots$ simultaneous with $0,\;-\frac{1}{2},+\frac{1}{2},\;-1,0,+1,\ldots$. 
The first eigenvectors are 
\begin{align*}
v_{(1,1)\times(1,1)}
&:=1\\
v_{(2,1)\times(1,2)}
&:=e^{-\frac{1}{\sqrt{2}}\varphi}\\
L_1^{(b_2)}	v_{(2,1)\times(1,2)}
%&:=(-1)(-2)\epsilon e^{+\frac{1}{\sqrt{2}}\varphi} 
%as twice differetiated $\epsilon Y(e^{\sqrt{2}\varphi})_{-1}e^{-\frac{1}{\sqrt{2}}\varphi}$
&:=2\epsilon \;e^{+\frac{1}{\sqrt{2}}\varphi}
\end{align*}
and here $L^{(b_1)}_{-1}=-\frac{1}{2\epsilon}f_0$.
\end{example}

%The situation changes if we consider modules with ${\ell}_{-1}\neq 0$ ($L_{-1}$ acts not nilpotently), and/or ${\ell}_0\neq 0$ (conformal weight scaling with $b^2$). For a generic Virasoro module $M^b_h$ where $h$ does not scale with $b$ we get a bundle over the space of Sturm-Liouville operators with ${\ell}_0=0,{\ell}_{-1}\neq 0$ and conformal structure by the formula above, which now involves an additional term $-2\epsilon \cdot {\ell}_{-1}e_n$, which  precisely matches the formula for $e_0$-twisted modules of the triplet algebra, up to nilpotent terms.\marginpar{BORIS?}. In particular we find $L_0$ Jordan blocks.\\
%Again, for a generic Virasoro module $M^b_h$ with $h/b^2={\ell}_0$ finite we get a bundle over the space of Sturm-Liouville operators with fixed ${\ell}_0$.

\section{Explicit computations for \texorpdfstring{$\sl_2,p=2$}{sl2,p=2}} \label{sec_explicit_p2}

\subsection{The limit of the N=4 superconformal algebra}\label{sec_explicit_p2_N4}

The large N=4 superconformal algebra $\sVir_{N=4}^{k,a}$ as discussed in \cite{CGL20} Section 2 is a quantum Hamiltonian reduction of $\hat{D}(2,1;a)_k$ and is generated by two commuting sets of $\hat{\sl}_2$ with generators $e,f,h$ and $e',f',h'$ at level

$$-\frac{a+1}{a}k-1,\qquad -(a+1)k-1$$

and fermionic generators $G^{\pm\pm}$. There are two commuting Virasoro actions $L_{\sl_2'}$ and $L_C:=L-L_{\sl_2'}$. All OPEs are explicitly given at the cited chapter.

In Section 2.4 the explicit case $k=1/2$ is discussed
\begin{example}[$k=1/2$]
The decomposition as $\sl_2'\times \sl_2$-module

$$\sVir_{N=4}^{k,a}=\bigoplus_{\lambda\in Q^+} \bV_{\lambda}^{-(a+3)/2}\otimes \bV_{\lambda}^{-(a^{-1}+3)/2}$$

and from the proof of Theorem 2.5 following \cite{Ad16} we find that the  highest weight vectors are $X_n=(G^{++})(\partial G^{++})\cdots(\partial^{n-1}G^{++})$.
\end{example}

The limit $a\to \infty$ of $\sVir_{N=4}^{(k,a)}$ as deformable family defined by the generators $e'/\kappa,f'/\kappa,h'/\kappa$ and $L_C,G^{\pm\pm}$ is computed in \cite{CGL20} Section 2.2. It acquires a large commutative subalgebra generated by $e'/\kappa,f'/\kappa,h'/\kappa$ isomorphic to the ring of functions on the space of $\sl_2$-connections. The zero-fibre coincides for $k=1/2$ with the so-called small N=4 superconformal algebra studied by Adamovic \cite{Ad16} using free-field realization.  \\

\subsection{The limit of \texorpdfstring{$\hat{\osp}(2|1)_\kappa$}{osp(1|2)}}

The Hamiltonian reduction of $\sVir_{N=4}^{1/2,a}$ with respect to $\hat{\sl}_2$  is calculated in \cite{CGL20} Section 2.5:

\begin{theorem}
$\HA^{(2)}[\sl_2,\kappa]$ with $\kappa=-(a+3)/2$ is isomorphic to $\rV^{-(a+3)/2}(\osp(1|2))$ with generators
$$e',f',h',x':=\frac{a+1}{\sqrt{2a}}G^{++},y':=-\frac{a+1}{\sqrt{2a}}G^{-+}$$
and the additional cohomology relation $e+1=0$. 
\end{theorem}

We now turn to the limit $\kappa\to \infty$ of 
$$\HA^{(2)}[\sl_2,\kappa]\cong \rV^{\kappa}(\osp(1|2))$$
where $\kappa$ is related to the parameters in the previous section by $\kappa=-(a+3)/2,k=1/2$. It has generators $e'/\kappa,\;f'/\kappa,\;h'/\kappa$ forming $(\hat{\sl}_2)_\kappa$ and fermionic generators $x'/\sqrt{\kappa},y'/\sqrt{\kappa}$. Their OPEs can be read off \cite{CGL20} p. 16 to be
\begin{align*}
    (x'/\sqrt{\kappa})(z)(x'/\sqrt{\kappa})(w)&\sim-(e'/\kappa)(w)(z-w)^{-1}\\
    (y'/\sqrt{\kappa})(z)(y'/\sqrt{\kappa})(w)&\sim(f'/\kappa)(w)(z-w)^{-1}\\
    (x'/\sqrt{\kappa})(z)(y'/\sqrt{\kappa})(w)&\sim
    (z-w)^{-2}
    +\frac{1}{2}(h'/\kappa)(w)(z-w)^{-1}
\end{align*}

An arbitrary fibre is defined by scalars $(e'/\kappa)_n,(f'/\kappa)_n,(h'/\kappa)_n$ and as in section \ref{sec_coset} we define a $\sl_2$-connection $\d+A(z)$ with $A(z)=\sum_{n\in \Z} A_nz^{-n-1}$ using the Killing form $(a/\kappa)_n=\langle A_n,a\rangle$, more explicitly for $\sl_2$
$$A_n= \frac{1}{4}f' (e'/\kappa)_n+\frac{1}{8}h' (h'/\kappa)_n+\frac{1}{4}e' (f'/\kappa)_n$$
We now introduce in the fibre $\d+A$ the variables 
$$\psi^{\d+A}:=x'/\sqrt{\kappa},\qquad \bpsi^{\d+A}:=y'/\sqrt{\kappa}$$
Then we find, by translating the singular part of OPE to (anti-)commutators 
$$[A_{-m-h_A},B_{-n-h_B}]_\pm=\sum_{l\leq 0} \binom{-m-1}{-l-1} (A_{-l-h_A}B)_{-(n+m-l)-h_{A_{-l-h_A}B}}$$

\begin{lemma}\label{lm_deformedModeAlgebra}
The limit of $\HA^{(2)}[\sl_2,\kappa]$ fibres over the space of regular $\sl_2$-connections and the fibre over the point $\d+\sum_{n\in\Z} A_{-n-1}z^n$ is given by the deformed mode algebra 
\begin{align*}
  %[(x'/\sqrt{a})_{-n-1},  (x'/\sqrt{a})_{-m-1} 
  %&=-(e'/a)_{-(n+m+1)-1}\\
  \{\psi_m^{\d+A}, \psi_n^{\d+A}\}
  &=-\langle A_{n+m},f'\rangle\\
  \{\bpsi_m^{\d+A}, \bpsi_n^{\d+A}\}
  &=\hphantom{-}\langle A_{n+m},e'\rangle\\
  %[(x'/\sqrt{a})_{-n-1},(y'/\sqrt{a})_m] 
  %&=-\frac{1}{2}(-n-1)1_{-(n+m+2)-0}+\frac{1}{2} (h'/a)_{-(n+m+1)-1}(w)\\
  \{\psi_m^{\d+A},\bpsi_n^{\d+A}\}&=\frac{1}{2}\langle A_{n+m},h'\rangle+
  m \delta_{m,-n}
\end{align*}
\end{lemma}
We now study the operators  $e'_0,f'_0,h'_0$. They are not part of the integral form, since $e'_0=\kappa(e'_0/\kappa)\to \infty$, but they can be considered if $(e'_0/\kappa),(f'_0/\kappa),(h'_0/\kappa)=0$, or respectively for a subset. From the remaining OPEs we read off easily that they form a well-defined Lie algebra $\sl_2$ and that $x',y'$ transform as a douplet under this symmetry. \\

We now discuss Virasoro actions. The Virasoro action provided by the Sugawara construction of $\osp(2|1)$ is 
$$(3+2\kappa)L^{\osp(1|2)}=
\frac{1}{2}:h'h': + :e'f': + :f'e': - :x'y': + :y'x':
\qquad \kappa=-(a+3)/2$$
By construction we have the embedding
$$(\hat{\sl}_2')_\kappa \times \Vir_c
\quad \hookrightarrow\quad 
\osp(2|1)_\kappa$$
where $\Vir_c$ is the reduction of $(\hat{\sl}_2)_{-\frac{a+1}{a}k-1}$ for $k=1/2$ and in the limit $(\hat{\sl}_2)_{-\frac{3}{2}}$, so the central charge is $c=-2$. Explicitly, this embedding is given by the coset Virasoro action 
$$L_n^C=L^{\osp(1|2)}_n-L^{\sl_2'}_n
=\left(\frac{2(\kappa+2)}{3+2\kappa}-1\right)L^{\sl_2'}_n
+\frac{2}{3+2\kappa}:y'x':_n$$
In the limit $\kappa\to \infty$ the first term disappears which leaves the second term $\sum_{i+j=n}:\bpsi_i\psi_j:$. Hence we precisely recover the Virasoro structure given by the energy-stress tensor of the symplectic fermions, for all fibres $\d+A(z)$. 

\begin{corollary}
The limit of $\HA^{(2)}[\sl_2,\kappa]$ fibres over the space of regular $\sl_2$-connections, and the fibre over each connection $\d+A(z)$ is isomorphic to a twist of symplectic fermions with the twisted mode algebra given above.
\end{corollary}

\begin{lemma}\label{lm_p2_translationconnection}
We have 
\begin{align*}
    \{L_n^{\d+A},\psi_k^{\d+A} \}
    &=\sum_{i+j=n} \bpsi_i^{\d+A}\psi_j^{\d+A}\psi_k^{\d+A}
    -\psi_k^{\d+A}\bpsi_i^{\d+A}\psi_j^{\d+A}\\
    &=\sum_{i+j=n} -\bpsi_i^{\d+A}\langle A_{j+k},f'\rangle
    -\bpsi_i^{\d+A}\psi_k^{\d+A}\psi_j^{\d+A}
    -\psi_k^{\d+A}\bpsi_i^{\d+A}\psi_j^{\d+A}\\
   &=\sum_{i+j=n} -\bpsi_i^{\d+A}\langle A_{j+k},f'\rangle
   -\frac{1}{2}\langle A_{i+k},h'\rangle\psi_j^{\d+A}-k\delta_{i,-k}\psi_j^{\d+A}\\
   &=-k\psi_{k+n}^{\d+A}    -\sum_l -\langle A_{l},f'\rangle\bpsi_{k+n-l}^{\d+A}
   -\frac{1}{2}\langle A_{l},h'\rangle\psi_{k+n-l}^{\d+A}
\end{align*}
and similarly for $\bpsi$, where we can ignore the normally ordered product because the commutators are central. Differently spoken, the new Virasoro action acts on the field $\psi(z)=\sum_k \psi_kz^{-1-k}$ by
$$L_n^{\d+A}=\frac{\partial}{\partial z}z^{n+1}+A(z)z^{n+1}$$
%(n-k)\psi_k z^{n-k-1}+ A_l.\psi_k z^{n-l-1-k}
%k'=k-n, k''=k+l-n
%-k'\psi_{k'+n}z^{-1-k'}+z^{-1-k''} A_l.\psi_{k''+n-l}
where $A(z)$ is an $\sl_2$-valued function acting on the vector-valued field with components $\psi(z),\bpsi(z)$. In particular $L_{-1}$ acts as the connection $\d+A(z)$.
\end{lemma}

\begin{example}
The zero fibre coincides with symplectic fermions. The fibre over $\d+A_0z^{-1}$, assuming by conjugation that $A_0$ is contained in $\g^{\geq0}$, is isomorphic to the twisted sector in section \ref{sec_twistedTriplet}.  
\end{example}
\begin{example}\label{ex_SFirregular}
Let us consider the first type of {\bf irregular} singularity $$A(z)=A_1z^{-2}+A_0z^{-1}+\cdots$$
where we assume $A_1\in\mathfrak{h}$ (called compatible framing in literature) and explicitly $A_1=4\xi h'$. As the easiest case let us assume $A_n=0,n\leq 0$, then the twisted mode algebra is 
\begin{align*}
  \{\psi_m^{\d+A}, \psi_n^{\d+A}\}
  &=0\\
  \{\bpsi_m^{\d+A}, \bpsi_n^{\d+A}\}
  &=0\\
  \{\psi_m^{\d+A},\bpsi_n^{\d+A}\}&= \xi\delta_{m+n,1}+m \delta_{m+n,0}
\end{align*}
In particular, in contrast to the regular singular case  the $2^4$-dimensional  subalgebra generated by $\psi_0,\bpsi_0,\psi_1,\bpsi_1$ is not any more (super-)commutative, rather
$$\{\psi_0,\bpsi_1\}=\{\bpsi_0,\psi_1\}=\xi$$
so that it admits no $1$-dimensional representations, in particular no highest-weight vectors. To determine the category of representations as in Lemma \ref{lm_repTwistedClifford}, we note that $\psi_1,\bpsi_1$ and $\psi_{-1},\bpsi_{-1}$ have unchanged commutator relations. A simple base change gives  modified base pairs 
\begin{align*}
    \psi_2'&:=\psi_2-\xi \psi_{1} \\
    \bpsi_2'&:=\bpsi_2+\xi \bpsi_{1} 
\end{align*}
which together with $\psi_{-2},\bpsi_{-2}$ fulfill again the untwisted relations among themselves and with $\psi_1,\bpsi_1$. Inductively  a simple base transformation give modified base pairs $\psi_n',\bpsi_n'$ and $\psi_{-n}',\bpsi_{-n}'$ with the untwisted relations. Next we introduce modified zero-mode operators
\begin{align*}
    \psi_0'&:=\psi_0+\xi \psi_{-1} \\
    \bpsi_0'&:=\bpsi_0-\xi \bpsi_{-1} 
\end{align*}
that commute with all other modified $\psi_n',\bpsi_n'$. Now again the representations of the twisted mode algebra is given by highest-weight representations and thus by representations of the zero-mode algebra. In the case at hand the relations after base change are also unchanged
\begin{align*}
    \{\psi_0',\psi_0'\}&=0 \\
    \{\bpsi_0',\bpsi_0'\}&=0 \\
    \{\psi_0',\bpsi_0'\}&= 0
\end{align*}
(this is not true anymore if further regular terms are present!) Hence the simple representations of the zero-mode algebra are $1$-dimensional objects $v\C$ with $\psi_0'v=\bpsi_0'v=0$ and the projective objects are the typical diamond freely generated by $v$. Differently spoken, the simple object is generated by a generalized highest-weight vector $v$ with $\psi_nv=\bpsi_nv=0$ for $n>0$ and 
$$\psi_0v=-\xi\psi_{-1}v,\qquad \bpsi_0v=\xi\bpsi_{-1}v$$
Note however that the action of the Virasoro algebra has changed. In particular on the highest weight vector $v$ with $\psi_n'v=0,n>0$ and at the same $\psi_n v=0,n>0$

\begin{align*}
L_0v &= (\bpsi_0\psi_0) v\\
&=\bpsi_0'\psi_0' +\xi(\bpsi_{-1}\psi_0'+\psi_{-1}\bpsi_0')-\xi^2 \bpsi_{-1}\psi_{-1}\\
L_1v &= (\bpsi_0\psi_1  - \psi_0\bpsi_1)v \\
&=0
\end{align*}
while the vectors $\psi_0v,\bpsi_0$ in the top space we find 
\begin{align*}
L_1\,\bpsi_0v &= (\bpsi_0\psi_1  - \psi_0\bpsi_1)\bpsi_0 v \\
&=\xi \bpsi_0 \\
L_1\,\psi_0v &= (\bpsi_0\psi_1  - \psi_0\bpsi_1)\psi_0 v \\
&=-\xi\psi_0
\end{align*}
Hence the highest weight-vector even for the simple modules is not $L_0$-eigenvector, and we find Whittaker vectors related to the irregular singularity. Note however that because there are no regular terms present, there are only additional contributions of higher degree and the  module is still filtered. \\

\begin{comment}
Starting from the assumption of a vector $v$ with $\psi_1v=\bpsi_1v=0$ we find for $\xi\neq 0$ a $4$-dimensional irreducible representation of this algebra (which is hence a full matrix algebra)
\begin{center}
    \begin{tikzcd}
    & v 
    \ar[dl,"\psi_0", shift right=-0.5ex]
    \ar[dr,swap,"\bpsi_0", shift left=-0.5ex] & \\
    \psi_0v 
    \ar[dr, "\bpsi_0", shift right=-0.5ex] 
    \ar[ur, "\bpsi_1", shift right=-0.5ex] && 
    \bpsi_0v 
    \ar[dl, swap, "\bpsi_0", shift left=-0.5ex] 
    \ar[ul, swap, "\psi_1", shift left=-0.5ex] \\
    & \psi_0\bpsi_0v 
    \ar[ul, "\psi_1", shift right=-0.5ex]
    \ar[ur, swap, "\bpsi_1", shift left=-0.5ex]
\end{tikzcd}  
\end{center}
This can be trivially extended to all $\psi_n,\bpsi_n,n\geq 0$, since all remaining generators are central, and induced up to full twisted representation. If we evaluate $L_n,n\geq 0$ on this top space we find 
\begin{align*}
    L_n &=0,\qquad\qquad\qquad  n\geq 2 \\
    L_1 &= \bpsi_0\psi_1 -\psi_0\bpsi_1 \\
    L_0 &= \bpsi_0\psi_0+(\bpsi_{-1}\psi_1-\psi_{-1}\bpsi_1)
\end{align*}
Hence $L_1$ acts on $v,\psi_0\bpsi_0v$ by zero and on $\psi_0v,\bpsi_0v$, 
by eigenvalues $\mp \xi$ (a Whittaker-type behaviour), and $L_0$ acts on $v,\psi_0\bpsi_0v$ as a Jordan block (a logarithmic behaviour) and terms involving $\psi_{-1},\bpsi_{-1}$, so the top space is not preserved, however these contributions are nilpotent, but they add higher correction terms on the true groundstates in a similar way the regular connection adds lower correction terms. \\
\end{comment}

We also briefly describe the perspective of differential equations. Consider
$$\left(\frac{\d}{\d z}
-\xi\begin{pmatrix} 1 & 0 \\ 0 & -1\end{pmatrix} z^{-2}
-\Lambda z^{-1}
- A_{reg}(z)\right)
\begin{pmatrix} x(z) \\ y(z) \end{pmatrix}=0$$
A substitution $\varphi(z) z^\Lambda \exp\left(-\xi z^{-1}\begin{psmallmatrix} 1 & 0 \\ 0 & -1\end{psmallmatrix}\right)$ gives a power series approach for $\varphi(z)$, with $z^\Lambda$ causing monodromy around $z=0$ prescribed by the regular singular term and $\exp(\mp \xi z^{-1})$ causing an essential singularity at $z=0$. In the easiest case $\Lambda=0,A_{reg}(z)=0$, corresponding to the twisted mode algebra above we have explicit solutions  
$$\begin{pmatrix} x(z) \\ y(z) \end{pmatrix}
=A\begin{pmatrix} e^{-\xi/z}  \\ 
0  \end{pmatrix}
+B\begin{pmatrix} 0 \\ 
e^{\xi/z}   \end{pmatrix}$$
And similarly for diagonal $\Lambda$ we get distinguished solutions $e^{\pm\xi/z}z^{\mp\lambda}$. 

%The next interesting and still explicit case is $\Lambda=\begin{psmallmatrix} \lambda & b \\ 0 & -\lambda\end{psmallmatrix}$, in which we can still quite explicitly determine, with substituting $t=z^{-1}$
%$$\begin{pmatrix} x(z) \\ y(z) \end{pmatrix}
%=\begin{pmatrix} \exp(-\xi z^{-1})z^{\lambda}
%&  \exp(-\xi z^{-1})z^{\lambda} (-b)\int \exp(2\xi t)t^{2\lambda-1} \; \mathrm{d} t\\ 
%0 &  \exp(\xi z^{-1})z^{-\lambda} 
%\end{pmatrix}$$
%For $\lambda=1/2$ 
%$$\begin{pmatrix} x(z) \\ y(z) \end{pmatrix}
%=\begin{pmatrix} \exp(-\xi z^{-1})z^{1/2}
%&  \exp(\xi z^{-1})z^{1/2} (-\frac{b}{2\xi})\\ 
%0 &  \exp(\xi z^{-1})z^{-1/2} 
%\end{pmatrix}$$

On the other hand, if we add a regular term, e.g. constant and nilpotent $A_{reg}=\begin{pmatrix} 0 & e \\ 0 & 0\end{pmatrix}$:

$$\begin{pmatrix} x(z) \\ y(z) \end{pmatrix}
=A\begin{pmatrix} e^{-\xi/z} \\ 0 \end{pmatrix}
+B\begin{pmatrix} (-ez)\mathrm{E_2}(-2\xi/z) e^{-\xi/z}\\ e^{\xi/z} \end{pmatrix} $$
by first solving the second equation and plugging into the first, where we encounter an exponential integral 
$$(-e)\int_{0}^z \cdot e^{-2\xi / t} \mathrm{d}t=-(2\xi e) \int_{\infty}^{-2\xi/z}  \frac{e^{-x}}{x^2} \mathrm{d}x=(-ez)\mathrm{E_2}(-2\xi/z)$$
It has monodromy $(2\xi e)$ around $z=0$, so the monodromy matrix in this basis of solutions is a Jordan block.\\

With these examples in mind, we now briefly review the Stokes phenomenon and the irregular Riemann Hilbert correspondence \cite{Boa14} in our case $\sl_2$ and an irregular singularity of order $z^{-2}$: All connections with fixed $(\xi,\Lambda)$ (called formal type) are equivalent under formal coordinate transformations. However for $\xi\neq 0$ they are not equivalent under holomorphic coordinate transformations, i.e. different choices of $A_{reg}$ can have essentially different solutions. The holomorphic equivalence classes are parametrized by Stokes matrices, see \cite{Boa02} Section 2: 
The anti-Stokes directions in $\C^\times$ separate the sectors which the differential equation has a defined asymptotic $z\to 0$, in our case the positive and negative real axis where the asymptotics of the essential singularity chances. The associated subgroups of Stokes multipliers are the unipotent subgroups $U^\pm$ and describe the transformation behaviour of the solution between adjacent sectors. \\

In the example $A_{reg}=0$ above the Stokes matrices are trivial, and accordingly there is a distinguished base of solutions. Adding regular terms produces connections which involve the exponential integral. In the example above with $A_{reg}$ constant and nilpotent, the Stokes phenomenon is well studied and directly visible from the asymptotic expansions overlapping for $\frac12 \pi < \arg < \frac32$
\begin{align*}
%\mathrm{Ei_1}(x) &= \frac{e^{-x}}{x}\sum_{s=0}^\infty \frac{(-1)^s s!}{x^s} 
\mathrm{E_2}(x) &= \frac{e^{-x}}{x}\sum_{s=0}^\infty \frac{(-1)^s (s+1)!}{x^s} 
\qquad & -\frac32 \pi < \arg < \frac32 \pi\\
%\mathrm{Ei_1}(x) &= \frac{e^{-x}}{x}\sum_{s=0}^\infty \frac{(-1)^s s!}{x^s} +(-2\pi\i)
\mathrm{E_2}(x) &=  \frac{e^{-x}}{x}\sum_{s=0}^\infty \frac{(-1)^s (s+1)!}{x^s}
+2\pi\i x
\qquad & \frac12 \pi < \arg < \frac72 \pi\\
\end{align*}
In particular there is no preferred pair of bases defined in terms of asymptotics around $z=0$, and there transformation, the Stokes matrices, correspond to the monodromy. The example relates by a substitution $t=z^{-1}$ to the \emph{dynamical KZ-equation} \cite{FMTV00}, for which there seems so be a deep connection to the quantum group with big center \cite{TLX23}.\\ 

We would hope that our approach makes the Stokes matrices visible in the representation theory of the twisted mode algebra (which depends on the regular terms) and more systematically explains the observed relation of the quantum Weyl group to the quantum group with a big center and to the negative definite case in Question \ref{quest_Liouville}. 
\end{example}

\subsection{The limit of the N=1 superconformal algebra and fermion}

The further Hamiltonian reduction of $\sVir_{N=4}^{1/2,a}$ with respect to $\hat{\sl}_2$ and then $\hat{\sl}_2'$ is calculated in \cite{CGL20} Section 2.5:

\begin{theorem}
$\HHA^{(2)}[\sl_2,\kappa]$ with respect to the BRST differential $\d'=b'(z)e'(z)+b'(z)$  is isomorphic to the N=1 superconformal algebra $\sVir_{N=1}$ with central charge $c=3/2+3(a+2+a^{-1})^2$ and generators the cohomology classes of 
\begin{align*}
L'&=L^{\osp(2|1)}+\frac{1}{2}\partial h'-:b'\partial c':-\frac{1}{2}:(\partial x')x':\\
G_{-3/2}&=\frac{\sqrt{-1}}{\sqrt{3+2\kappa}}\left( :h'x':+2:e'y':-(1+2\kappa)e'\partial x'\right)
\end{align*}
times a free fermion with OPE
$$x'(z)x'(w)\sim -e'(z-w)^{-1}=1(z-w)^{-1}$$
and the additional relation $e'+1=0$ up to coboundary. 
\end{theorem}

\begin{comment}%More calculations for this limit
We can compare in the previous integral form the limits 
\begin{align*}
    L^{\sl_2'}/\kappa
    &=\frac{1}{4}:(h'/\kappa)(h'/\kappa):+\frac{1}{2}:(e'/\kappa)(f'/\kappa):+\frac{1}{2}:(f'/\kappa)(e'/\kappa):\\&+\frac{1}{2}\partial(h'/\kappa)-\cancel{\frac{1}{\kappa}:b'\partial c':} \\
    L'/\kappa
    &=\frac{1}{4}:(h'/\kappa)(h'/\kappa):+\frac{1}{2}:(e'/\kappa)(f'/\kappa):+\frac{1}{2}:(f'/\kappa)(e'/\kappa):\\
    &-\cancel{\frac{1}{2\kappa}:(x'/\sqrt\kappa)(x'/\sqrt\kappa):}-\cancel{\frac{1}{2\kappa}:(y'/\sqrt\kappa)(y'/\sqrt\kappa):}\\
    &+\frac{1}{2}\partial(h'/\kappa)-\cancel{\frac{1}{\kappa}:b'\partial c':}
    -\frac{1}{2}:(\partial x'/\sqrt\kappa)(x'/\sqrt\kappa):
\end{align*}
We can solve to obtain a new integral form (up to scale)
$$y'/\sqrt\kappa=G_{-3/2}-\kappa((h'/\kappa)(x'/\sqrt\kappa)+(\partial x'/\sqrt\kappa))$$
\end{comment}

It is not difficult to compute from this and the previous integral form the limit, the fibration and zero-fibre. We feel, however that this second reduction maybe not precisely the right one for our purposes. In the following we try to use a rescaled character to impose the relation $e'/\kappa+1=0$, which makes not difference for finite $\kappa$, but in the limit. The deeper reason for this is probably that we wish to have a reduction compatible with the classical Hamiltonian reduction in then limit, and this requires for $\sl_2$ a connection 

$$\d-\begin{pmatrix} 0 & q(z) \\1 & 0\end{pmatrix}
\qquad\Longleftrightarrow\qquad  
\frac{\partial^2}{\partial z^2}-q(z)$$
\CommentsForMe{$h$-condition? Adapt and compare p=1 and fibration results!!}

which corresponds choosing a subset of fibres with $(h'/\kappa)_n=0$
and $(e'/\kappa)_{n}=\delta_{n,-1}(-1)$, which translates into the cohomological condition $e'=1$. Let us modify the previous result accordingly and add some further information we require

\begin{theorem}
$\HHA^{(2)}[\sl_2,\kappa]$ with the modified BRST differential $\underline{\d}'=b'(z)(e'/\kappa)(z)+b'(z)$ is isomorphic to the N=1 superconformal algebra $\sVir_{N=1}$ with central charge $c=3/2+3(a+2+a^{-1})^2=-\frac{3(1+2\kappa)(5+4\kappa)}{2(3+2\kappa)}$, with generators $G,L'$ given in the proof, times a free fermion with generator 
$$(x'/\sqrt{\kappa})(z)(x'/\sqrt{\kappa})(w)\sim (-e'/\kappa)(z-w)^{-1}=1(z-w)^{-1}$$
where up to coboundary $e'/\kappa+1=0$. 
\end{theorem}
\begin{proof}
Let $\rho$ be the Lie algebra automorphism of $\hat{\osp(1|2)}_\kappa$ defined by
$$\rho(e')=e'/\kappa,\qquad
\rho(h')=h',\qquad
\rho(f')=f'\kappa,\qquad
\rho(x')=x'/\sqrt{\kappa},\qquad
\rho(y')=y'\sqrt{\kappa}
$$
which we see also preserves the three Virasoro elements $L^{\osp(1|2)}, \; L^{\sl_2'},\;L^{C}$. We extend $\rho$ trivial on $b',c'$ to the BRST complex $\bV^\kappa(\osp(1|2))\otimes \langle b',c'\rangle$, then $\rho$ sends the differential to the modified differential $\rho\,\d'\rho^{-1}=\underline{\d}'$. In particular $\underline{\d}'(c')=\rho(e'+1)=e'/\kappa+1$ is a coboundary as intended. In this way we can apply $\rho$ to elements in the kernel of $\d'$, or cochains, which are given in the proof of \cite{CGL20} Lemma 2.11 and obtain cochains for $\underline{\d}'$ and their OPEs:
\begin{itemize}
    \item $x'/\sqrt{\kappa}$ is a cocycle for $\underline{\d}'$, since because $e'_nx'=0,n\geq 0$, and the cohomology class $[x'/\sqrt{\kappa}]$ has with itself the OPE of a free fermion:
    $$[x'/\sqrt{\kappa}](z)[x'/\sqrt{\kappa}](z)=[-e'/\kappa](w)(z-w)^{-1}=(z-w)^{-1}$$
    \item The usual corrected Sugawara Virasoro elements 
    %note the sl2 source http://arxiv.org/abs/hep-th/9308037v1 has X^0=\frac{1}{2} h' and b,c switched.
    \begin{align*}
        \tilde{L}^{\sl_2'}&:=L^{\sl_2'}+\frac{1}{2}\partial h'-:b'\partial c':\\
        \tilde{L}^{\osp(1|2)}&:=L^{\osp(1|2)}+\frac{1}{2}\partial h'-:b'\partial c':\\
        \intertext{remain unchanged and cocycles, as does without correction the coset Virasoro element}
        \tilde{L}^{C}&=L^{\osp(1|2)}-L^{\sl_2'}
        \qquad \to :(y'\sqrt{\kappa})(x'/\sqrt{\kappa}):
    \end{align*}
    \item The Virasoro element for the free fermion changed accordingly to $\frac{1}{2}:(\partial x'/\sqrt{\kappa})(x'/\sqrt{\kappa}):$, as a consequence we have a definition slightly different from \cite{CGL20}
    $$L':=\tilde{L}^{\osp(1|2)}-\frac{1}{2}:(\partial x'/\sqrt{\kappa})(x'/\sqrt{\kappa}):$$ 
    The new $[L']$ again commutes with $[x']$ up to $\underline{\d}'$-coboundary.
    \item The following element is a $\underline{\d}'$-cocycle (the according modification of what is called $\psi$ in \cite{CGL20}) 
    \begin{align*}
    G&:= \frac{\sqrt{-1}}{\sqrt{3+2\kappa}}\left(:h'x':/\sqrt{\kappa}+
    2:e'y':/\sqrt{\kappa}-(1+2\kappa):e'\partial x':/\kappa\sqrt{\kappa}\right)\\
    &=\frac{\sqrt{-1}}{\sqrt{3+2\kappa}}\left(:h'x':/\sqrt{\kappa}-
    2y'\sqrt{\kappa}+\frac{1+2\kappa}{\kappa}\partial x'\sqrt{\kappa}\right)
    \end{align*}
    and its class commutes with $[x]$ and generates with $[L']$ a copy of $\sVir_{N=1}$ with central charge $-\frac{3(1+2\kappa)(5+4\kappa)}{2(3+2\kappa)}$
\end{itemize}
We have in the familiar way a Virasoro element cohomologous to $\tilde{L}^{\sl_2'}$ that is not unchanged under $\rho$ and now reads

%$$\underline{\tilde{L}}^{\sl_2'}
%=\frac{1}{\kappa+2}(\frac{1}{2}h'+bc)^2+(1-\frac{1}{\kappa+2})\partial(\frac{1}{2}h'+bc)+\frac{1}{\kappa+2}f'\kappa$$
$$\tilde{L}^{\sl_2'}_n
=\frac{1}{\kappa+2}\sum_{i+j=n}(\frac{1}{2}h'_i+(bc)_i)(\frac{1}{2}h'_j+(bc)_j)+(1-\frac{1}{\kappa+2})(-1-n)(\frac{1}{2}h'_{n}+(bc)_{n})+\frac{\kappa}{\kappa+2}f'_{n+1}$$
Assume now we are acting on monomials not containing $b',c'$ or $h'$ then $(bc)_i=0$ unless $i<0$. and $h'_i=0$ unless $i\leq 0$
\begin{align*}
    {\tilde{L}}^{\sl_2'}_1
    &=\frac{\kappa}{\kappa+2}f'_{2}\\
    {\tilde{L}}^{\sl_2'}_0
    &=\frac{1}{\kappa+2}\frac{1}{4}(h'_0)^2
    +(1-\frac{1}{\kappa+2})\frac{1}{2}(-1)h'_{0}
    +\frac{\kappa}{\kappa+2}f'_{1}\\
    {\tilde{L}}^{\sl_2'}_{-1}
    &=\frac{1}{\kappa+2}2(\frac{1}{2}h'_{-1}+(bc)_{-1})h'_0
    +\frac{\kappa}{\kappa+2}f'_{0}
\end{align*}
%We note that $y'_{-1}+c_0b_{-1}x'_{-1}......$ is a cochain, since
%$$\underline{\d}y'_{-1}=(b_{-1}(e'/\kappa)_0+...+b_0).y'_{-1}=b_{-1}x'_{-1}/\kappa$$
%$$\underline{\d}(c_0b_{-1}x'_{-1}/\kappa)=(b_{-1}(e'/\kappa)_0+...+b_0)c_0b_{-1}x'_{-1}/\kappa=b_{-1}x'_{-1}/\kappa$$
\end{proof}

We now compute the limit $\kappa\to \infty$ in the integral form generated by $x'/\sqrt{\kappa},y'/\sqrt{\kappa}$. Note that $\underline{d}'(y'/\sqrt{\kappa})=b_{-1}x'_{-1}/\kappa^{3/2}\to 0$ so up to asymptotically vanishing terms also $y'/\sqrt{\kappa}$ is a cocycle. Note also $L_n^{\sl_2'}/\kappa-f'_{n+1}/\kappa\to 0$. For $h'/\kappa=0$ we have
\begin{align*}
    [x'/\sqrt{\kappa}](z)[x'/\sqrt{\kappa}](z)
    &=(z-w)^{-1}\\
    [y'/\sqrt{\kappa}](z)[y'/\sqrt{\kappa}](z)
    &=(L_{n-1}^{\sl_2'}/\kappa)(z-w)^{-1} \\
     [x'/\sqrt{\kappa}](z)[y'/\sqrt{\kappa}](z)
    &=(z-w)^{-2}\\
\end{align*}
which is consistent with $\frac{\partial^2}{\partial z^2}+q(z)$ with $q(z)=\sum_{n\in\Z}L_{n}^{\sl_2'}z^{-2-n}$. The zero-fibre or $L_{n-1}^{\sl_2'}/\kappa=0$ is symplectic fermions 

\begin{align*}
\tilde{L}^{\sl_2'}_{0}x'&=-\frac{1}{2}x' \\
\tilde{L}^{\sl_2'}_{0}y'&=+\frac{1}{2}y' \\
&\\
\tilde{L}^{\sl_2'}_{-1}x'&=y'
\end{align*}
\CommentsForMe{contributions from $\partial x'$ ?}
\CommentsForMe{Hamiltonian reduction means: some nonzero fibre $e'/h$ is again a vertex algebra when restricting to a subalgebra. }

To match this with $\sVir_{N=1}\times \mathcal{F}$, note that $x'/\sqrt{\kappa}$ is the fermion and 
$$y'/\sqrt{\kappa}=-\sqrt{-1/2}G/\sqrt{\kappa}-\partial (x'/\sqrt{\kappa})$$
\CommentsForMe{Check OPE and L0y and match to II}
Note that the OPE cannot be preserved by the action of $L_1$, which is (in hindsight) what one should expect from a Hamiltonian reduction.

\begin{remark}
Note that we could also use the reduction $e'+1=0$, at the price of loosing the classical Hamiltonian reduction picture including the remaining $\sl_2 ^\geq$-action. On the other hand, then the fibre would again be symplectic fermions without any deformations, since $e'/\kappa\to 0$. We have no guess as to if this is interesting to consider.  
\end{remark}

\subsection{The limit of the N=1 superconformal algebra and fermion II}

\renewcommand{\P}{|P\rangle}
\newcommand{\Null}{|0\rangle}
\newcommand{\ferm}{f}
We consider the same case as in the previous section but now follow closely the explicit computation by joint work of the first author in \cite{BBFLT13} Section 3.2.1, which is related to the supersymmetric version of the GKO construction. 

Consider the vertex algebra $\V=\sVir_{N=1}\otimes \mathcal{F}$, where $\sVir_{N=1}$ is the Neveu-Schwarz sector of the super-Virasoro algebra with generators $L_n,G_r$, where $r\in\frac{1}{2}+\Z$ in the Neveu-Schwartz sector, and relations
\begin{align*}
[L_n,L_m]&=(n-m)L_{n+m}+\frac{1}{8}c(n^3-n)\delta_{n+m}\\
\{G_r,G_s\}&=2L_{r+s}+\frac{1}{2}c(r^2-\frac{1}{4})\delta_{r+s}\\
[L_n,G_r]&=(\frac{1}{2}n-r)G_{n+r}
\intertext{and $\mathcal{F}$ an additional fermion ($\ferm_r,r\in\frac{1}{2}+\Z$))}
\{\ferm_r,\ferm_s\}&=\delta_{r+s}\\
[L_n,\ferm_s]&=0\\
\{G_r,\ferm_s\}&=0				
\end{align*}
For the representation theory of $\sVir_{N=1}$ see \cite{IK03}. We consider highest weight modules of $\sVir_{N=1}$, generated by a vector $\P$ with $L_0\P=\Delta_P\P$ for  $\Delta_P=\frac{1}{2}(Q^2/4-P^2)(Q/2-P)$ and
    $$L_n\P=0,\;n> 0,\qquad 
    G_r\P=0,\;r>0,\qquad 
    \ferm_n\P=0,\;r>0$$

Now, similar to the previous chapter, we have two explicit commuting actions of on  $\Vir^{b_1},\Vir^{b_2}$ on $\sVir_{N=1}^{b}\otimes \mathcal{F}$ with 
$$ b_1=\frac{2b}{\sqrt{2-2b^2}},
\quad b_2^{-1}=\frac{2b^{-1}}{\sqrt{2-2b^{-2}}}$$ 
fulfilling the relation $b_1^2+b_2^{-2}=-2$. \\
%In [44] we find 
%\begin{align*}
%	T_k&=\frac{1}{2}\frac{k}{k+3} T_1 + \frac{1}{2}\frac{k+4}{k+3} T_k^S+\frac{\sqrt{(k+2)(k+4)}}{2(k+3)} G f\\
%	T_{k+1}&=\frac{1}{2}\frac{k+6}{k+3} T_1 + \frac{1}{2}\frac{k+2}{k+3} T_k^S-\frac{\sqrt{(k+2)(k+4)}}{2(k+3)} G f\\
%	&T_k+T_{k+1}=T_1+T_k^S\\
%	c&=\frac{3}{2}-\frac{12}{(k+2)(k+4)}=1+6(b+b^{-1})^2....\\
%\end{align*}
%In the Limit $k\to -2$ 
%\begin{align*}
%T_k&=-T_1 + T_k^S+\frac{1}{\sqrt{2}} G f\sqrt{k+2}\\
%T_k(k+2)&=T_k^S(k+2)\\
%T_{k+1}&=2 T_1+\frac{1}{2}T_k^S(k+2)-\frac{1}{\sqrt{2}} G f\sqrt{k+2}\\
%\end{align*}
%and we have $L_0$ eigenvalues $-\frac{1}{2},1$ on $f_{-\frac{1}{2}}$ and $+\frac{1}{2},0$ on $G_{-\frac{1}{2}}$. We note $(T_k)_0=-L_0^{ferm}+L_0^{NS}$. 
%
%IN FEIGIN ARTICLE
\begin{align*}
L_n^{(b_1)}&
=\frac{1}{1-b^{2\hphantom{-}}}L_n^{(b)}
-\frac{1+2b^{2}}{2(1-b^{2\hphantom{-}})}\sum_{r\in\Z+\frac{1}{2}} r :\ferm_{n-r}\ferm_r:
+\frac{b}{1-b^{2\hphantom{-}}}\sum_{r\in\Z+\frac{1}{2}}  :\ferm_{n-r}G_r:\\
L_n^{(b_2)}&
=\frac{1}{1-b^{-2}}L_n^{(b)}
-\frac{1+2b^{-2}}{2(1-b^{-2})}\sum_{r\in\Z+\frac{1}{2}} r :\ferm_{n-r}\ferm_r:
+\frac{b^{-1}}{1-b^{-2}}\sum_{r\in\Z+\frac{1}{2}}  :\ferm_{n-r}G_r:
\end{align*}
We note $L_n^{(b_1)}+L_n^{(b_2)}=L_n^{(b)}+L_n^{\mathcal{F}}$, where $L_n^{\mathcal{F}}=\frac{1}{2}\sum_{r\in\Z+\frac{1}{2}} r :\ferm_{n-r}\ferm_r:$ is the usual Virasoro action for $\mathcal{F}$.

\begin{example}\label{ex_p2}
\allowdisplaybreaks
    We compute $L_n^{(b_1)}$ on the first few elements, the respective expressions for $L_n^{(b_1)}$ are obtained by replacing $b\leftrightarrow b^{-1}$
	\begin{align*}
	L_0^{(b_1)}\P&=\frac{1}{1-b^{2\hphantom{-}}}\Delta_P\P\\
	L_0^{(b_1)}L_{-1}\P&=\frac{1}{1-b^{2\hphantom{-}}}(\Delta_P+1)L_{-1}\P\\
	L_0^{(b_1)}\ferm_{-\frac{1}{2}}\P
	&=-\frac{1+2b^{2}}{2(1-b^{2\hphantom{-}})} (\frac{1}{2}-(-\frac{1}{2})) \ferm_{-\frac{1}{2}}\P
	+\frac{b}{1-b^{2\hphantom{-}}}(-1)G_{-\frac{1}{2}}\P\\
	L_0^{(b_1)}G_{-\frac{1}{2}}\P
	&=\frac{1}{1-b^{2\hphantom{-}}}(\Delta_P+\frac{1}{2})G_{-\frac{1}{2}}\P
	+\frac{b}{1-b^{2\hphantom{-}}}(\frac{1}{2}c0+2\Delta_P)\ferm_{-\frac{1}{2}}\P\\
	%L_0^{(b_1)}G_{-\frac{3}{2}}
	%&=\frac{1}{1-b^{2\hphantom{-}}}\frac{3}{2}G_{-\frac{3}{2}}\\
	%&+\frac{b}{1-b^{2\hphantom{-}}}((\frac{1}{2}c\frac{8}{4}+2\Delta_P)\ferm_{-\frac{3}{2}}\P+2\ferm_{-\frac{1}{2}}L_{-1}\P)\\
	&\\
	L_1^{(b_1)}\P&=0\\
	L_1^{(b_1)}\ferm_{-\frac{1}{2}}\P&=0\\
	L_1^{(b_1)}\ferm_{-\frac{3}{2}}\P&=-\frac{1+2b^2}{2(1-b^2)}(\frac{3}{2}-(-\frac{1}{2}))\ferm_{-1/2}\P+\frac{b}{1-b^2}(-1)G_{-1/2}\P\\
	L_1^{(b_1)}G_{-\frac{1}{2}}\P&=0\\
	L_1^{(b_1)}G_{-3/2}&=\frac{b}{1-b^2}(2\Delta_P+c)\ferm_{-1/2}\P\\ 
	%G_{3/2}G_{-3/2}=2L_0+c1
	&\\
	L_{-1}^{(b_1)}\P
	&=\frac{1}{1-b^{2\hphantom{-}}}L_{-1}\P\\
	&+\frac{b}{1-b^{2\hphantom{-}}}\ferm_{-\frac{1}{2}}G_{-\frac{1}{2}}\P\\
	L_{-1}^{(b_1)}\ferm_{-\frac{1}{2}}\P
	&=\frac{1}{1-b^{2\hphantom{-}}}\ferm_{-\frac{1}{2}}L_{-1}\P\\
	&-\frac{1+2b^{2}}{2(1-b^{2\hphantom{-}})}(\frac{1}{2}-(-\frac{3}{2}))\ferm_{-\frac{3}{2}}\P\\
	&+\frac{b}{1-b^{2\hphantom{-}}}(-1)G_{-\frac{3}{2}}\P \\
	L_{-1}^{(b_1)}G_{-\frac{1}{2}}\P
	&=\frac{1}{1-b^{2\hphantom{-}}}G_{-\frac{1}{2}}L_{-1}\P\\
	&+\frac{b}{1-b^{2\hphantom{-}}}(2\ferm_{-\frac{1}{2}}L_{-1}+\ferm_{-3/2}\Delta_P)\P	
	\end{align*}
\end{example}
\begin{example}
For example, let $\Delta_P=0$, then  highest weight vectors of both $L^{(b_1)},L^{(b_2)}$ are 
$$\Null,\quad G_{-\frac{1}{2}}\Null,\quad \ferm_{-\frac{1}{2}}\Null+\frac{1}{b+b^{-1}}G_{-\frac{1}{2}}\Null$$
 with $L_0^{(b_1)}$-eigenvalues 
 $$0,\qquad \frac{1/2}{1-b^2},\qquad -\frac{1+2b^{2}}{2(1-b^{2})}$$
 and correspondingly $L_0^{b_1}+L_0^{b_2}$-eigenvalues
 $$0,\qquad 
 \frac{1/2}{1-b^2}+\frac{1/2}{1-b^{-2}}=1/2\qquad 
 -\frac{1+2b^{2}}{2(1-b^{2})}-\frac{1+2b^{-2}}{2(1-b^{-2})}
 =
-\frac{1+2b^{2}}{2(1-b^{2})}+\frac{b^2+2}{2(1-b^{2})} 
=1/2
$$
and they are obviously a basis of the $L_{(b_1)}+L_{(b_2)}$-eigenspace with eigenvalue $0$ and $1/2$. They are the highest weight vectors of the defining decomposition of a generic module. In the vacuum module we have $L_{-1}\Null=G_{-1/2}\Null=0$.\\

Note for the following discussion that in the limit 
$$L_0^{(b_1)}:\;0,0,1\qquad
L_0^{(b_1)}: 0,\frac{1}{2},-\frac{1}{2}$$
\end{example}
%We now list decendents that should disappear if we go from Verma to irreducible modules:
%\begin{align*}
%L_{-1}^{b_1}\Null&=...\to L_{-1}/b^2-f_{-1/2}G_{-1/2}/b \\
%L_{-1}^{b_2}\Null&=...\to L_{-1}+f_{-1/2}G_{-1/2}/b \\
%\intertext{before the limit $L_{-1}\Null=f_{-1/2}G_{-1/2}\Null=0$}
%L_{-1}^{b_1}^2+...L_{-2}^{b_1}...&=
%\end{align*}

We now want to consider the limit $b\sim 2b_2\to\infty$. We choose the integral form generated by $\ferm_r,G_r/b$. We first compute the elements in the large central subalgebra generated by $\ell_n=L^{(b_2)}_n/b_2^2$. As in the case $\g,p=1$ we find 
$$L^{(b_2)}/b_2^2=4L_n/b^2$$ 
Hence the fibres $\V^\infty|_{q(z)}$ for $q(z)=\sum_{n\in\Z} \ell_nz^{-n-2}$ are generated by $G_r,\ferm_n$.\\

We observe that in the semiclassical limit of $\sVir_{N=1}$ we have well-defined Lie brackets between $L_{-1},L_0,L_{1}$ and $G_{-1/2},G_{1/2}$ which form Lie superalgebra $\osp(1|2)$. The $\sl_2$-action on the generators, which is well-defined for $\ell_1=\ell_0=\ell_{-1}=0$ can be obtained as limit of the calculations in Example \ref{ex_p2} for the vacuum module, and we find a douplet  

\begin{align*}
    L^{(b_2)}_0 \ferm_{-1/2}
    &=(-\frac{1}{2}) \ferm_{-1/2}\\
    L^{(b_2)}_{-1} \ferm_{-1/2}
    &=-\ferm_{-3/2}+G_{-3/2}/b\\
    L^{(b_2)}_0 (-\ferm_{-3/2}+G_{-3/2}/b)
    &=(+\frac{1}{2})(-\ferm_{-3/2}+G_{-3/2}/b)
\end{align*}
with conformal $1$ with respect to the new Virasoro action $L_0^{(b_1)}$.
\CommentsForMe{check $L^{(b_i)}_0$ for last term}
Let us now abbreviate these doublet elements
\begin{align*}
x&=\ferm_{-1/2},\hspace{2.9cm}
x(z)=\sum_{n\in\Z} \ferm_{n+1/2} z^{-n-1}\\
y&=-\ferm_{-3/2}+G_{-3/2}/b,
\qquad y(z)=\sum_{n\in \Z}\left(n\ferm_{n-1/2}+G_{n-1/2}/2\right)z^{-n-1}
\end{align*}
The relations of the mode algebra are 
\begin{align*}
    \{x_m,x_n\}
    &=\{ \ferm_{m+1/2} , \ferm_{n+1/2} \}
    = \delta_{m+n,-1}\\
    \{x_m,y_n\}
    &=\{ \ferm_{m+1/2} , n\ferm_{n-1/2} \}
    =\delta_{m,-n}n\\
    \{y_m,y_n\}
    &=\{ m\ferm_{m-1/2} , n\ferm_{n-1/2} \}
    +\{ G_{m-1/2}/b , G_{n-1/2}/b \} \\
    &=mn\delta_{m+n,1}+\left(2L_{m+n-1}/b^2+\frac{1}{2}c((m-\frac{1}{2})^2-\frac{1}{4})/b^2\right)\delta_{m+n,1}
    =2\ell_{m+n-1}
\end{align*}\CommentsForMe{This cannot be an automorphism doublet?}
with $c\sim 2b^2$. \CommentsForMe{Factor of $c^{NS}$ strange?}

\newcommand\arxiv[2]      {\href{http://arXiv.org/abs/#1}{#2}}

\end{document}